\documentclass[12pt,a4paper]{article}

\long\def\drop#1{} %{{\footnotesize\par\color{magenta} #1\normalsize\par\EEE}}

\usepackage[margin=25.4mm]{geometry} 

\usepackage{enumitem,etoolbox,todonotes}
\setlist{nolistsep}

\usepackage{amsfonts,mathrsfs}
\DeclareSymbolFont{bboldx}{U}{bboldx}{l}{n}
\DeclareMathSymbol{\bbSigma}{\mathord}{bboldx}{6}

\usepackage[utf8]{inputenc}  
\usepackage{epic,eepic}
\usepackage{xcolor,scrtime,graphicx}
\usepackage{amsmath,amsthm,relsize}
\usepackage{thmtools} % allow for exact copies of theorems

\usepackage{quiver}

\usepackage{bef_alex}
\def\mafo{\mathrm}

\newtheorem{ass}[theorem]{Assumption}
\theoremstyle{definition}
\usepackage{etoolbox}
\newtheorem{RmRemark}[theorem]{Remark}
\AtEndEnvironment{RmRemark}{\null\hfill\qedsymbol}%
\newtheorem{RmExample}[theorem]{Example}
\AtEndEnvironment{RmExample}{\null\hfill\qedsymbol}%
\renewcommand{\ti}{{\times}}
\newcommand{\oti}{{\otimes}}

\usepackage{xspace}
\newcommand\generic{{\smaller GENERIC}\xspace} % version using SMALLER
\newcommand\Generic{G{\smaller ENERIC}\xspace}  % works with all fonts
\newcommand{\sde}{\textsc{SDE}\xspace}  % left over from earlier experiment with fonts
\renewcommand\generic{{\smaller GENERIC}\xspace}
\renewcommand\Generic{{\smaller GENERIC}\xspace}

\usetikzlibrary{calc}
\tikzset{>=latex}

\numberwithin{equation}{section}
\numberwithin{figure}{section}
\usepackage{bm,upgreek}

\usepackage{longtable}  % for the list of notation

\makeatletter
\renewcommand*\env@cases[1][1.2]{%
  \let\@ifnextchar\new@ifnextchar
  \left\lbrace  \def\arraystretch{#1} \array{@{}c@{\ \ }l@{}}}
\makeatother

\usepackage{dsfont}
\def\One{\mathds{1}}

\DeclareMathOperator{\cov}{cov}

\newcommand{\sfZ}{\mathsf Z}

\newcommand{\Dom}{\mathrm{Dom}}
\DeclareMathOperator{\supp}{supp}
\newcommand{\Range}{\mathrm{Range}}
\newcommand{\Nullspace}{\mathrm{Ker}}

\newcommand{\GEN}{{\mathrm{GEN}}}

\DeclareMathOperator{\trace}{tr}
\renewcommand{\div}{\mathop{\mathrm{div}}\nolimits}
\newcommand{\Lebesgue}{{\mathscr L}}

\newcommand{\Expectation}{\bbE}
\DeclareMathOperator{\Prob}{Prob}
\def\longrightharpoonup{\relbar\joinrel\rightharpoonup}
\newcommand{\ProbMeas}{\calP}  % space of probability measures
\DeclareMathOperator{\Span}{span}

\newcommand{\Hmac}{\calH_{\mathcal{A}}} %% {\calH_{\mathrm{mac}}} 
\newcommand{\Jmac}{\bbJ_{\mathcal{A}}} %%{\mathrm {mac}}}
\newcommand{\HB}{\calH_{\mathcal{B}}} %% {\mathrm{HB}}}
\newcommand{\HC}{\calH_{\mathcal{C}}} %% {\calH_{\mathrm{coupl}}}
\newcommand{\Hzw}{\calH_{zw}}  % macroscopic z,w-energy

\newcommand{\JHB}{\bbJ_{\mathcal{B}}} %% {\mathrm {HB}}}
\newcommand{\ip}[2]{\left( #1\middle | #2\right)} % inner product
\newcommand{\dual}[2]{\left\langle #1, #2\right\rangle} % inner product
\newcommand{\Dual}[2]{\big\langle #1, #2\big\rangle} % inner product
\newcommand{\pdelta}{\bm{\updelta}}  % special delta function
\newcommand{\DDsym}{{\bbD_\mafo{sym}}}
\newcommand{\DDanti}{{\bbD_\mafo{skw}}}
\newcommand{\DDantistar}{{\bbD_\mafo{skw}^*}}

\newcommand{\sfcc}{\mathsf{c}}
\newcommand{\extG}{\wh S}

\newcommand{\AAA}{}
\newcommand{\EEE}{}

\usepackage{mathtools}
\DeclarePairedDelimiter{\abs}{\lvert}{\rvert}
\DeclarePairedDelimiter{\norm}{\lVert}{\rVert}
\DeclarePairedDelimiter{\bra}{(}{)}
\DeclarePairedDelimiter{\pra}{[}{]}
\DeclarePairedDelimiter{\skp}{\langle}{\rangle}
\usepackage{ifthen}
\newlength{\leftstackrelawd}
\newlength{\leftstackrelbwd}
\def\leftstackrel#1#2{\settowidth{\leftstackrelawd}%
{${{}^{#1}}$}\settowidth{\leftstackrelbwd}{$#2$}%
\addtolength{\leftstackrelawd}{-\leftstackrelbwd}%
\leavevmode\ifthenelse{\lengthtest{\leftstackrelawd>0pt}}%
{\kern-.5\leftstackrelawd}{}\mathrel{\mathop{#2}\limits^{#1}}}

\usepackage{xstring}
\makeatletter
\def\@font@info#1{%
\IfBeginWith{#1}{No math alphabet change to frozen version}{\relax}{%
\GenericInfo{(Font)\@spaces\@spaces\@spaces\space\space}%
{LaTeX Font Info: \space\space\space#1}%
}%
}%
\makeatother

\catcode`\@=11
\def\raisearound#1#2#3{\mathpalette\raisearoundA{{#1}{#2}{#3}}}
\def\raisearoundA#1#2{\raisearoundB#1#2}
\def\raisearoundB#1#2#3#4{{%
    \setbox0=\hbox{#2}%
    \setbox1=\hbox{#3}%
    \setbox2=\hbox{$\m@th#1#4$}%
    \raise\dimexpr\ht2-\ht0\relax\box0%
    \copy2%
    \raise\dimexpr\ht2-\ht1\relax\box1%
}}
\catcode`\@=12

\DeclareFontFamily{U}{mathx}{}
\DeclareFontShape{U}{mathx}{m}{n}{<-> mathx10}{}
\DeclareSymbolFont{mathx}{U}{mathx}{m}{n}
\DeclareMathAccent{\widecheck}{0}{mathx}{"71}

  \def\sfC{{\mathsf C}}
\def\sfD{{\mathsf D}}  
  
\def\sfJ{{\mathsf J}} \def\sfK{{\mathsf K}} 
  
 \def\sfQ{{\mathsf Q}}

 \def\sfZ{{\mathsf Z}}

\usepackage[colorlinks=true, linkcolor=black, citecolor=black, urlcolor=black]{hyperref}

\usepackage[normalem]{ulem}

\begin{document}

\title{Deriving a GENERIC system from \\
a Hamiltonian system}

\author{Alexander Mielke\thanks{WIAS and Humboldt Universit\"at zu Berlin, Germany}
  \and Mark A. Peletier\thanks{Department of Mathematics and Computer Science, Eindhoven University of Technology, The Netherlands, \texttt{m.a.peletier@tue.nl}}
  \and Johannes Zimmer\thanks{School of Computation, Information, and Technology, Technische Universit\"at M\"unchen, Germany}}

\date{2 June 2025} 
\maketitle

\begin{abstract}
We reconsider the fundamental problem of coarse-graining infinite-dimensional
Hamiltonian dynamics to obtain a macroscopic system which includes dissipative
mechanisms. In particular, we study the thermodynamical implications concerning
Hamiltonians, energy, and entropy and the induced geometric structures such as Poisson
and Onsager brackets (symplectic and dissipative brackets).

We start from a general finite-dimensional Hamiltonian system that is coupled
linearly to an infinite-dimensional heat bath with linear dynamics.  The latter
is assumed to admit a compression to a finite-dimensional dissipative semigroup
(i.e., the heat bath is a dilation of the semigroup) describing the dissipative
evolution of new macroscopic variables.

Already in the finite-energy case (zero-temperature heat bath) we obtain the so-called GENERIC structure
(General Equation for Non-Equilibrium Reversible Irreversible Coupling), with conserved energy, nondecreasing
entropy, a new Poisson structure, and an Onsager operator describing the dissipation. However, their origin is
not obvious at this stage.  After extending the system in a natural way to the case of positive temperature,
giving a heat bath with infinite energy, the compression property leads to an exact multivariate
Ornstein-Uhlenbeck process that drives the rest of the system. Thus, we are able to identify a conserved
energy, an entropy, and an Onsager operator (involving the Green-Kubo formalism) which indeed provide a
GENERIC structure for the macroscopic system.

\providecommand{\keywords}[1]{\par\medskip\noindent\textbf{Keywords and phrases: } #1}
\providecommand{\MSC}[1]{\par\medskip\noindent\textbf{AMS Mathematics Subject Classification: } #1}
\MSC{37K06, 37L05, 37D35, 80A05, 82C35}

\end{abstract}

\tableofcontents

\section{Introduction}
\label{sec:Introduction}

\subsection{Coarse-graining, \generic, and compressions}
\label{sec:coarse-grain-dilat}

We revisit the classical topic of coarse-graining of Hamiltonian systems, and derive a thermodynamically consistent reduction for Hamiltonian systems which linearly couple a finite-dimensional Hamiltonian system to an infinite-dimensional Hamiltonian system acting as heat bath. A crucial element in  the reduction procedure is the theory of compressions. We now describe the main concepts in more detail.

\emph{Coarse-graining} is the process by which less detailed (`coarse-grained') models of physical systems are
derived from more detailed ones (`fine-grained'; see, e.g.,~\cite{WeinanE11,Kuehn15}).  One methodology for
coarse-graining is the Mori-Zwanzig formalism, which allows to give macroscopic descriptions of microscopic
systems such as those arising in molecular dynamics. In particular, this allows one to study the emergence of
macroscopic dissipation from microscopically reversible Hamiltonian systems.  In this paper we study a
classical example of this type; {see~\cite{JaksicPillet98,Zwanzig1973a} and also the
  monograph~\cite{Zwanzig01} {for related treatments}}.

\AAA \textsl{\Generic} describes a class of thermodynamic evolution equations
that is formulated in terms on an energy and an entropy functional and two
geometric structures for the Hamiltonian and the dissipative dynamics.  The
origins of such systems goes back to \cite{Grme84PBFK, Grme85BFDT} and the
metriplectic theory developed in \cite{Morr84BFIC, Morr86PJHD, Morr09TBDO},
whereas the {name \emph{General Equation for Non-Equilibrium Reversible Irreversible
Coupling} and the acronym \Generic} were introduced in~\cite{GrmOtt97DTCF1, OttGrm97DTCF2}.  The \generic
framework provides a \emph{general} way to write down thermodynamically
consistent \emph{equations} for the coupling of two ideal systems, namely
a \emph{reversible} (or Hamiltonian) system and an \emph{irreversible} (or
dissipative) system. \Generic systems form a subclass of all thermodynamically
consistent systems that allow for a better understanding of the interaction
between the two different phenomena of reversible and irreversible dynamics.
In particular, it will help us to see how dissipative systems can be obtained
via coarse graining of a purely Hamiltonian systems. 
\EEE 

If an evolution equation can be put in \generic form, then it is automatically
thermodynamically consistent. In particular, it satisfies conservation of an
energy and monotonicity of an entropy. In \generic, this thermodynamic
consistency is implemented by combining a symplectic operator acting on the
energy and a semidefinite operator acting on the entropy. In addition, the
book~\cite{Otti05BET} describes a Mori-Zwanzig-type derivation of equations of
\generic type from microscopic Hamiltonian systems with a clear separation of
time scales. In this paper we derive a \generic form for the coarse-grained
evolution, and this allows us to give an interpretation for the various
components of the \generic structure.

A key ingredient in this mathematically rigorous derivation is the third component, \emph{compressions}, which can be thought of the reverse of {dilations}. A \emph{dilation} of a semigroup~$\bbS_t$ on a space~$\bfY$
is a unitary group~$\bbU_t$ on a larger Hilbert space $\bfH$ whose projection onto~$\bfY$ coincides with
$\bbS_t$~\cite{NagyFoiasBercoviciKerchy10}. 
Conversely, compressions yield a projection of a unitary group on a large space
  onto a dissipative semigroup on a subspace. It is the existence of a suitable compression subspace that allows us to write the Mori-Zwanzig reduction in Markovian form and relate the dissipative part to the fluctuations.

\bigskip

\subsection{Sketch of the microscopic Hamiltonian system}
\label{s:setting-this-paper}

We first describe the starting point of this paper, which we call the `microscopic' system. Later, from  Section~\ref{s:CG-deterministic} onwards, we will describe how we apply a coarse-graining map to the state space of this system, which will lead to a system that we call `macroscopic'. The terms `microscopic' and `macroscopic' do not refer to actual space or time scales, but are relative to the operation of coarse-graining. The coarse-graining map is highly non-injective, and therefore the macroscopic system has fewer degrees of freedom than the microscopic one. 

The starting point of this paper, the `microscopic system', resembles well-known models in the
literature (see,
e.g.,~\cite{JaksicPillet98,KSTT02LTBL,Reyb06OCS}). It couples
two separate Hamiltonian systems: System~A is finite-dimensional, and System~B
is infinite-dimensional and is interpreted as a heat bath. System~A may be
nonlinear, while System~B and the coupling are linear. This is illustrated in
Figure~\ref{fig:Step1}. The whole system is Hamiltonian, and in the interesting cases, the system also is reversible in time. The evolution is
deterministic, while the initial data for the heat bath may be random, as we discuss below.

\begin{figure}[ht]
  \centering
%  \includegraphics[width=0.1\hsize]{Fig1}
%\qquad
\begin{tikzpicture}[scale=0.65]
\begin{scope}[xshift=-5.0cm] 
\draw[fill= gray!20] (-2,0) rectangle (2,3);
\node[color=blue] at (0,2){System~A};
\node[color=black] at (0,1){$(\bfZ,\calH_\calA, \bbJ_\calA)$};
\end{scope}
\draw[very thick,<->] (-2.4,1.2)--(2.4,1.2);
\node[color=blue] at (0,1.7){coupling}; 
\node[color=black] at (0,0.5){$\calH_\calC= \ip {\sfC z}{\bbP\eta}_\bfH$};

\begin{scope}[xshift=5.0cm] 
\draw[fill= gray!20] (-2.1,0.2) rectangle (2.4,3.2);
\node[color=blue] at (0,2.5){System~B};
\node[color=blue] at (0,1.9){$=$ heat bath};
\node[color=black] at (0,0.9){$(\bfH,\calH_\calB, \JHB)$};
\end{scope}
\end{tikzpicture}
\caption{The microscopic system is a coupling of a fixed Hamiltonian system
  (`System~A') and a heat bath (`System~B'). See
  Section~\ref{s:setting-this-paper} for a discussion. }
\label{fig:Step1}
\end{figure}

\medskip

System~A is described by a triple $(\bfZ,\calH_\calA, \bbJ_\calA)$, where $\bfZ$ is a state space,
$\calH_\calA$ is a Hamiltonian, and $\bbJ_\calA(z)\colon\bfZ \to \bfZ$ denotes the Poisson (also called
co-symplectic) operators. We make the assumption that $\div_z \bbJ_\calA = 0$, corresponding to the Lebesgue measure being invariant (see Remark~\ref{rem:role-of-stationary-measure-in-HamSys}). The evolution of System~A alone is given by the Hamiltonian evolution equation
$\dot z = \bbJ_\calA(z)\rmD\calH_\calA(z)$. 

System~B, the heat bath, is described by a similar triple,
$(\bfH,\calH_\calB, \JHB)$, where $\bfH$ is a Hilbert space, $\calH_\calB$ the Hamiltonian of the heat
bath and $\JHB$ a symplectic evolution operator. The evolution of the heat bath is linear. Specifically,
we take
\begin{equation*}
  \calH_\calB(\eta) = \frac12\|\eta\|_\bfH^2 \quad \text{for all }\eta \in
  \bfH. 
\end{equation*}

The coupling between system~A and system B is described by a Hamiltonian~$\calH_\calC$, which
couples the two systems linearly by
\begin{equation}
  \label{eqdef:HC-pre}
  \calH_\calC(z,\eta)= \ip {\bbA z}{\bbP\eta}_\bfH,
\end{equation}
where $\ip{\cdot}{\cdot}_{\bfH}$ is the inner product in $\bfH$. Here
$\bbA \colon \bfZ \to \bfH$ is a (linear) embedding operator; $\bbP$ is an
orthogonal projection operator that is specified later.  Since $\bbP$ is
self-adjoint, we can alternatively define $\sfC := \bbP \bbA\colon \bfZ\to\bfH$
and write~\eqref{eqdef:HC-pre} as
\begin{equation}
  \label{eqdef:HC}
  \calH_\calC(z,\eta)= \ip {\sfC z}{\eta}_\bfH.
\end{equation}
The operator $\sfC$ satisfies $\sfC = \bbP\sfC$ and $\sfC^* = \sfC^* \bbP$.

{The bilinearity of~\eqref{eqdef:HC} is assumed for simplicity; we expect that similar results can be
  proved also for functionals that are nonlinear in $z$ but linear in $\eta$, e.g., the model studied
  in~\cite{Zwanzig1973a}.}

The
full (coupled) microscopic evolution for $(z,\eta)\in \bfZ\ti\bfH$ is governed
by the Hamiltonian
\begin{equation}
  \label{eqdef:Htotal-intro}
  \calH_{\mathrm{total}}(z, \eta) = \calH_\calA(z) +  \calH_\calB(\eta) +   \calH_\calC(z, \eta),
\end{equation}
together with the joint symplectic operator $\bbJ_\calA \oplus \JHB$.
The corresponding evolution equations are 
\begin{subequations}
  \label{eqdef:Evol-full-intro}
  \begin{align}
     \label{eqdef:Evol-full-1-intro}
    \dot z_t &= \Jmac(z) \big( \rmD \Hmac (z_t) +
             \sfC^* \eta_t\big), \\
    \dot \eta_t &= \JHB\,\big(\eta_t + \sfC z_t\big),
      \label{eqdef:Evol-full-2-intro}	
  \end{align}
\end{subequations}  
and the full Hamiltonian system is denoted by the triple
$(\bfZ\ti \bfH, \calH_{\mathrm{total}}, \bbJ_\calA {\oplus} \JHB)$. {The
aim of this paper is to coarse-grain the coupled system
\eqref{eqdef:Evol-full-intro} to a system where the macroscopic observable
$z\in \bfZ$ is maintained, but the infinite-dimensional degree of freedom
$\eta \in \bfH$ is reduced to a finite-dimensional coupling variable $w$. The
so-called compression property (see Assumption \ref{ass:dilation} below) will
imply that despite the loss of information the resulting system in
$(z,w)\in \bfZ\ti \bfW$ still is Markovian.}

\begin{RmExample}[Running example, part 1]
\label{ex:RunningExa1}
To facilitate the understanding we provide a more explicit example for the
whole theory that explains all the constructions in a concrete case.
  
For System A we consider a canonical Hamiltonian system with
\[
  z=(q,p)\in \R^n\ti \R^n =\bfZ, \quad \calH_\calA(q,p)= \frac12|p|^2 + V(q),
  \quad \text{and} \quad \bbJ_\calA=\left(\begin{smallmatrix}0&\bbI_{n\ti
        n}\\{-}\bbI_{n\ti n}&0
  \end{smallmatrix}\right)  .
\]
For System B we choose
\[
  \bfH= \rmL^2(\R; \R^3), \quad \calH_\calB(\eta)=\frac12\|\eta\|^2_{\rmL^2},
  \quad \text{and} \quad \JHB  = \left(\begin{smallmatrix} -\pl_x
      &0&0\\ 0&-\pl_x&0\\ 0&0&-\pl_x
  \end{smallmatrix}\right) .
\]
In particular, the unitary group $(\ee^{t\JHB})_{t\in \R}$ is the group of
shifts $\bbS_t f := f(\,\cdot \,{-}t)$.  \AAA For simplicity {we assume in this
example} that coupling to the heat bath occurs only via \EEE the
state variable $q$ (but not the momentum $p$), i.e.\
$\calH_\calC(q,p,\eta)= \ip{\sfK q}{\eta}$ for a linear bounded operator
$\sfK\colon\R^n \to \bfH$.

In summary, the example system has state $(q,p,\eta) \in \R^n\ti
\R^n \ti \rmL^2(\R;\R^3)$ and is a coupled system of a nonlinear ODE
and a linear transport equation for the three components of $\eta$:
\[
\frac{\rmd}{\rmd t} \bma{c} q\\ p \\ \eta\ema = \bma{ccc} 0&\bbI_{n\times n}\!\!
 &0\\ \!\!{-}\bbI_{n\times n}\!\!&0&0\\
0&0& \!\!{-}\pl_x\! \ema \bma{c} \rmD V(q) + \rmD_q\calH_\calC(q,\eta) \\ p \\ 
\eta + \rmD_\eta \calH_\calC(q,\eta)   \ema 
= \bma{c} p \\ \!\! -\rmD V(q) - \sfK^* \eta \!\! 
 \\[0.3em]
 - \pl_x( \eta + \sfK q) \ema .
\]
The first two equations are ODEs in $\R^n$, whereas the third equation is a 
transport equation in $\rmL^2(\R;\R^3)$ that has to be interpreted in a
suitable weak or mild form, 
if we want to treat general initial conditions (see the abstract construction
in Theorem~\ref{t:ex-un-full}).
\end{RmExample}

\subsection{The compression property}
\label{ss:compression-assumption}
We next choose the projection operator $\bbP$ and its range $\bfW$. The space $\bfW$ needs to be large enough to contain the range of $\sfC$, but may also be larger. The central requirement that determines the choice of $\bfW$ and the projection $\bbP$ onto $\bfW$ is the following `compression' property.

\begin{restatable}[Compression property of the heat bath]{ass}
{restatableCompressionAssumption} 
  \label{ass:dilation}
  There exists a linear operator~$\bbD$ on $\bfW$ such that 
\begin{equation}
\label{eqass:dilation}
\bbP \ee^{t\JHB} \big|_\bfW   = \begin{cases}
  \ee^{-t\bbD}  & t \geq 0\\
  \ee^{t\bbD^*} &   t \leq 0 
\end{cases}
\qquad\text{on }\bfW.
\end{equation}
The eigenvalues of {$\bbD + \bbD^*$} are strictly positive.
\end{restatable}

\noindent
The left-hand side in the identity~\eqref{eqass:dilation} should be interpreted
as follows.  If we evolve the heat bath starting from $w\in\bfW$, the result
$\ee^{t\JHB}w$ typically is not an element of~$\bfW$. Then~$\bbP$ on the
left-hand side projects $\ee^{t\JHB}w$ back onto $\bfW$, as one does in
the Mori-Zwanzig framework. In particular, with the compression property the
flow map $w\mapsto \bbP\ee^{t\JHB}w$ is equivalent to the memory-less
evolution $\dot w = -\bbD w$ for $t\geq0$.

Assumption~\ref{ass:dilation} in fact implies that $\bbP \ee^{t\JHB}$ is a
strictly dissipative semigroup, since
\[
\norm[\big]{\ee^{-t\bbD}w}_\bfH \leq \ee^{-\alpha t} \|w\|_\bfH 
\qquad \text{for all $t\geq0$ and $w\in \bfW$},
\]
where $\alpha>0$ is the smallest eigenvalue of $(\bbD+\bbD^*)/2$.

\begin{RmRemark}\label{re:Dilation*}
  In Assumption~\ref{ass:dilation}, it is equivalent to assume only the
  behavior for $t\geq0$: the unitary nature of $\ee^{t\JHB}$ then implies that
  $ \bbP\ee^{t\JHB}\bbP = \bbP \bra*{\ee^{-t\JHB}}^* \bbP = \bra*{\bbP
    \ee^{-t\JHB}\bbP}^* = \ee^{t\bbD^*}\bbP$ for $t\leq0 $.
\end{RmRemark}

The context of the compression property is the theory of (dilations and)
compressions. The relevant background is summarised in
Appendix~\ref{sec:Appendix-Dilations}. In essence, the theory of dilations
states that for any contraction semigroup $(C_t)_{t\geq0}$ one can find a
unitary group $(S_t)_{t\in \R}$ on a sufficiently larger Hilbert space such
that $C_t$ is a `coarse-grained version' of $S_t$.  Conversely, one can in this
setting regard the dissipative semigroup as a compression of a given unitary
group.  In our context, the unitary group is $\ee^{t\JHB}$ and represents the
linear Hamiltonian system and $\ee^{-t\bbD}$ is the contraction semigroup. We
interpret a compression as a coarse-graining of a unitary group.

The crucial implication of the Compression Property~\ref{ass:dilation} is that
the evolution $\ee^{-t\bbD}$ on~$\bfW$ is a semigroup. It is thus in particular
Markovian in the sense that the evolution equation does not involve the
past. In~\eqref{eqdef:Evol-full-2}, we give a representation of the heat bath
variable $\eta$ which involves a memory term. The compression property allows
us to rewrite that memory term as a local term.

For the coupled system, the projection onto $\bfW$, i.e., the part of the heat
bath which interacts with System A, yields a compression semigroup. The
existence of a compression, and with it the associated compression subspace
$\bfW$ is related to the theory of Markovian subspaces in linear systems
theory, see~\cite[Ch.~10]{Lindquist2015a}. Necessary and sufficient conditions
for the existence of compression subspaces are given
in~\cite[Th.~5.4]{Picci1986a}. We discuss the Compression
Property~\ref{ass:dilation} in more detail in Section~\ref{s:dilations}.

{
\begin{remark}[The compression property replaces ergodicity]
  In finite-dimensional Hamiltonian systems, dissipative behavior is often associated with ergodic
  properties. In the context of the infinite-dimensional Hamiltonian heat bath of this paper, the Compression
  Property~\ref{ass:dilation} mimics the property of ergodic systems. In the rest of this paper, we indeed
  show that the resulting system has thermodynamically consistent dissipative behavior.
\end{remark}}

\begin{RmExample}[Running example, part 2]
\label{ex:RunningExa2}
Example \ref{ex:RunningExa1} was chosen such that the unitary shift group
$\bbS_t= \ee^{t\JHB}$ admits an explicit and simple compression. We
choose a linear operator $\bbD$ on $\R^3$ with eigenvalues $\vartheta_1>0$ and
$\vartheta_2\pm \rmi \varsigma$ with $\vartheta_2,\varsigma >0$, and define the
functions $f_j\colon\R\to \R^3$ by
\begin{align*}
\text{$\ds  f_1(y) =  \frac{\chi_{_{\smaller\leq}}\!(y)\ee^{\vartheta_1
      y}}{\ts\sqrt{2\vartheta_1}} \! \bma{@{}c@{}} 1\\ 0 \\ 0\ema\!, \
  f_2(y) =   \frac{\chi_{_{\smaller\leq}}\!(y)\ee^{\vartheta_2
      y}}{\sqrt{2\vartheta_2}} \!
  \bma{@{}c@{}} 0\\ \cos(\varsigma y) \\ \sin(\varsigma y)\ema \!,
  \
  f_3(y) =  \frac{\chi_{_{\smaller\leq}}\!(y) \ee^{\vartheta_2 y}}{\sqrt{2\vartheta_2}} 
  \! \bma{@{}c@{}} 0\\ \!{-} \sin(\varsigma y)\! \\ \cos(\varsigma y)\ema \! ,
$}
\end{align*}
where $\chi_{_{\smaller\leq}}= \bm1_{{]-\infty,0]}}$ is the indicator function
for the negative part of the real line.  We then define $\bfW$ to be the
three-dimensional subspace $\bfW=\mathrm{span}\big\{ f_1, f_2, f_3\big\}$, and
the operator $\bbP$ to be the orthogonal projection onto $\bfW$.  Note that
$\{f_1,f_2,f_3\}$ forms an orthonormal basis of $\bfW\subset \bfH$.
  
The validity of the compression property can roughly be seen as follows: 
For $t>0$ the functions $\bbS_t f_j=f_j(\,\cdot\,{-}t)$ immediately leave
$\bfW$, but they return to $\bfW$, if we cut them off at $y=0$ by multiplying 
with $\chi_{_{\smaller\leq}}$, i.e.\ $\chi_{_{\smaller\leq}} \bbS_t f_j \in
\bfW$. Moreover, we can express
this action by the contraction semigroup $\ee^{-t\bbD}$, see the calculations in
the proof of Proposition~\ref{pr:ExplicitDilation}.
  
Using the orthonormal basis $\{f_1,f_2,f_3\}$ we can identify $\bfW$ with $\R^3$
and using the orthogonal projections $\bbP$ from $\bfH$ to $\bfW$ we find
\begin{align*}
  &\bbD= \begin{pmatrix} \vartheta_1& 0&0\\ 0&\vartheta_2&-\varsigma \\
    0&\varsigma&\vartheta_2 \end{pmatrix} \ \text{ and } 
\\
& \bbP \ee^{t
    \JHB} \big|_\bfW = \ee^{-t \bbD} = \begin{psmallmatrix}
    \ee^{-\vartheta_1 t}&0&0 \\ 0& \ee^{-\vartheta_2t}\cos(\varsigma t)&
    \ee^{-\vartheta_2t}\sin(\varsigma t)\\ 0&-\ee^{-\vartheta_2t}\sin(\varsigma
    t)& \ee^{-\vartheta_2t}\cos(\varsigma t)
  \end{psmallmatrix} \text{ for } t \geq 0.
\end{align*}

Having the explicit representation of $\bfW$, the coupling
$\sfK\colon\R^n \to \bfW\subset \bfH$ can be specified as
$\sfK q = (\sfcc^1{\cdot}q)f_1+ (\sfcc^2{\cdot}q)f_2+ (\sfcc^3{\cdot}q)f_3 $,
where $\sfcc^1, \sfcc^2, \sfcc^3\in \R^n$ are given coupling vectors.  Hence,
the coupling energy reads
\[
  \calH_\calC(q,p,\eta)= \ip{\sfK q}{\bbP\eta}_\bfH = \sum_{j=1}^3
  (\sfcc^j{\cdot} q ) \int_{-\infty}^0 \!\! f_j(x){\cdot} \eta(x) \dd x,
\]
and the coupled ODE-PDE system takes the explicit form 
\[
\frac{\rmd}{\rmd t} \bma{c} q\\ p \\[0.3em] \eta\ema 
= \bma{c} p \\ \!\! -\rmD V(q) - \sum_1^3 \int_{-\infty}^0f_j(y){\cdot} \eta(y)
\dd y \; \sfcc^j 
 \\[0.3em]
 - \pl_x\big( \eta + \sum_1^3 (\sfcc^j{\cdot} q) f_j(\cdot)\big) \ema.
\]
We see that the coupling between the ODEs (first two equations) and the linear
transport equation (third equation) occurs through the observable 
functions $f_j \in \bfW$, which are exponentially decaying. 

As mentioned above, the transport equation has to be interpreted in a suitable
weak or mild form if we want to treat general initial conditions. Note also
that the functions $f_j$ do not lie in the domain of the generator $\JHB$ of
the unitary group $\ee^{t\JHB}$, because $\mathrm{dom}(\JHB)=\rmH^1(\R;\R^3)$
and hence $\bfW\cap \mathrm{dom}(\JHB)= \{0\}$.
\end{RmExample}

\subsection{Coarse-graining, part 1: Deterministic initial data}
\label{s:CG-deterministic}

In Theorem~\ref{th:recognize-as-GENERIC}, we rigorously derive an equation for
a coarse-grained version of the microscopic system, and show how it can be
recognized as a \Generic equation.  \Generic equations are evolution equations
for an unknown $y(t) \in \bfY$ that can be written in terms of two functionals
$\calE$ and $\calS$ and two operators~$\bbJ$ and~$\bbK$ as
\begin{equation}
  \label{eq.generic-intro}
  \dot y = \bbJ(y)\rmD\calE(y) + \bbK(y)\rmD\calS(y).
\end{equation}
The quintuple $(\bfY,\calE,\calS,\bbJ,\bbK)$ is used to denote the \generic
system.  Here $\calE$ and $\calS$ are referred to as the energy and entropy
functionals, and $\rmD$ is an appropriate concept of derivative (see
Remark~\ref{re:notation-derivative-D}). 
The operators $\bbJ$ and $\bbK$ are called the \emph{Poisson} and \emph{Onsager} operators, and they define a Poisson bracket and a dissipative bracket. 
In the \Generic framework, the
components $\calE$, $\calS$, $\bbJ$, and $\bbK$ are required to satisfy a
number of additional conditions, which in particular imply that along any
solution of~\eqref{eq.generic-intro} the functional $\calE$ is constant and the
functional $\calS$ is non-decreasing. We describe the \Generic framework in
more detail in Section~\ref{s:GENERIC}.

To describe this result precisely we first define the coarse-graining map $\pi_{z,w,e}$, 
\begin{equation}
  \label{eqdef:projection-zwe}
\bfZ\ti\bfH \ni (z,\eta) \ \stackrel{\pi_{z,w,e}}\longmapsto\  
\bra*{\,z,\,w:=\bbP \eta,\, e := \tfrac12\|\eta\|_\bfH^2 - 
\tfrac12\|w\|_\bfH^2 }\in \bfZ\ti\bfW\ti\R.
\end{equation}
\AAA Thus, we apply the coarse-graining map to the coupled system
$\big(\bfZ\ti\bfH, \calH_{\mathrm{total}},\Jmac{\oplus}\JHB\big)$, which is an
infinite-dimensional Hamiltonian system, and obtain a coarse-grained system on
the reduced state space $\bfY := \bfZ\ti\bfW\ti\R$. This is done by
projecting out infinitely many degrees of freedom in $\bfH$ and keeping only
the new observables $w = \bbP \eta$ in the finite-dimensional subspace 
$\bfW $ of $\bfH$. Note
that System A is considered already as a `coarse-grained' system, which is 
preserved by $\pi_{z,w,e}$, and coupled to the new observable $w\in \bfW$. \EEE 
The third variable $e$ allows us to recover
an important piece of information from the missing part
$\eta-w = \eta -\bbP\eta$, namely its total energy
$\tfrac12\|\eta-w\|_\bfH^2 = \tfrac12 \|\eta\|_\bfH^2 -\tfrac12 \|w\|_\bfH^2$.
The variable $e$ is only used to keep track of the energy exchanged with the
heat bath. Already for $(z,w)=(z,\bbP\eta)$ we obtain the closed system of ODEs
\begin{equation}
  \label{eq:ClosedODE.z.w}
\dot z = \bbJ_\calA(z)\big( \rmD \calH_\calA(z)+ \sfC^*w), \qquad \dot w=-\bbD
w + \sfC z,
\end{equation}
see Figure \ref{fig:DetIniDat}. However, this reduced model does not have a
thermodynamical structure. 
\begin{figure}[ht]
  \centering

\begin{tikzpicture}[scale=0.65]

\node[color=blue] at (-8.5,6.99) {\bfseries\itshape macro};
\node[color=blue] at (-8.5,2.3) {\bfseries\itshape micro};

\begin{scope}[xshift=0cm, yshift=5cm]

\draw[fill=blue!10] (-7,0) rectangle (7.2,4);
\node[right] at (-6.7,3.2) {\small System of ODEs for $(z,w)\in \bfZ\ti \bfW$}; 
\begin{scope}[xshift=-4.5cm, yshift=0cm] 
\draw[fill= gray!20] (-2,0.5) rectangle (2,2.5);
\node[color=black] at (0,1.5){$(\bfZ,\calH_\calA, \bbJ_\calA)$};
\end{scope}
\begin{scope}[xshift=-0.3cm, yshift=0cm] 
\draw[very thick,<->] (-1.5,1.5)--(1.5,1.5);
\node[color=blue] at (0,2){coupling}; 
\node[color=black] at (0,0.8){$ \ip {\sfC z}{w}_\bfW$};
\end{scope}
\begin{scope}[xshift=4.2cm, yshift=0cm] 
\draw[fill= gray!20] (-2.5,0.5) rectangle (2.5,2.5);
\node[color=black] at (0,2){ $w_t  = \bbP\eta_t \in \bfW$}; 
\node[color=black] at (0,1){ $\dot w_t=-\bbD w_t {+}\sfC z_t$};
\end{scope}
\end{scope}

\draw[fill=blue!10] (-7,0.0) rectangle (7.2,4);
\node[right] at (-6.7,3.2) {\small Hamiltonian system for $(z,\eta)\in \bfZ\ti \bfH$};
 
\begin{scope}[xshift=-4.5cm] 
\draw[fill= gray!20] (-2,0.5) rectangle (2,2.5);
\node[color=black] at (0,1.5){$(\bfZ,\calH_\calA, \bbJ_\calA)$};
\end{scope}
\draw[very thick,<->] (-2,1.5)--(2,1.5);
\node[color=blue] at (0,2){coupling}; 
\node[color=black] at (0,0.8){$ \ip {\sfC z}{\bbP\eta}_\bfH$};

\begin{scope}[xshift=4.5cm, yshift=0cm] 
\draw[fill= gray!20] (-2.0,0.5) rectangle (2.2,2.5);
\node[color=black] at (0,1.5){$(\bfH,\calH_\calB, \JHB)$};
\end{scope}

\draw[ultra thick, color=blue!70!red, ->] (3.8,2) to [out=20, in=-50] 
     %  .. controls (8,3) .. % 
  node[pos=0.54, right]{$\bm\bbP$}(6.4,7.0);
\end{tikzpicture}

\caption{The special coupling and the compression property lead to the closed
  deterministic systems of ODEs \eqref{eq:ClosedODE.z.w} for $(z,w)\in \bfZ\ti
  \bfW$, where $w$ is the observable part $w = \bbP\eta$. }
\label{fig:DetIniDat}
\end{figure}

Using the additional scalar energy variable $e$, Part~\ref{part-1:t:main-coarse-graining-result} of
Theorem~\ref{th:recognize-as-GENERIC} states the following. If $(z,\eta)$ evolves as described in
Section~\ref{s:setting-this-paper}, starting from a given point $z(t=0) = z_0\in \bfZ$ and
$\eta(t=0) = w_0\in \bfW$, then $y(t) := \bra[\big]{z(t),w(t),e(t)} := \pi_{z,w,e}\bra[\big]{z(t),\eta(t)}$ is
an exact solution of the ordinary differential equation
\begin{multline}
  \label{eq:formulation-as-GENERIC-intro}
\frac{\rmd}{\rmd t} \begin{pmatrix} z\\w\\e\end{pmatrix}
= 
\underbrace{\bma{ccc} \Jmac  (z)&0&0\\0 &-\DDanti&0 \\0 &0&0\ema}_{\bbJ_\GEN(z,w,e)}
\underbrace{\bma{c} \rmD \Hmac(z) + \sfC^* w \\ w +\sfC z \\ 1 \ema}_{\rmD
  \calE_\GEN(z,w,e)}\\ 
{}+ \underbrace{\frac1\beta \bma{ccc} 
0&0&0\\ 
0&\DDsym & -\DDsym(w{+}\sfC z) \\
0&-{\ip{\DDsym(w{+}\sfC z)}{ \Box}_\bfH} &\ip{\DDsym(w{+}\sfC z)}{w{+} \sfC
  z}_\bfH \ema}_{\bbK_\GEN(z,w,e)} 
\underbrace{\bma{c} 0 \\ 0 \\ \beta \ema}_{\rmD \calS_\GEN(z,w,e)}.
\end{multline}
Here the operator $\bbD$ is decomposed into symmetric and skew-symmetric parts,
\begin{equation}
  \label{eq:D-decom}
    \bbD= \DDsym+\DDanti \quad \text{with } \DDsym=\frac12(\bbD{+}\bbD^*) \ \text{
      and } \ \DDanti:=\frac12(\bbD{-}\bbD^*),
\end{equation}
and we have defined
\begin{align*}
  \calE_\GEN(z,w,e) &:= \calH_\calA(z) + \ip{\sfC z}{w}_\bfH+  \frac12\|w\|_\bfH^2 + e ,   \\
  \calS_\GEN(z,w,e) &:= \beta e.
  \end{align*}
Here $\beta>0$ is a parameter that will be discussed in Section~\ref{ss:CG-part2-intro}. 

\medskip

The result~\eqref{eq:formulation-as-GENERIC-intro} of Theorem~\ref{th:recognize-as-GENERIC} is interesting
for a number of reasons.  As we already noted in Section~\ref{ss:compression-assumption}, it is remarkable
that coarse-graining of this type leads to a closed system: typically, if in a dynamical system one disregards
a part of the state space, then unique solvability tends to break down. In this case the Compression
Property~\ref{ass:dilation} guarantees that although knowledge of the initial values $z(0)$ and $w(0) =\bbP\eta(0)$ is not enough to uniquely solve for the microscopic system
$(z(t),\eta(t))$, it \emph{does} suffice for the prediction of the \emph{projected} version of
$(z(t),\eta(t))$. This points to a special role of the compression property in the study of this type of
systems.  The state space of the infinite-dimensional Hamiltonian system contains a finite-dimensional
manifold that is exactly invariant under the full flow, and reduction of the flow to the manifold is described
by a \Generic\ equation.

Secondly, by following how the microscopic Hamiltonian-system components morph
through coarse-graining into the \Generic
equation~\eqref{eq:formulation-as-GENERIC-intro} we can gain some understanding
about the \Generic framework itself. In
Sections~\ref{ss:GENERIC-interpretation-EJK}, \ref{ss:explanation-S-part1},
and~\ref{ss:explanation-S-part2}, we use the results of this paper to shed
light on the modeling origin of the four components $\calE$, $\calS$, $\bbJ$,
and $\bbK$ and the relations between them.

\subsection{Coarse-graining, part 2: Heat bath with positive temperature}
\label{ss:CG-part1.5-intro}

In Section \ref{s:CG-deterministic} the elements $\eta \in \bfH$ still have finite
energy $\frac12\|\eta\|^2$, and now we want to extend the theory to a heat bath
with positive temperature $1/\beta$, which also implies infinite energy. We first only look to
the interplay of the positive-temperature heat bath and the compression
property. This will explain basic features of our theory without too much
technicalities, because everything can be explained by Gaussian measures and
processes. In Section \ref{ss:CG-part2-intro} we will then do the same
extension on the full Hamiltonian system $(\bfZ\ti \bfH,\calH_\mafo{total},
\bbJ_\calA {\otimes} \JHB)$.   
 
For creating a Gaussian equilibrium measure corresponding to an energy
distribution ``$\ee^{-\beta\|\eta\|_\bfH^2/2}\dd \eta$'' one embeds $\bfH$ into
a bigger Hilbert space $\bfX$ by a symmetric and positive definite trace-class
operator $\bbC:\bfX\to \bfX$ such that $\bfH=\bbC^{1/2}\bfX$ and
$\|\bbC^{1/2} \eta\|_\bfH=\|\eta\|_\bfX$. With this a Gaussian measure
$\gamma_\beta$ is defined such that the characteristic function
$\chi= \calF \gamma_\beta$ takes the form
$\chi(\xi) = \exp\big(-\frac1{2\beta}\|\bbC^{1/2} \xi\|_\bfX^2\big) $. By the
general theory of Gaussian measures, it follows that the orthogonal projector
$\bbP: \bfH \to \bfW\subset \bfH$ extends to a random variable $\wh\bbP$ on the
space $\rmL^2(\bfX,\gamma_\beta)$, see
Section~\ref{ss:Microscopic-positive-temperature}. We show that
it is also possible to choose $\bfX$ and $\bbC$ in such a way that the unitary
group $(\ee^{t\Jmac})_{t\in\R}$ on $\bfH$ extends by density to a strongly
continuous group $(\extG_t)_{t\in \R}$ on $\bfX$ (see Section
\ref{ss:extens-chosen-contin}). Thus, for each (infinite-energy) state
$\eta\in \bfX$ of the heat bath, we can define the process
\[
t \mapsto Y_t(\eta) := \wh\bbP \extG_t \eta \ \in \bfW.
\]
It turns out that the process $Y$ is completely independent of the choice of
$\bfX$ and $\bbC$ and is uniquely determined by the compression property
alone. Indeed, it is a stationary multivariate Ornstein-Uhlenbeck process with
mean $0\in \bfW$ and covariance matrix 
\[
\bbE\big( Y_t\oti Y_s\big) = \begin{cases} \frac1\beta\, \ee^{-(t-s)\bbD} &
                       \text{for }t\geq s, \\[0.2em] 
               \frac1\beta \,\ee^{-(s-t)\bbD^*} & \text{for } t\leq s. 
       \end{cases}
\]
see Lemma \ref{l:Y-is-OU}. Hence, the paths $Y_t(\eta)$ are the solutions of
the linear SDE
\[
 \rmd Y_t = - \bbD Y_t \,\rmd t + \Sigma\,\rmd B_t \quad \text{with } \Sigma=
 \big(\tfrac{\ds 1}{\ds\beta}\,(\bbD{+}\bbD^*) \big)^{1/2} .
\]
See Figure \ref{fig:HeatBath.to.OU} for an explanation how the drift term
and the noise term arise via the projection $\wh\bbP:\bfX\to \bfW$ and the
extended semigroup $\extG_t:\bfX\to \bfX$.
\begin{figure}[ht]
  \centering
\begin{tikzpicture}[scale=0.65]

\node[color=blue] at (-8.5,8) {\bfseries\itshape macro};
\node[color=blue] at (-8.5,2.3) {\bfseries\itshape micro};

\begin{scope}[xshift=0cm,yshift=6.2cm]
\draw[fill=blue!10] (-7,0) rectangle (8.7,4);
\node[right] at (-6.7,3.3) {\small Stochastic process for $Y_t=\wh\bbP \extG_t
  \eta_0 \in \bfW$}; 

\begin{scope}[xshift=-0.7cm, yshift=0cm] 
\draw[fill= gray!20] (-0.1,0.3) rectangle (9.2,2.6);
\node[color=black, right] at (-6.1,2.1){{\smaller $\bfW$-valued}};
\node[color=black, right] at (-6.1,1.5){{\smaller Ornstein-Uhlenbeck}};
\node[color=black, right] at (-6.1,0.9){{\smaller  process}}; 
\node[color=black, right] at (0,2){ $\rmd Y_t = - \bbD Y_t \dd t + \Sigma
  \rmd B_t $}; 
\node[color=black, right] at (0,1){ $Y_t = \ee^{-t\bbD} \wh\bbP \eta_0 +
  \int_0^t \ee^{-(t-s) \bbD} \Sigma \dd B_s$};
\end{scope}
\end{scope}

\draw[fill=blue!10] (-7,-0.8) rectangle (8,5);
\node[right] at (-6.7,4.2) {\small Linear Hamiltonian systems};
 
\node at(-3.8,2.6) {{\smaller positive temperature $1/\beta$}};
\node at(-3.8,0.4) {{\smaller zero temperature}};
\begin{scope}[xshift=1.4cm] 
\draw[fill= gray!20] (-2.0,-0.3) rectangle (6,3.3);
\node[color=black] at (0,2.6){$(\bfX,\HB^\infty, \gamma_\beta)$};
\node[color=black] at (4.2,2.6){$\eta_t = \extG_t \eta_0$};
\draw[thick,->] (0,1)--(0,2);
\node[color=black,right] at (0,1.5){extension};
\node[color=black] at (0,0.4){$(\bfH,\calH_\calB, \JHB)$};
\node[color=black] at (4.2,0.4){$\eta_t = \ee^{t\JHB} \eta_0$};
\end{scope}

\draw[ultra thick, color=blue!80!red,->] (5.1,3) -- 
        node[pos=0.7, left]{$\ee^{-t\bbD} = \wh\bbP \extG_t|_\bfW$} (1.8,6.7);
\draw[ultra thick, color=blue!80!red,->] (5.3,3) -- 
        node[pos=0.7, right]{$\wh\bbP \extG_t(I{-}\wh\bbP)\eta_0$\!\!} (6.1,6.7);
\end{tikzpicture}

\caption{The Gaussian measure $\gamma_\beta$ encodes the temperature
  $1/\beta>0$. After extending $\bbP:\bfH \to \bfW$ and $\ee^{t\JHB}$ to
  $\wh\bbP$ and $\extG_t$, respectively, the compression property implies that
  $t\mapsto Y_t = \wh\bbP \extG_t\eta_0$ defines a Gaussian process, namely the
  $\bfW$-valued Ornstein-Uhlenbeck process with drift $\bbD Y_t$ and noise
  intensity $\Sigma$ with $\Sigma\Sigma^*=\frac1\beta(\bbD{+}\bbD^*)$. }
\label{fig:HeatBath.to.OU}
\end{figure}

In the following we discuss how this construction can be coupled to the
Hamiltonian System~A.

\subsection{Coarse-graining, part 3: Random initial data}
\label{ss:CG-part2-intro}

In Section \ref{s:CG-deterministic} we already derived a \generic
structure for the case of a heat bath with zero temperature. But the result was ad
hoc and didn't explain why this is the ``correct form'' of $\calE_\GEN$,
$\calS_\GEN$, $\bbJ_\GEN$, and $\bbK_\GEN$. This question can only be answered
by looking at the coupled system at positive temperature $1/\beta>0$. This
means that we consider now initial conditions $(z_0,\eta_0)\in \bfZ\ti \bfX$
that are random according to a suitable equilibrium measure.
 
This approach will reveal the deeper connections to the coarse-graining of
systems with thermal fluctuations. For instance, it allows us to give a meaning
to the parameter $\beta>0$ in~\eqref{eq:formulation-as-GENERIC-intro} as an
inverse temperature and to `entropy' $\calS_\GEN(z,w,e) = \beta e$ as a
micro-canonical surface area.  In addition, it addresses the slightly unnatural
initial datum of the previous section. There we started the heat bath at
$\eta(t=0) =w_0\in \bfW$, which fixes the visible part of the bath to some
non-zero~$w_0$, while requiring the (infinite-dimensional) `invisible'
component to be zero. This is an unnatural choice in a Hamiltonian system with a
Hamiltonian $\calH_\calB(\eta) = \frac12\|\eta\|_\bfH^2$ that gives no special
status to $\bfW$, and our choice of random initial data will also improve this
aspect.\bigskip

We choose a distribution for the random initial data that builds on the
`canonical' invariant measure
$\mu_\beta := ``{\ee^{-\beta \calH_{\mathrm{total}}(z,\eta)}}\text{''}$ of the
microscopic Hamiltonian system, where $\calH_{\mathrm{total}}$ is the total
Hamiltonian given in~\eqref{eqdef:Htotal-intro}. Like the Gaussian measure
$\gamma_\beta$ is defined on $\bfX \supset \bfH$, also $\mu_\beta$ can be
rigorously defined on $\bfZ\ti \bfX$ but is no longer Gaussian unless
$\calH_\calA$ is also quadratic. For any $(z_0,w_0)\in \bfZ\ti\bfW$
we then choose a random initial datum $(z_0,\eta_0)$ from the conditional
measure $\mu_\beta(\,\cdot\,|z_0,w_0)$ obtained by conditioning $\mu_\beta$ to
$z=z_0$ and $\wh\bbP\eta = w_0$. This leads to an initial datum with
deterministic `visible' part and random `invisible' part in the heat bath. Once
the initial datum is selected, the microscopic Hamiltonian system generates a
deterministic evolution for $(z(t),\eta(t))$. This leads to a \sde for the
variables $(z,w)$ on the space $\bfZ\ti \bfW$, see Figure \ref{fig:CG}. As in
Section \ref{ss:CG-part1.5-intro} it is the Gaussianity of $\gamma_\beta$ and
the Compression Property~\ref{ass:dilation} that guarantee
that the (Mori-Zwanzig type) coarse graining map 
\begin{equation}
  \label{eqdef:pi-zw-intro}
\bfZ\ti\bfX \ni (z,\eta) \ \stackrel{\pi_{z,w}}\longmapsto\  
(z,w:=\wh\bbP \eta )\in \bfZ\ti\bfW 
\qquad \qquad \text{(see Section~\ref{ss:Microscopic-positive-temperature} for
  $\wh\bbP$)}, 
\end{equation}
leads to a memory-less Markov process that can be described by a \sde. 

\begin{figure}[ht] 
\centering 
\begin{tikzpicture}[scale=0.65]

\node[color=blue] at (-8.5,8) {\bfseries\itshape macro}; \node[color=blue] at
(-8.5,2.3) {\bfseries\itshape micro};

\draw[fill=blue!10] (-7,6) rectangle (8.7,10); \node[right] at (-6.7,9.1)
{\small Stochastic process for $(z,w)\in \bfZ\ti \bfW$};
\begin{scope}[xshift=-4.5cm, yshift=6cm] \draw[fill= gray!20] (-2,0.5)
rectangle (2,2.5); \node[color=black] at (0,1.5){$(\bfZ,\calH_\calA,
\bbJ_\calA)$};
\end{scope}
\begin{scope}[xshift=-0.7cm, yshift=6cm] \draw[very thick,<->]
(-1.5,1.5)--(1.5,1.5); \node[color=blue] at (0,2){coupling}; \node[color=black]
at (0,0.8){$ \ip {\sfC z}{w}_\bfW$};
\end{scope}
\begin{scope}[xshift=4.5cm, yshift=6cm] \draw[fill= gray!20] (-3.4,0.3)
rectangle (3.7,2.6); \node[color=black, right] at (-3.2,2){ $w_t = \wh\bbP
\eta_t \in \bfW$}; \node[color=black, right] at (-3.2,1){ noise generated by
$\wh\bbQ \eta_t$};
\end{scope}

\draw[fill=blue!10] (-7,-0.8) rectangle (8,5); \node[right] at (-6.7,4.1)
{\small Hamilton.\ sys. for $(z,\eta)\in \bfZ\ti \bfH$ or $\bfZ\ti \bfX$};
 
\begin{scope}[xshift=-4.5cm] \draw[fill= gray!20] (-2,0.5) rectangle (2,2.5);
\node[color=black] at (0,1.5){$(\bfZ,\calH_\calA, \bbJ_\calA)$};
\end{scope} \draw[very thick,<->] (-2,1.5)--(2,1.5); \node[color=blue] at
(0,2){coupling}; \node[color=black] at (0,0.8){$ \ip {\sfC
z}{\wh\bbP\eta}_\bfH$};

\begin{scope}[xshift=4.5cm] \draw[fill= gray!20] (-2.0,-0.3) rectangle (3,3.3);
\node[color=black] at (0,2.6){$(\bfX,\HB^\infty, \wh\JHB)$};
\draw[thick,->] (0,1)--(0,2); \node[color=black,right] at (0,1.5){extension};
\node[color=black] at (0,0.4){$(\bfH,\calH_\calB, \JHB)$};
\end{scope}

\draw[ultra thick, color=blue!80!red,->] (4.7,3.5) -- node[pos=0.34, above
right]{\!$\wh\bbP$} (3.3,7.7); 
\draw[ultra thick, color=blue!80!red,->] (5.0,3.5) --
node[pos=0.64, right]{$\wh\bbQ=I{-}\wh\bbP$} (7.1,6.6);
\end{tikzpicture}

\caption{The coarse-graining step, illustrated as the transition from the
bottom row to the top row. We separate the heat-bath variable $\eta$ into the
observable part $w = \wh\bbP\eta$ and the remainder $\wh\bbQ\eta=\eta-\wh\bbP\eta$. The pair
$(z,w)$ can be considered a stochastic process driven by $\wh\bbQ\eta$.}
\label{fig:CG}
\end{figure}

Similarly to the deterministic case in Section~\ref{s:CG-deterministic} we also
need the `internal-energy' variable~$e$. For random initial data the formula
$e=\frac12\|\eta\|_\bfH^2 - \frac12\|w\|_\bfH^2$ is not available, however,
since the Hamiltonian $\calH_\calB(\eta) = \frac12\|\eta\|_\bfH^2$ is almost
surely infinite (see Remark~\ref{rem:pos-temp-means-inf-energy}). We circumvent
this problem by following an alternative but equivalent route: we first coarse-grain
 $(z,w)=\pi_{z,w}(z,\eta) =(z,\wh\bbP\eta)$ and then reconstruct $e$
 \emph{a posteriori} from $(z(t),w(t))$ as the energy
 exchanged with the heat bath.\medskip 

The main result for this setup is the second part of Theorem~\ref{th:recognize-as-GENERIC}, which shows that the curve $(z(t),w(t),e(t))$ constructed this way satisfies the following ``\Generic \sde'':
\begin{multline}
  \rmd\! \begin{pmatrix} z\\[-0.3em]w\\[-0.3em]e\end{pmatrix}
  = \pra[\Big]{\bbJ_\GEN \rmD \calE_\GEN(z,w,e) 
  + \bbK_\GEN(z,w)\rmD \calS_\GEN(z,w,e)  
    + \div \bbK_\GEN(z,w)}\dd t \\
  + \bbSigma_\GEN(z,w) \dd B(t),
  \label{eq:GENERIC-SDE1-intro}
\end{multline}
where $\bbJ_\GEN(z)$, $\bbK_\GEN$, $\calE_\GEN$, and $\calS_\GEN$ are as in Section~\ref{s:CG-deterministic}, and the mobility $\bbSigma_\GEN$ is given by
\begin{equation*}
  \bbSigma_\GEN(z,w) := 
  \begin{pmatrix}
    0 \\ \Sigma \\ -{\ip{w{+}\sfC z}{ \Sigma\,\square}_\bfH}
  \end{pmatrix} \quad \text{ with }\ \Sigma\Sigma^* = \frac1\beta (\bbD{+}\bbD^*).
\end{equation*}
This equation is a natural stochastic counterpart of the deterministic GENERIC equation~\eqref{eq.generic-intro}.
We discuss equations of this form in more detail in Section~\ref{ss:GENERIC-SDE}.

\bigskip

Again, the derivation of~\eqref{eq:GENERIC-SDE1-intro} is interesting for a
number of reasons. As in the deterministic case, the fact that this is an
\emph{exact} reduction is remarkable: not only is the drift term
in~\eqref{eq:GENERIC-SDE1-intro} Markovian, as we already remarked in
Section~\ref{ss:compression-assumption}, but in addition the noise term is
exact Brownian noise, and therefore memoryless. Considering that the evolution
of the microscopic variable $(z,\eta)$ is deterministic, and therefore
`maximally memory-full', this indeed is remarkable.

In addition, the mobility $\bbSigma_\GEN$ that colors the noise is coupled to the Onsager operator $\bbK_\GEN$ by the fluctuation-dissipation relation
\begin{equation*}
  \bbSigma_\GEN\bbSigma_\GEN^* = \frac 2\beta \bbK_\GEN.
\end{equation*}
This underscores the shared origin of $\bbSigma_\GEN$, $\bbK_\GEN$, and $\calS_\GEN$ in the randomness of the initial data. We indeed show in Section~\ref{ss:explanation-S-part2} how $\calS_\GEN$ can be interpreted as the mass of a microcanonical measure, and how this leads to the corresponding drift $\bbK_\GEN\rmD \calS_\GEN$  in~\eqref{eq.generic-intro} and~\eqref{eq:GENERIC-SDE1-intro}.

\subsection{Outline of this paper}
\label{sec:outline-this-paper}

The layout of this paper is as follows. The microscopic system is described in more detail in
Section~\ref{s:micro}.

In Section~\ref{s:dilations} we employ a compression to derive the coarse-grained system. Specifically, in
Section~\ref{s:compressions-dilations} we introduce compressions and derive an Ornstein-Uhlenbeck process on the projected space.
In Section~\ref{s:macro}
we perform the coarse-graining described in Sections~\ref{s:CG-deterministic} and~\ref{ss:CG-part2-intro}
above, which leads to deterministic and stochastic evolutions. At this stage they still are without \Generic
structure.

In Section~\ref{s:GENERIC} we describe  the \Generic structure in detail and establish some properties of so-called \Generic SDEs. By taking into account
the energy flowing in and out of the heat bath, as described in Sections~\ref{s:CG-deterministic} and~\ref{ss:CG-part2-intro},  we can formulate the equations obtained in Section~\ref{s:macro} as a \generic system, either deterministic or stochastic, and identify the
components. This also gives insight in how the \generic restrictions arise in the coarse-graining
process.

In Section~\ref{s:conclusion} we summarize the results of the paper, and discuss connections with various
other developments. In Appendix~\ref{app:GaussianMeasures} we collect a number of properties of Gaussian
measures that we use, and in Appendix~\ref{app:delayed-proofs} we give the proofs of a number of lemmas.

\subsection*{List of notation, setting, and assumptions}
We denote Hilbert spaces with boldface symbols and their elements with lower case letters.  Below we will use  Hilbert spaces $\bfH$, $\bfW$, and $\bfX$ satisfying $\bfW\subset \bfH\subset\bfX$; $\bfW$ and $\bfH$ share the same inner product $\ip\cdot\cdot _\bfH$, while $\bfX$ has a weaker inner product $\ip\cdot\cdot _\bfX$. 
We use the notation $\bbA^*$ for the Hilbert adjoint of an operator $\bbA$ between two Hilbert spaces.

We use the Lebesgue measure on general finite-dimensional Hilbert spaces; this is defined by identifying the Hilbert space with $\R^N$ through any orthonormal basis. The resulting measure is independent of the choice of basis. 
By the same type of identification any finite-dimensional Hilbert space admits a `standard normal distribution' and a `standard Brownian motion'.
The `narrow' convergence of measures on a Hilbert space~$\bfY$ is defined by duality with $\rmC_b(\bfY)$.

\vfill\eject

\subsubsection*{Other notation}
\begin{small}
\begin{longtable}{lll}
  $\beta>0$ & inverse temperature\\
  $\gamma_\beta$ & canonical Gaussian measure on the heat bath & \eqref{eq:gauss-rigorous}\\
  $\sfC$, $\bbA$ & coupling operator & \eqref{eqdef:HC}, \eqref{eqdef:Ham-full-1a}\\
  $\rmD$ & derivative operator & Rem.~\ref{re:notation-derivative-D}\\
  $\bbD$ & generator of the compressed semigroup & Ass.~\ref{ass:dilation} \\
  $\DDsym, \ \DDanti$ & symmetric and skew-symmetric parts of $\bbD$ &
  \eqref{eq:D-decom}
\\
  $\calH_{\calA,\calB,\calC,\mathrm{total}}$ & components of the microscopic
  Hamiltonian & Sec.~\ref{s:setting-this-paper}
\\
  $\calH_\calB^\infty$ & Extension of $\calH_\calB$ from $\bfH$ to $\bfX$ &
  Sec.~\ref{ss:CG-part2-intro}
\\
  $\bbI$ or $\bbI_{n\times n}$ & identity operator or matrix
\\
  $\bbJ_\calA$, $\JHB$ & Poisson operators of microscopic system &
  Secs~\ref{s:setting-this-paper}, \ref{ss:Microscopic-finite-energy}
\\
  $\mafo{Lin}(\bfX,\bfY)$ & bounded linear maps from $\bfX$ to $\bfY$
\\
  $\mu_\beta$ & invariant measure for the microscopic system &
  Secs~\ref{ss:CG-part2-intro}, \ref{ss:inv-measure-on-product-space}
\\
  $\mu_{\beta,\bfZ}$ & invariant measure for System~A &
  \eqref{def:mubeta-mubetaz-assumptions-section}, \eqref{eqdef:mubetaZ}
\\
  $\calN(m,\bbC)$ & normal distribution & App.~\ref{app:GaussianMeasures}
\\
  $\bbP$, $\wh\bbP$ & projection onto $\bfW$ from $\bfH$ or $\bfX$
  &Secs~\ref{ss:Microscopic-finite-energy},
  \ref{ss:Microscopic-positive-temperature}
\\
  $\pi_{z,w}$ & projection onto coarse-grained variables $(z,w)$ &
  \eqref{eqdef:pi-zw-intro}
\\
  $\pi_{z,w,e}$ & projection onto coarse-grained variables $(z,w,e)$ &
  \eqref{eqdef:projection-zwe}
\\
  $\bbQ$, $\wh\bbQ$ & complementary projections: $\bbQ= I-\bbP$, $\wh\bbQ = I-
  \wh\bbP$ & Secs~\ref{ss:Microscopic-finite-energy}, 
\ref{ss:Microscopic-positive-temperature}
\\
  $\wh S_t$ & extension of the group $\ee^{t\JHB}$ from $\bfH$ to $\bfX$ &
  Ass.~\ref{ass:HB-evol-op-ct-on-X}
\\
  $\bbU_t$ & evolution operator of the microscopic system & Th.~\ref{t:ex-un-full}\\
\end{longtable}
\end{small}

\bigskip\noindent
\textbf{Key assumptions} (see Section~\ref{s:micro} for details)
\medskip

{\setlength{\leftmargini}{1em}
\begin{itemize}\itemsep=0.3em
  \item $\bfH$, $\bfX$, and $\bfZ$ are real separable Hilbert spaces, with
    $\dim \bfZ < \infty$ and with %trace-class 
    compact embedding $\bfH\hookrightarrow \bfX$ such that
    $\bfH=\bbC^{1/2}\bfX$, where $\bbC=\bbC^*:\bfX\to \bfX$ is positive definite and
    trace class;
  \item $\JHB$ is the generator of both a strongly continuous unitary group
    $\ee^{t\JHB}$ on $\bfH$, and $\extG_t$ is an extension forming a
    strongly continuous group on $\bfX$;
  \item $\bbP$ is the orthogonal projection in $\bfH$ onto the
    finite-dimensional linear subspace $\bfW\subset\bfH$;
  \item $\sfC$ is a bounded linear map from $\bfZ$ to $\bfH$ satisfying
    $\bbP\sfC = \sfC$;
  \item $\beta>0$ and $\Hmac \in \rmC^2(\bfZ;\R)$ is such 
  \[
    \exp\bra*{-\beta \bra*{\Hmac(z) - \tfrac12 \|\sfC z\|_\bfH^2}}\in L^1(\bfZ),
  \]
  and $\rmD \calH_\calA$ is globally Lipschitz continuous; 
  \item $\Jmac(z)$ is a  skew-adjoint linear operator on $\bfZ$;
  \item $\gamma_\beta = \calN_{\bfX}\bra*{0,\beta^{-1}\bbC}$, \ $\mu_\beta  =
    \mu_{\beta,\bfZ}(\rmd z) \calN_\bfX (-\sfC z,\beta^{-1}\bbC)(\rmd \eta)$,
    and \vspace{-0.4em} 
  \begin{equation}
    \label{def:mubeta-mubetaz-assumptions-section}
    \mu_{\beta,\bfZ}(\rmd z) = \frac1{\sfZ_\beta} \exp\!\big({-}\beta\Hmac(z) {+}
    \frac\beta 2\|\sfC z\|_\bfH^2 \big) \dd z;
  \end{equation}
\item Compression Property~\ref{ass:dilation}:
  $\bbP\;\! \ee^{t\JHB} |_\bfW = \ee^{-t\bbD}$ for $t\geq0$; and
  $\bbD{+}\bbD^*$ has strictly positive eigenvalues.
\end{itemize}
}

\subsubsection*{Acknowledgements}
The authors wish to thank No\'e Corneille, Greg Pavliotis, Artur Stephan, and
the `Wednesday morning session' at Eindhoven University of Techology for
valuable input during the preparation of this manuscript. The research of
A.M. has been partially funded by Deutsche Forschungsgemeinschaft (DFG) through
grant CRC 1114 ``Scaling Cascades in Complex Systems'' (project no.\
235221301), Subproject C05 Effective Models for Materials and Interfaces with
Multiple Scales.

\subsubsection*{Data availability statement}
No datasets were generated in the preparation of this paper. 

\subsubsection*{Statements and Declarations}
The authors have no relevant financial or non-financial interests to disclose.

\section{The microscopic system}
\label{s:micro}

The microscopic setup of this paper is a well-known model (see,
e.g.,~\cite{JaksicPillet98,KSTT02LTBL,Reyb06OCS}). As mentioned above, two separate
Hamiltonian systems are coupled, where one is finite-dimensional (`System~A'), and the other is
infinite-dimensional and is interpreted as a heat bath (`System~B'). System~A may be nonlinear, while
System~B and the coupling will be linear. This is illustrated in Figure~\ref{fig:Step1}.

\noindent 
We construct the microscopic system and its properties in the following steps:
\begin{enumerate}

\item We define the two subsystems in spaces $\bfZ$ and $\bfH$ and their
  coupling in state space $\bfZ\ti\bfH$
  (Section~\ref{ss:Microscopic-finite-energy});

\item We consider the heat-bath subsystem at positive temperature, which
  requires extending the state space $\bfH$ to a larger space $\bfX$
  (Section~\ref{ss:Microscopic-positive-temperature});

\item We study the evolution of the heat bath on this larger state space, and
  characterize the observable part (Lemma~\ref{l:props-Y});

\item We use the properties of the heat bath to establish existence of
  solutions of the coupled system (Theorem~\ref{t:ex-un-full}).
\end{enumerate}

\subsection{Finite energy, deterministic evolution}
\label{ss:Microscopic-finite-energy}
We first recall from the Introduction the microscopic system at finite energy and without
randomness.  System~A is a fixed Hamiltonian system $(\bfZ,\Hmac,\Jmac)$, where
$\bfZ$ is a finite-dimensional real Hilbert space, $\Hmac$ is a nonlinear
Hamiltonian on $\bfZ$ satisfying some conditions (see below), and  $\Jmac$
defines a possibly state-dependent Poisson structure, i.e.\
$\Jmac(z)\colon\bfZ\to \bfZ$ is a skew-symmetric linear operator satisfying $\div_z \Jmac = 0$ and 
Jacobi's identity~\eqref{eq:JacobiIdent}.  We write $z\in \bfZ$ for the
variables in System~A, and the uncoupled motion of System~A is governed by the
equation $\dot z = \Jmac(z) \rmD\Hmac(z)$.

System~B is imagined to be a `heat bath'. Concretely, System~B is a Hamiltonian
system $(\bfH,\HB,\JHB)$, where $\bfH$ is an infinite-dimensional separable
real Hilbert space, $\HB := \frac12\|\cdot\|_\bfH^2$, and $\JHB$ is a
densely defined, unbounded skew-adjoint operator on $\bfH$ that generates
a strongly continuous, unitary group $(\ee^{t\JHB})_{t\in \R}$. We
write $\eta\in\bfH$ for the states of the heat bath, and the uncoupled
evolution of the heat bath is formally given by $\dot \eta = \JHB \eta  =\JHB
\rmD\HB(\eta)$.  In particular, $\eta(t)=\ee^{t\JHB} \eta(0)$ describes
the evolution of the heat bath. 

We assume that the heat bath is `large'\footnote{In the context of this paper, the heat bath being `large'
  translates into assumptions that the heat bath has many dimensions and a large total energy; in the limit
  this becomes an assumption of infinite-dimensionality and a Gaussian measure. We explore this limit in
  detail in Section~\ref{ss:explanation-S-part2}.}, but that only part of the heat bath is observable by
System~A. This limited observability is implemented by an orthogonal projection operator
$\bbP\colon\bfH \to \bfW =\bbP\bfH \subset \bfH$ and its complement $\bbQ := \bbI - \bbP$.  System~A only
perceives the heat bath through the operator $\bbP$. The range $\bfW$ of $\bbP$ is assumed to be a
finite-dimensional subspace of $\bfH$, and we use the important Compression Property~\ref{ass:dilation} about
$\bbP$ and $\bfW$ that we already mentioned in the introduction. Note that $\bbP$ is self-adjoint in $\bfH$.

The coupling between the macroscopic system and the heat bath is given via a linear coupling operator
$\sfC \colon\bfZ\to\bfH$ defining the bilinear coupling Hamiltonian 
\[
  \HC(z,\eta) := \ip {\bbA z}{\bbP\eta}_\bfH =
  \bra[\big]{\underbrace{\bbP\bbA}_{\sfC} z|\eta}_\bfH. 
\]
This expression shows how the impact of the heat bath on System~A is mediated
by the operator~$\bbP$. The total microscopic Hamiltonian system then is the
sum of the three parts,
\begin{align}
  \label{eqdef:Ham-full-1}
  \calH_{\mathrm{total}}(z,\eta)
  &= \Hmac(z)+ \HB(\eta)+ %% \frac 1 2 \|z\|_\bfZ^2
                   \HC(z,\eta) \\
  &= \Hmac(z) + \frac12\|\eta\|_\bfH^2 + \ip {\sfC z}{\eta}_\bfH \label{eqdef:Ham-full-1a}\\
  &= \Hmac(z) - \frac12 \|\sfC z\|_\bfH^2 + \frac12 \|\eta {+} \sfC z\|_\bfH^2 .
    \label{eqdef:Ham-full-2}
\end{align}

We can now make the conditions on $\Hmac$ concrete: we assume that
$\Hmac\in C^2(\bfZ)$ such that $\rmD\Hmac$ is globally Lipschitz continuous,
and that for our chosen $\beta>0$ we have 
\begin{equation}
\label{eq-cond:augmented-Hmac-integrable}\exp\bra*{-\beta\bra*{\Hmac(z) -
    \tfrac12 \|\sfC z\|_\bfH^2}} \in L^1(\bfZ). 
\end{equation}

The coupled (finite-energy) microscopic system is defined in terms of the
total Hamiltonian system
$(\bfZ\ti \bfH, \calH_{\mathrm{total}}, \bbJ_\calA{\oplus}\JHB )$, where
the Poisson operator
$\bbJ_{\mathrm{total}}(z) := \bbJ_\calA(z) {\oplus} \JHB\mathbin{\colon} % add \mathbin to allow breaking inside the formula
\bfZ\ti
\mafo{dom}(\JHB)\to \bfZ\ti \bfH$ still is skew-adjoint and satisfies Jacobi's
identity \eqref{eq:JacobiIdent}, since $\JHB $ is constant. The associated
Hamiltonian equations are~\eqref{eqdef:Evol-full-intro}, repeated here for convenience
\begin{subequations}
  \label{eqdef:Evol-full}
  \begin{align}
    \label{eqdef:Evol-full-1}
    \dot z_t &= \Jmac(z) \big( \rmD \Hmac (z_t) +
             \sfC^* \eta_t\big), \\
    \dot \eta_t &= \JHB\,\big(\eta_t + \sfC z_t\big).  \label{eqdef:Evol-full-2}	
  \end{align}
\end{subequations}
Here we follow the convention from probability to denote the time-dependence of a variable with a subscript
$t$.  When necessary for
better readability we write the time dependence in parentheses, as
in~$z(s)$. Existence and uniqueness of solutions of these equations is proved
in Theorem~\ref{t:ex-un-full} below. This will need a suitable weak
formulation of \eqref{eqdef:Evol-full-2} because $\mafo{range}(\sfC) \subset
\bfW$ is not contained in the domain of the unbounded operator $\JHB$ (see Remark \ref{re:Dom.JB.bfW}). 
% We write $\bbU_t$ for the solution group of the coupled system in $\bfZ \ti \bfH$.

\begin{RmRemark}[Convention for the derivative $\rmD$]
  \label{re:notation-derivative-D}
  In the context of Hilbert spaces, there is a choice to consider either the Fr\'echet derivative, which is an element of the dual space, or the `gradient', which is the Riesz representative of the Fr\'echet derivative and therefore an element of the primal space. In this paper we always choose the second interpretation, as  the Riesz representative of the Fr\'echet derivative.
  With this choice, we have for instance for all $\eta\in \bfH$ and $z\in \bfZ$,
  \[
    \rmD_\eta \frac12 \|\eta\|^2_\bfH = \eta, \qquad 
    \rmD_\eta \ip{\sfC z}{\eta}_\bfH = \sfC z, \qquad\text{and}\qquad
    \rmD_z \ip{\sfC z}{\eta}_\bfH =\rmD_z \ip{z}{\sfC^*\eta}_\bfZ = \sfC^* \eta.
    \qedhere
  \]
  This choice also implies that an operator such as $\JHB$ is an unbounded operator from $\bfH$ to $\bfH$ (and not from $\bfH^*$ to $\bfH$).
\end{RmRemark}

\subsection{Positive temperature: The heat bath}
\label{ss:Microscopic-positive-temperature}

We now extend this setup to positive temperature, first focusing on just the
heat bath and its evolution. We recall in Appendix~\ref{app:GaussianMeasures}
the properties of Gaussian measures that we use.\medskip

Let $\beta>0$ be interpreted as inverse temperature. Since the heat-bath
Hamiltonian is $\calH_\calB(\eta)= \frac12 \|\eta\|_\bfH^2$, a thermodynamic or
statistical-mechanical point of view (see Section~\ref{ss:explanation-S-part2})
suggests to consider a Gaussian measure with the formal expression
\begin{equation}
  \label{eq:gauss-formal}
  \gamma_\beta(\rmd\eta)  = \exp\bra*{-\frac\beta 2 \|\eta\|_\bfH^2}\rmd\eta
\end{equation}
on a suitable space containing $\bfH$. This measure should be the natural invariant measure for the
Hamiltonian evolution $\dot \eta_t = \JHB \eta_t$ of the isolated heat bath. To define such a Gaussian
  measure on an infinite-dimensional Hilbert space $\bfH$, an embedding into a bigger space is required, and
  is provided by the Cameron-Martin setting sketched below. A reason is that, unlike in finite dimensions, a
  special structure, namely trace-class nature, of the covariance is required. For example, in infinite
  dimensions there is no Gaussian measure with the identity as covariance, unlike in finite dimensions.

We briefly recall the Cameron-Martin setting which establishes a meaning for
this expression; see Appendix~\ref{app:GaussianMeasures} for more detail. Let
the Hilbert space $\bfH$ be embedded into a larger separable Hilbert
space~$\bfX$, in such a way that the identity map $\mathrm{id}: \bfH\to\bfX$ is
a Hilbert-Schmidt operator. Then
$\bbC := \mathrm{id}\,\mathrm{id}^*:\bfX\to\bfX$ is of trace class, and we
rigorously define $\gamma_\beta$ as the centered Gaussian measure
\begin{equation}
  \label{eq:gauss-rigorous}
\gamma_\beta := \calN_\bfX\bra*{0,\beta^{-1} \bbC}. 
\end{equation}
The space  $\bfH$ can be considered to be the `Cameron-Martin space' of $\gamma_\beta$.

From now on we will consider the larger space $\bfX$ as the state space for the
heat bath, and the larger space $\bfZ\ti \bfX$ as the state space for the whole
system. We therefore need to extend the operator $\bbP$ from $\bfH$ to
$\bfX$. By Lemma~\ref{l:measurable-extensions}, the operator
$\bbP\colon\bfH\to \bfW\subset \bfH$ has a \emph{measurable linear extension}
$\wh \bbP\colon \bfX \to \bfW\subset \bfH$, which is $\gamma_\beta$-a.e.\
uniquely characterized by the property that it is linear and coincides with
$\bbP$ on~$\bfH$ (see Definition~\ref{def:measurable-linear-extensions}). The
complementary operator $\wh \bbQ := \mathrm{id} - \wh\bbP:\bfX\to\bfX$
similarly is a measurable linear extension of $\bbQ$.  We also define
$\wh\sfC^* := \sfC^* \wh \bbP \colon\bfX\to \bfZ$, which is a measurable extension of 
$ \sfC^* \colon\bfH\to \bfZ$.

\begin{RmRemark}[Positive temperature means infinite energy]
  \label{rem:pos-temp-means-inf-energy}
  The properties of Gaussian measures imply that if $\eta\sim \gamma_\beta$,
  then $\HB(\eta) = \frac12\|\eta\|^2_\bfH = +\infty$ almost surely---despite
  the fact that $\HB(\eta)$ appears in the exponent in~\eqref{eq:gauss-formal}.
  This can be recognized by choosing an orthonormal basis $\{e_k\}_{k\in \N}$
  of $\bfH$, and observing that the vector
  $\bra[\big]{\ip{\eta}{e_1}_\bfH,\dots,\ip{\eta}{e_n}_\bfH}$ of the first $n$
  coordinates of $\eta$ in this basis has variance $n/\beta$
  (see~\eqref{eq:variance-of-finite-rank-operator}); therefore the 
  expectation of $\|\eta\|_\bfH^2$ can not be finite.  See,
  e.g.,~\cite[Sec.~2.1]{Reyb06OCS} or~\cite[Rem.~2.2.3]{Boga98GM} for a
  discussion of this property.

As a result, `finite energy' requires `zero temperature'; positive temperature
implies infinite energy.
\end{RmRemark}

\bigskip

We formally denote the extended system with finite temperature (and
infinite energy) by
$\mathfrak S_\beta:=(\bfZ\ti \bfX, \calH_\mafo{total}^\infty,
\Jmac{\oplus}\JHB,\gamma_\beta)$, such that the new ingredients
$\bfX\supset \bfH$ and $\gamma_\beta = \calN_\bfX(0,\beta^{-1} \bbC)$ are
visible. Here $\calH_\mafo{total}^\infty $ coincides with
$\cal\calH_\mafo{total}$ on $\bfZ\ti \bfH$ and equals $+\infty$ for
$\eta\in \bfX\setminus\bfH$.  In Section \ref{su:WellPosednessSDE} we discuss how the
solution operator for the finite-energy microscopic system
$(\bfZ\ti \bfH, \calH_\mafo{total}, \Jmac{\oplus}\JHB)$ can be
extended to the infinite-energy system~$\mathfrak S_\beta$.\bigskip

The evolution operator $\ee^{t\JHB}$, that describes the heat-bath unitary
evolution in the smaller space $\bfH$, can similarly be extended
$\gamma_\beta$-a.e.\ to an operator $\extG_t$ on
$\bfX$. Without loss of generality we make the following simplifying assumption
(see Section~\ref{ss:extens-chosen-contin} for a discussion):

\begin{ass}[Continuous evolution of the heat bath]
\label{ass:HB-evol-op-ct-on-X}
The operator $\ee^{t\JHB}\colon\bfH\to \bfH$ can be extended to a strongly
continuous group $(\extG_t)_{t\in \R}$ of bounded linear operators $\extG_t$
from $\bfX$ to $\bfX$.
\end{ass}

The extended operators $\extG_t$ leave $\gamma_\beta$ invariant as is
shown in Lemma~\ref{le:GauMea.inv.extG}.\medskip 

For random initial datum $\eta_0\sim \gamma_\beta$, the process 
\[
t\mapsto \extG_t\eta_0 \in \bfX
\]
is a stochastic process in $\bfX$; its evolution is deterministic in time, but
the randomness of the initial datum implies that the process as a whole is
stochastic. Under Assumption~\ref{ass:HB-evol-op-ct-on-X}, this process has
continuous sample paths in $\bfX$.  The next lemma states some properties of
the \emph{projected} version of this process, which has values in $\bfW$, which
is equipped with the stronger norm of $\bfH$.

\begin{lemma}
\label{l:props-Y}
Adopt the setup above, including Assumption~\ref{ass:HB-evol-op-ct-on-X}, and set
\[
Y_t(\eta_0) := \wh \bbP \extG_t \eta_0 \quad \text{for all} 
 \quad t\in \R,\  \eta_0 \in \bfX.
\]
Then, the following holds:
\begin{enumerate}
\item\label{l:props-Y:ct-in-H} For every $\eta_0\in \bfH$ we have  $t\mapsto Y_t (\eta_0)\in C(\R;\bfW)$.
\item\label{l:props-Y:ct-in-X} Choose the probability space $(\bfX, \calB(\bfX)_{\gamma_\beta^{}}, \gamma_\beta)$, where $\calB(\bfX)_{\gamma_\beta^{}}$ is the $\gamma_\beta$-completion of the Borel $\sigma$-algebra $\calB(\bfX)$. If $\eta_0$ is drawn from $\bfX$ with law $\gamma_\beta$, then both ${Y_t} = {Y_t}(\eta_0)$ and ${Y_t}={Y_t}(\wh\bbQ\eta_0)$ are measurable $\bfW$-valued stochastic processes with sample paths in $L^2_{\mathrm{loc}}(\R;\bfW)$; in addition ${Y_t}(\eta_0)$ is stationary.

\end{enumerate}
\end{lemma}
\noindent
The proof is given in Appendix~\ref{app:delayed-proofs}.
Part~\ref{l:props-Y:ct-in-H}, the continuity in $\bfH$, follows directly from
the definitions; but more interesting is part~\ref{l:props-Y:ct-in-X},
where we re-interpret the evolution as a stochastic process. The randomness of
this process is naturally parametrized by the starting point $\eta_0$ drawn
from $\gamma_\beta$, leading to the choice of probability space
$(\bfX, \calB(\bfX)_{\gamma_\beta^{}}, \gamma_\beta)$. The fact that the sample
paths are in $L^2_{\mathrm{loc}}$ follows from the fact that for fixed $t$,
$Y_t(\eta_0)$ is a Gaussian random variable in $\bfW$ with finite and
$t$-independent variance, so that
\[
\Expectation_{\eta_0\sim\gamma_\beta} \int_a^b \|Y_t(\eta_0)\|_\bfH^2 \,\rmd t
  =   \int_a^b \Expectation _{\eta_0} \|Y_t(\eta_0)\|_\bfH^2
  = C\;\!(b{-}a)
  \qquad \text{for } {-}\infty<a<b<\infty.
\]
The additional characterization of ${Y_t}(\wh\bbQ\eta_0)$ will be useful below, when we solve the initial-value problem with a given value of $w_0 := \wh\bbP\eta_0$.

\subsection{Well-posedness of the stochastic process on  the product space}
\label{su:WellPosednessSDE}

With the characterization of the process ${Y_t}$ by Lemma~\ref{l:props-Y} we can now construct the solution of the coupled microscopic system on the product space $\bfZ\ti\bfH$ (deterministic and with finite energy) or $\bfZ\ti\bfX$ (random and with infinite energy).

In Theorem~\ref{t:ex-un-full} below we also transform the equation~\eqref{eqdef:Evol-full-2} into a form that allows for a rigorous definition of a solution. The issue in~\eqref{eqdef:Evol-full-2} is that the range of the operator~$\sfC$ might not be included in the domain of the unbounded operator~$\JHB$; the integrated versions~\eqref{eq:EvolFullExistTheorem-2} and~\eqref{eq:EvolFullExistTheorem-ZX-2} address this.

\begin{theorem}
\label{t:ex-un-full}
  Assume that $\rmD \calH_\calA $ is globally Lipschitz continuous on $\bfZ$. Then there exists a collection of nonlinear operators $\{\bbU_t\}_{t\in\R}$ on $\bfZ\ti \bfX$  with the following properties:
\begin{enumerate}
  \setlist{nolistsep}
\item \label{i:t:ex-un-full-ZH} For $(z_0,\eta_0)\in \bfZ\ti \bfH$, the mapping
  $t\mapsto (z_t,\eta_t):= \bbU_t(z_0,\eta_0)$ has the regularity
  $z\in C^1(\R;\bfZ)$ and $\eta\in C(\R;\bfH)$, and it satisfies
\begin{subequations}
\label{eq:EvolFullExistTheorem}
\begin{align}
  \label{eq:EvolFullExistTheorem-1}
  \dot z_t &= \Jmac ( \rmD \Hmac (z_t) +
            \sfC^* \eta_t) &\text{for }& t\in \R\\
    \eta_t &= \ee^{\JHB t} \left( \eta_0+ \sfC z_0\right) - \sfC z_t + \int_0^t \! 
  \ee^{\JHB(t-s)} \sfC\dot z_s  \dd  s  &\text{for }& t\in \R.
  \label{eq:EvolFullExistTheorem-2}
\end{align}
\end{subequations}
\item\label{i:t:ex-un-full-ZX} For all $z_0\in \bfZ$ and $\gamma_\beta$-a.e.\
  $\eta_0\in\bfX$, the function $t\mapsto (z_t,\eta_t):= \bbU_t(z_0,\eta_0)$
  has the regularity $z\in W^{1,2}_{\mathrm{loc}}(\R;\bfZ)$,
  $\eta\in C(\R;\bfX)$ and it satisfies
\begin{subequations}
  \label{eq:EvolFullExistTheorem-ZX}
  \begin{align}
    \label{eq:EvolFullExistTheorem-ZX-1}
    \dot z_t &= \Jmac ( \rmD \Hmac (z_t) +
              \wh \sfC^* \eta_t), && \text{a.e. }t\in \R\\
      \eta_t &= \extG_t \left( \eta_0 {+} \sfC z_0\right) - \sfC z_t + \int_0^t \! 
    \ee^{\JHB(t-s)} \sfC\dot z_s  \dd  s  && t\in \R
    \label{eq:EvolFullExistTheorem-ZX-2}
  \end{align}
\end{subequations}
(note the extended operators\/ $\wh\sfC^*$
in~\eqref{eq:EvolFullExistTheorem-ZX-1} and $\extG_t$
in~\eqref{eq:EvolFullExistTheorem-ZX-2}).  For fixed $z_0$ and random
$\eta_0\sim\gamma_\beta$, the resulting stochastic process
$(z_t,\eta_t)_{t\in\R}$ is a measurable $\bfZ\ti\bfX$-valued process on the
probability space $(\bfX,\calB(\bfX)_{\gamma_\beta},\gamma_\beta)$.
\item\label{i:t:ex-un-full-ZXw} Fix $z_0\in \bfZ$ and $w_0\in \bfW$, and let
  $\eta_0 = w_0 + \xi_0$ with $\xi_0\in \wh\bbQ\bfX$ and
  $\xi_0\sim \wh\bbQ_\#\gamma_\beta$. For almost every such $\xi_0$ the curve
  $(z_t,\eta_t) :=\bbU_t(z_0,\eta_0)$ has the same properties as under
  part~\ref{i:t:ex-un-full-ZX} above.
%\item $\bbU$ is a group: for all $s,t\in \R$, $\bbU_{t+s} = \bbU_t \circ \bbU_s$, and $\bbU_0=\mathrm{id}$.
\end{enumerate}
\end{theorem}
\begin{proof}
  By integrating~\eqref{eqdef:Evol-full-2} and substituting
  into~\eqref{eqdef:Evol-full-1} we find the equation for $z$ alone,
\begin{equation}
\label{eq:z-Y}	
\dot z_t = \Jmac \sfC^*Y_t
 + \Jmac \bra*{\rmD\Hmac(z) + \sfC^*(\ee^{t\JHB}\sfC z_0 {-} \sfC z_t) } 
 + \Jmac \int_0^t \sfC^*  \ee^{(t-s)\JHB}\sfC \dot z_s\dd s,
\end{equation}
where ${Y_t} = {Y_t}(\eta_0)$.  For $(z_0,\eta_0)\in \bfZ\ti \bfH$, the process
$t\mapsto Y_t(\eta_0)$ is continuous with values in~$\bfW$ by
Lemma~\ref{l:props-Y}, and a standard ODE proof based on contraction for small
times gives short-time existence and uniqueness of a solution of~\eqref{eq:z-Y}
with the specified regularity. The global Lipschitz bound on $\rmD\calH_\calA$
implies that the time of existence does not depend on the initial datum, and we
obtain existence globally in time. Equation~\eqref{eq:EvolFullExistTheorem-1}
can then be recovered from~\eqref{eq:z-Y} by differentiating in time, and
equation~\eqref{eq:EvolFullExistTheorem-2} can be used to reconstruct
$\eta$. This proves the part 1 of the theorem.

For $z_0\in \bfZ$ and $\eta_0\sim \gamma_\beta$, the stochastic process
$Y(\eta_0)$ has sample paths almost surely in $L^2_{\mathrm{loc}}(\R;\bfW)$ by
Lemma~\ref{l:props-Y}. Again an ODE contraction proof gives the unique
existence of a solution of~\eqref{eq:z-Y} with regularity
$z\in W^{1,2}_{\mathrm{loc}}(\R;\bfZ)$. To recover $\eta$, remark that
$X = X_t(\eta_0) := \extG_t\eta_0\in \bfX$ has sample paths in
$C(\R;\bfX)$; again equation~\eqref{eq:EvolFullExistTheorem-2} can be
used to define $\eta$ in terms of $z$. This proves the existence in
part~\ref{i:t:ex-un-full-ZX}. The measurability of $(z,\eta)$ jointly in $t$
and $\eta_0$ follows from observing that $(t,\eta_0)\mapsto Y_t(\eta_0)$ is
measurable by Lemma~\ref{l:props-Y}, and by Fubini the measurability transfers
to the integral in time.

For part~\ref{i:t:ex-un-full-ZXw}, we use that $Y_t(\eta)$ is linear in
$\eta$ and split $Y_t$ in the first term on the
right-hand side in~\eqref{eq:z-Y} as
\[
  Y_t(\eta_0) = Y_t(\wh\bbQ\eta_0) + Y_t(\wh\bbP\eta_0) = Y_t(\wh\bbQ\eta_0) +
  \bbP\ee^{t\JHB}\wh\bbP\eta_0.
\]
The second term on the right-hand side above equals $\bbP\ee^{t\JHB}w_0$ and is
continuous in time; the first term again has sample paths almost surely in
$L^2_{\mathrm{loc}}(\R;\bfW)$ by Lemma~\ref{l:props-Y}. The sum therefore
almost surely is in $L^2_{\mathrm{loc}}(\R;\bfW)$ and can be treated the same
way as in the proof of part~\ref{i:t:ex-un-full-ZX}.
\end{proof}

\begin{RmExample}[Running example, part 3]
\label{ex:RunningExa3} For our model discussed in Examples~\ref{ex:RunningExa1}
and~\ref{ex:RunningExa2} we can explicitly construct the larger Hilbert
space $\bfX\supset \bfH= \rmL^2(\R;\R^3) = \rmL^2(\R)^3$. In fact, this can be done
independently for each of the three components. We follow the classical
construction in  \cite{Reyb06OCS} and choose 
\[
\bfX = X\ti  X \ti X \quad \text{with } X:= \rmH^{-2}(\R) + \rmL^2_{-2}(\R) .
\]
Here $X$ contains distributions $T$ that can be decomposed into $T=T_1+T_2$
with $T_1\in \rmH^{-2}(R)$ and $T_2\in \rmL^2_{-2}(\R)$, i.e.\ $x\mapsto
(1{+}x^2)^{-1}T_2(x) \in \rmL^2(\R)$. Clearly, the shift operator $\extG_t: T
\mapsto T(\!\;\cdot\;\!{-}t) $ generates a continuous semigroup that is the
extension of the unitary shift group $(\ee^{t\pl_x})_{t\in \R}$ on
$\rmL^2(\R)$.  

The regularity theory for Brownian motion shows that for all $\alpha\in
[0,1/2)$ the paths $t\mapsto B_t$ lie in
$\rmH^\alpha_\mafo{loc}(\R)$ almost surely, but not in $\rmH^{1/2}_\mafo{loc}(\R)$ (see e.g.\
\cite{HytonenVeraar08}). Thus, the white-noise distributions $\eta = \rmd B_t$ 
lie in $H^{\alpha-1}_\mafo{loc}(\R)$ only. 
Moreover, our dilation functions $f_j$, $j=1,2,3$, in
Example \ref{ex:RunningExa2} have a jump at $y=0$ and are otherwise
smooth. Hence, they lie in $\rmH^{\beta}(\R)$ with $\beta<1/2$, but not in
$H^{1/2}(\R)$. Thus, the extension of the projection $\bbP: h \to \sum_{j=1}^3
\ip{f_j}{h} f_j$ to $\wh\bbP$ remains nontrivial in the sense that the subspace
$L$ of $\bfX$ of full measure needs to be constructed by the abstract methods 
of Lemma \ref{l:measurable-extensions}. 
\end{RmExample}

\subsection{Invariant measure on the product space $\protect\bfZ\ti \bfX$}
\label{ss:inv-measure-on-product-space}

Since the stochastic process on $\bfZ\ti\bfX$ is a Hamiltonian system with the
Hamiltonian $\calH_{\mathrm{total}}$ given in~\eqref{eqdef:Ham-full-1}, a
natural stationary measure should be of the (formal) form
\[
  \exp\bra[\big]{-\beta \calH_{\mathrm{total}}(z,\eta)}\dd z\dd\eta.
\]
We make this expression rigorous by splitting off the
infinite-dimensional Gaussian component, as in~\eqref{eqdef:Ham-full-2}:
\[
   \exp\bra*{-\beta \calH_{\mathrm{total}}(z,\eta)}
  = \exp\bra*{-\beta\big(\Hmac(z) - \frac12\|\sfC z\|_\bfH^2 \big)}
  \exp\bra*{-\frac\beta 2\|\eta {+} \sfC z\|_\bfH^2}.
\]
This leads to the rigorous definition 
\begin{equation}
  \label{eqdef:mubeta}
  \mu_\beta(\rmd z\dd\eta) :=  \mu_{\beta,\bfZ}(\rmd z) \calN_\bfX (-\sfC
  z,\beta^{-1}\bbC)(\rmd \eta), 
\end{equation}
where 
\begin{equation}
  \label{eqdef:mubetaZ}
\mu_{\beta,\bfZ}(\rmd z) = \frac1{\sfZ_\beta} \exp\bra*{-\beta\big(\Hmac(z) -
  \frac12\|\sfC z\|_\bfH^2 \big) }\dd z. 
\end{equation}

The following result shows that the measure $\mu_\beta$ as defined above is
indeed invariant under the flow. The proof needs some nontrivial steps because
of the  coupling of the finite-dimensional and the
infinite-dimensional system. On the one hand, $\bfW$ is not contained in the
domain of $\Jmac$ and hence not in the domain of the generator of the extended
group $(\extG_t)_{t\in \R}$. On the other hand, the measure $\gamma_\beta$ is
invariant under $(\extG_t)_{t\in \R}$, but the covariance matrix $\bbC\colon \bfX\to
\bfX$ is not, see the proof of Lemma \ref{le:GauMea.inv.extG}.

\begin{theorem}[Invariance of $\mu_\beta$]
\label{th:MeasInvariant} Assume the above setting, including Assumption~\ref{ass:dilation}, as well as
\begin{align}
\label{eq:Div.Jmac=0}
&\Jmac \in \rmC^1_\rmb(\bfZ;\mafo{Lin}(\bfZ)) \quad \text{and}\quad 
    \div \Jmac(z)=0 \text{ on } \bfZ,
\\
\label{eq:StrongerIntegrability} 
& z\mapsto \big(1+ \| \rmD\Hmac(z) {-} \sfC^*\sfC z\|\big)
\ee^{-\beta(\Hmac(z)-\|\sfC z\|^2/2)} \quad \text{lies in } \rmL^1(\bfZ). 
\end{align}
Then, the measure
$\mu_\beta$ defined in \eqref{eqdef:mubeta} is invariant under the flow generated by
\eqref{eq:EvolFullExistTheorem-ZX}. 
\end{theorem}

\begin{proof}
We do a coordinate transformation to simplify the structure of
$\mu_\beta$ as well as of the evolution equation. 

By introducing the new variables $(z_t,\zeta_t)=\wt M_t(z,\eta) = \big( 
z,\extG_{-t}(\eta{+}\sfC z)\big)$ 
the measure $\mu_\beta$ is transformed into 
\[
\wt \mu_\beta = (\wt M_t)_\# \mu_\beta = \mu_{\beta,\bfZ} (\rmd z)
\,\gamma_\beta(\dd\zeta). 
\] 
The result is simple for $t=0$ because we only shift by the mean value. This
leads to~$\wt \mu_\beta$. The subsequent transformation by $\extG_t$ does not change the
resulting product measure because $\gamma_\beta$ is invariant (see Lemma~\ref{le:GauMea.inv.extG}).

Thus, it remains to show that the measure $\wt\mu_\beta$ is invariant under the
transformed evolution equation for $(z,\zeta)$. Using the  linear
structure of  $\wt M_t$ we see that $(z_t,\eta_t)$ is a solution of 
\eqref{eq:EvolFullExistTheorem-ZX} if and only if $(z_t,\zeta_t)=\wt
M_t(z_t,\eta_t)$ solves  
\[
\dot z_t= \Jmac(z_t)\big( \rmD\calH_\sfC(z_t) + \sfC^*\wh\bbP\extG_t
\zeta_t\big)  , \quad \zeta_t = \zeta_0 + \int_0^t \ee^{-s\JHB} \sfC \dot z_s \dd
s,
\]
where $\calH_\sfC(z):=\Hmac(z)- \frac12\|\sfC z\|^2$. 
In particular, we conclude that $t\mapsto \zeta_t\in \bfX$ is differentiable
with a derivative in the smaller space $\bfH$. Hence, we can use that $\zeta_t$ satisfies the differential equation 
\[
\dot \zeta_t = \ee^{-t\JHB} \sfC \dot z_t   \in \bfH  \quad \text{with initial
  condition }\zeta_0=\eta_0+\sfC
z_0 \in \bfX. 
\]
The advantage of using $\zeta$ is that the unbounded generator $\wh
\JHB$ of the group $(\extG_t)_{t\in \R}$ does not appear. 

We will show in Lemma~\ref{l:Y-is-OU} that $\zeta_0\sim \gamma_\beta$ implies
that $Y_t =\wh\bbP \extG_t \zeta_0$ is an Ornstein-Uhlenbeck process, such that
the function $t\mapsto Y_t$ is almost surely continuous. As a consequence one
can show as in Theorem \ref{t:ex-un-full} that
$t\mapsto (z_t,\zeta_t{-}\zeta_0)$ lies in
$\rmC^1_\mafo{loc}(\R;\bfZ\ti \bfH)$.

Let us denote by $\Phi_t\colon \bfZ\ti \bfX\to \bfZ\ti \bfX$ with
$(z_t,\zeta_t)=\Phi_t(z_0,\zeta_0)$ the corresponding flow map. Then,
invariance of $\wt\mu_\beta$ means $\wh\mu_\beta^{(t)}:=(\Phi_t)_\#
\wt\mu_\beta=\wt\mu_\beta $ for all $t\in \R$. The invariance follows from the
transport equation for $t\mapsto \wt \mu_\beta^{(t)}$, namely 
\begin{align*}
&\pl_t \wh\mu_\beta^{(t)} + \div_{\bfZ\ti\bfX} \big( V_t(z,\zeta) 
  \wh\mu_\beta^{(t)} \big) = 0,
\\ &\text{with } 
V_t(z,\zeta) = \bma{c} \Jmac(z)\, ( \rmD\calH_\sfC(z) {+}
  \sfC^*\wh\bbP\extG_t \zeta ) \\[0.3em] \extG_{-t}\, \sfC \,\Jmac(z)\,(
    \rmD\calH_\sfC(z) {+} \sfC^*_{}\wh\bbP\extG_t \zeta ) \ema ,
\end{align*}
which is to be understood in the sense of distributions by applying suitable
test functions in $\rmC^1_\rmb(\bfZ\ti \bfX)$. (For clarity, we have written
$\sfC^*\wh\bbP$ instead of $\wh{\sfC^*}$, which emphasizes that the only
unbounded term is $\wh\bbP$.)

Thus, it remains to show the identity 
\begin{equation}
  \label{eq:Div=0}
  \int_{\bfZ\ti\bfX} \Big\langle \binom{\rmD_z\widecheck\Phi(z,\zeta)}{\rmD_\zeta
  \widecheck\Phi(z,\zeta)} ,  V_t(z,\zeta) \Big\rangle \,\wt \mu_\beta 
  (\dd z, \dd \zeta)  = 0 \text{ for all } \widecheck\Phi \in \rmC^1_\rmb(\bfZ\ti \bfX).
\end{equation}
Indeed, testing with functions
$\widecheck\Phi(z,\zeta) = \wt \Phi(z,\extG_t\zeta)$ and using the invariance
of $\wt\mu_\beta= \mu_{\beta,\bfZ}\,\gamma_\beta$ under
$\mafo{id}_\bfZ{\oplus} \extG_t$ we see that it suffices to prove
\eqref{eq:Div=0} for $t=0$ only.

For the final step we work with cylindrical test functions, see
\cite[Def.\,5.1.11]{AmGiSa05GFMS}. For this, we choose the orthogonal basis
generated by $\bbC=\bbC^*>0$ in $\bfX$, and denote by $\sfQ_n:\bfX\to \bfX$ the
orthogonal projection on to the span of the eigenfunctions associated with the $n$
largest eigenvalues. Because  $\bfH = \bbC^{1/2}\bfX$ we have $\bfH_n \coloneq \sfQ_n \bfX
\subset \bfH$ and we may use the decomposition $\bfX = \bfH_n \oti \bfY_n$ with
$\bfY_n = (\bbI{-}\sfQ_n)\bfX$. Moreover, since the kernel of
$\bbC$ is trivial, the restriction $\bbC_n=\bbC|_{\bfH_n}\colon \bfH_n \to
\bfH_n$ is invertible and we may choose the scalar product of
$\bfH$ on the finite-dimensional space $\bfH_n$.

Writing $\eta= h+y$ with $h \in \bfH_n$ and $y \in \bfY_n$, we now consider
the cylindrical test functions  $\widecheck\Phi(z,\eta)=\Psi(z,h)$ with
$\Psi\in \rmC^1_\rmb(\bfZ\ti \bfH_n)$, i.e.\ we restrict to a
finite-dimensional space, where computations are classical.   

To integrate over the infinite-dimensional space $\bfY_n$ we observe that
assumption~\eqref{eq:StrongerIntegrability}  implies that the integrand in
\eqref{eq:Div=0} lies in $\rmL^1(\bfZ\ti \bfH_n\ti\bfY_n, \wt\mu_\beta)$. For
this we also use that  $\bfX \ni \eta \mapsto \wh\bbP \eta \in \bfW$
  lies in $\rmL^2(\bfX,\gamma_\beta)$, see
  part~\ref{l:measurable-extensions-operators} of
  Lemma~\ref{l:measurable-extensions}, i.e.\
  equation~\eqref{eq:variance-of-finite-rank-operator}. 

Hence, we can apply Fubini's
theorem and integrate first over $y\in \bfY_n$. Here we can exploit that the
integrand has the affine form $a_\Psi(z,h) + \big(b_\Psi(z,h) \big| \wh\bbP y \big)_\bfH$
and that the measure on $\bfY_n$ is a centered Gaussian. Thus, the linear term
vanishes, and we are left with an  integral over $\bfZ\ti \bfH_n$, namely 
\begin{align*}
\text{LHS\eqref{eq:Div=0}} &= \int_{\bfZ\ti \bfH_n}  \!\! \Big\langle
\binom{\rmD_z\Psi(z,h)}{ \rmD_y \Psi(z,h)}, 
\bma{@{}c@{}} \Jmac(z)( \rmD\calH_\sfC(z) {+}\sfC^* h) \\ 
\sfC \Jmac(z)( \rmD\calH_\sfC(z) {+}\sfC^* h) \ema \!\Big\rangle 
 \, \frac{\ee^{-\beta( \calH_\sfC(z) + \frac12\|h\|_\bfH^2)}}{\wt Z_\beta} \dd z \dd h ,
\end{align*}  
where we used that $h \in \bfH_n\subset \bfH$ which implies that $\wh\bbP y= \bbP y$ and
hence $\sfC^*\wh\bbP y = \sfC^*\bbP y = \sfC^* y$. Moreover, the induced
measure from $\gamma_\beta$ on the subspace $\bfH_n$ is $\wh Z_\beta^{-1}
\ee^{-\beta \|h\|_\bfH^2/2} \dd h$.  

Neglecting the irrelevant normalization constant $\wt Z_\beta$ and integrating by
parts we obtain 
\begin{align*}
\int_{\bfZ\ti \bfH_n} \Psi(z,h) &\big( {-} M_1(z,h)  - M_2(z,h)
+\beta M_3(z,h) \big)  \ee^{-\beta(
    \calH_\sfC(z) + \frac12\|h\|_\bfH^2)} \dd z \dd h 
\\
  \text{with } & M_1(z,h)= \div_{\bfZ}(\Jmac(z)( \rmD\calH_\sfC(z){+}\sfC^* h)), 
\\
 & M_2(z,h) = \div_{\bfH_n}( \sfC \Jmac(z)( \rmD\calH_\sfC(z) {+}\sfC^* h)),
\\
\text{and }& M_3(z,h)=\big\langle \bma{@{}c@{}} \rmD\calH_\sfC(z)\\[-0.3em] h
  \ema, V_0(z,h) \big\rangle .
\end{align*}
Explicit calculations show that all three terms are zero. Indeed, using $\div
\Jmac\equiv 0$ (see \eqref{eq:Div.Jmac=0}) and $\Jmac(z)^*=-\Jmac(z)$ we have 
\[
M_1(z,h)= \big(\div \Jmac\big)( \rmD\calH_\sfC(z){+}\sfC^* h) + \Jmac(z){:}
(\rmD^2 \calH_\sfC(z) +0) = 0,
\]
where ``$:$'' means matrix scalar product, $A{:}B=\sum_{i,j} A_{ij}B_{ij}$.  
Similarly, we have 
\[
M_2(z,h) =   \sfC \Jmac(z): \sfC^*\nabla_{\!y} h= (\sfC \Jmac(z) \sfC^*) : \bbI
=0.
\]
Finally, the third term can be rearranged as 
\[
M_3(z,h)= \ip{\rmD\calH_\sfC(z){+}\sfC^*h}{ \Jmac(z)
  (\rmD\calH_\sfC(z){+}\sfC^*h) }_\bfZ =0,
\]
because of the skew-symmetry of $\Jmac(z)$. 

Using the transport theory from \cite[Ch.\,8]{AmGiSa05GFMS} we see that the
projected measures $(\sfQ_n)_\#\wh\mu_\beta^{(t)}$ are constant, namely
\[
(\sfQ_n)_\# \wh\mu_\beta^{(t)} =(\sfQ_n)_\#\wh\mu_\beta^{(0)}= (\sfQ_n)_\#
\wt\mu_\beta.
\]
We can now take the limit $n\to \infty$ and use $(\sfQ_n)_\#
\wt\mu_\beta=(\sfQ_n)_\#\mu_\beta^{(t)} \to \mu_\beta^{(t)}$, see
\cite[Prop.\,8.3.3]{AmGiSa05GFMS}. Thus, we obtain $\wt\mu_\beta ^{(t)} =
\wt\mu_\beta$, which is the desired invariance. 
\end{proof}

\section{Compression and the coarse-grained system}
\label{s:dilations}

\subsection{Compressions and dilations}
\label{s:compressions-dilations}

The theory of dilations, developed in the 1950s by B\'ela Sz.-Nagy and
Foias~\cite{Szna53CEH, SznFoi70HAOH}, provides an embedding of a strongly
continuous semigroup in a unitary group. This can be interpreted as the
embedding of irreversible dynamics in reversible dynamics on a larger space,
such that the irreversible dynamics can be obtained by projection to a
subspace. For us, the `reverse' procedure is of interest, namely the
compression of a unitary group to a contraction semigroup.

While for coarse-graining compressions are the natural concept, the theory of
dilations provides crucial motivation for the choice of the model discussed in
this paper. For a given strongly continuous contraction semigroup (thought of
as a large-scale, thermodynamic description), the theory of dilations provides
the existence of a unitary semigroup which is unique up to isomorphism, and has
a representation as time shift in a suitable space (the theory of dilations is
closely related to one-dimensional scattering theory, cf.\
\cite[Cha.\,III.2]{LaxPhi89ST}). We give an explicit construction of this
dilation in Appendix~\ref{sec:Appendix-Dilations}.

An important aspect of this specific construction is that it can simultaneously
be used to define the larger space $\bfX$ and the extension of the unitary
group $(\ee^{t\JHB})_{t\in \R} $ on $\bfH$  to a strongly continuous
group $(\extG_t)_{t\in \R}$ on the space $\bfX$, as required by
Assumption~\ref{ass:HB-evol-op-ct-on-X}.

\begin{definition}[Contractions, dilations and compressions]
  \label{de:ContractDilat}
  A strongly continuous semigroup $(C_t)_{t\geq 0}$ of bounded linear operators
  on a Hilbert space $\bfG$ is called a \emph{contraction semigroup} if
  $\| C_t z\|\leq \|z\|$ for all $t\geq 0$ and $z\in \bfG$.

  A \emph{dilation} of the contraction semigroup $(C_t)_{t\geq 0}$ on $\bfG$
  is a pair $\big(\bfH,(S_t)_{t\in \R}\big)$ of a Hilbert
  space $\bfH$ that contains $\bfG$ as a closed subspace  (with the same
  scalar product)  and a strongly
  continuous unitary group $ (S_t)_{t\in \R}$ such that
  \[
    \bbP S_t \big|_\bfG = C_t \in \mafo{Lin}(\bfG;\bfG) \quad \text{for all } t
    \geq 0,
  \]
  where $\bbP\colon\bfH\to \bfG\subset \bfH$ is the orthogonal projection of
  $\bfH$ onto $\bfG$. In this case, $(C_t)_{t\geq 0}$ is called a
  \emph{compression} of the group $(S_t)_{t\in\R}$.
\end{definition}

Every contraction semigroup has a dilation; however, one cannot expect a
general unitary group to have a compression. We come back to this point in
Appendix~\ref{sec:Appendix-Dilations}.

In light of this definition, the important Compression Property~\ref{ass:dilation} states that the evolution
of the heat bath can be compressed to a contraction semigroup $C_t = \ee^{-\bbD t}$. With the additional
structure of Assumption~\ref{ass:dilation} we can say much more about the process
$Y_t(\eta_0) := \wh \bbP \extG_t\eta_0$ than in Lemma~\ref{l:props-Y}:

\begin{lemma}
  \label{l:Y-is-OU}
  Under the same assumptions as Lemma~\ref{l:props-Y}, and in addition
  Assumption~\ref{ass:dilation}, $t\mapsto Y_t(\eta_0)$ is a stationary $\bfW$-valued
  Ornstein-Uhlenbeck process with drift\/ $-\bbD Y$ and mobility $\Sigma$
  satisfying $\Sigma\Sigma^* = \tfrac1\beta (\bbD + \bbD^*)$.

That is, $Y_t = Y_t(\eta_0)$ is a stationary solution of the stochastic differential equation
\begin{equation}
\label{eqdef:Y_t}
\rmd Y_t + \bbD Y_t \rmd t = \Sigma \dd B_t,
\end{equation}
where $B$ is a standard $\bfW$-valued Brownian motion. The process $Y$ has
covariance matrix
\[
\Expectation Y_t \oti Y_s = 
\bfR(s,t) := 
\begin{cases} 
  \frac1\beta\ee^{-(t-s)\bbD} & \text{for }t\geq s, \\ 
  \frac1\beta\ee^{-(s-t)\bbD^*} & \text{for }t\leq s . 
\end{cases} 
\] 
In particular, $Y_t$ is a Gaussian process with continuous sample paths.

Similarly, if $\xi_0 \sim \wh\bbQ_\# \gamma_\beta$, then $Y(\xi_0)$ is a
(non-stationary) Ornstein-Uhlenbeck process satisfying~\eqref{eqdef:Y_t} and
$Y_0(\xi_0) = 0$ almost surely.
\end{lemma}
\begin{proof}
  For any $t_1,\dots,t_n\in \R$, the vector
  $(Y_{t_1}(\eta_0),\dots , Y_{t_n}(\eta_0))\in \bfW^n$ consists of $n$
  measurable linear maps applied the Gaussian random variable $\eta_0$, and
  therefore is a multivariate Gaussian random variable. This means that the
  process $t\mapsto Y_t(\eta_0)$ is a Gaussian process. Since $\eta_0$ is
  centered, $Y_t$ also is centered, and since $t\mapsto \extG_t\eta_0$ is
  stationary, the same holds for $Y_t$.
  
  We calculate the covariance function of $Y_t(\eta_0)$ for
  $\eta_0\sim \gamma_\beta$. First remark that by the Compression
  Property~\ref{ass:dilation}
  \[
  \bfR(s,t) = \frac1\beta \bbP \ee^{(t-s)\JHB} \bbP.
  \]
  For fixed $h_1,h_2\in \bfH$ and $s,t\in \R$ we then calculate the expectation
  with respect to $\eta_0$,
  \begin{align*}
  & \Expectation\ip{Y_t}{h_1}_H\ip{Y_s}{h_2}_H 
   = \Expectation \ip{\wh\bbP \extG_t\eta_0}{h_1}_H\ip{\wh\bbP \extG_s\eta_0}{h_2}_H
  = \Expectation \ip{ \eta_0}{\ee^{-t\JHB}\bbP h_1}_H\ip{\eta_0}{\ee^{-s\JHB}\bbP h_2}_H\\
  &\hspace{4em} \leftstackrel{\eqref{eq:formula-for-covariance}}= 
    \frac1\beta \ip{\ee^{-t\JHB}\bbP h_1}{ \ee^{-s\JHB}\bbP h_2}_H
  = \frac1\beta \ip{h_1}{\bbP \ee^{(t-s)\JHB}\bbP h_2}_H = \ip{h_1}{\bfR(s,t)h_2}_H.
  \end{align*}
  Therefore $Y$ is a stationary centered Gaussian process with covariance
  matrix $\bfR(s,t)$; this property characterizes the Ornstein-Uhlenbeck
  process with drift $-\bbD Y$ and mobility~$\Sigma$.\footnote{This follows
    from recognizing that covariance functions uniquely determine stationary
    centered Gaussian processes, that for solutions of~\eqref{eqdef:Y_t}  and
    $t>s$, we have 
  \[
  \partial_t \Expectation Y_t \oti Y_s = -\bbD \,\Expectation Y_t \oti Y_s
  \qquad\text{and}\qquad
  \partial_t \Expectation Y_s \oti Y_t = -\Expectation Y_s \oti Y_t \,\bbD^*,
  \]
  and that $\Expectation Y_0\oti Y_0 = \beta^{-1} \bbI$ since $\ee^{-\beta\|y\|^2_\bfW/2}\dd y$ is stationary. See also~\cite[Sec.\,9.1]{MulSch85LV}.}

  If $\xi_0\sim \wh\bbQ_\#\gamma_\beta$, then by choosing $w_0\in\bfW$ randomly according to $\wh\bbP_\#\gamma_\beta$, the sum $\eta_0:=w_0+\xi_0$ is distributed according to $\gamma_\beta$ (see Lemma~\ref{l:measurable-extensions}), and the remarks above apply. 
  Writing for $t\geq0$
  \[
  Y_t(\xi_0) = Y_t(\eta_0) - Y_t(w_0) = Y_t(\eta_0) - \bbP\ee^{t\JHB}w_0 
  = Y_t(\eta_0) - \ee^{-t\bbD}w_0 ,
  \]
  we calculate that
  \begin{align*}
  \rmd Y_t(\xi_0) &= \rmd Y_t(\eta_0) - \rmd \ee^{-t\bbD}w_0
  = -\bbD Y_t(\eta_0)\dd t  + \Sigma \,\dd B_t + \bbD \ee^{-t\bbD}w_0 \dd t\\
  &= -\bbD Y_t(\xi_0)\dd t  + \Sigma \,\dd B_t.
  \qedhere
  \end{align*}
\end{proof}

\subsection{Derivation of the coarse-grained equations}
\label{s:macro}

With the characterisation of the projected stochastic process $Y = Y(\eta_0)$ of Lemma~\ref{l:Y-is-OU}, we
can now derive the coarse-grained system, i.e., the evolution in $(z,w)$.

\begin{theorem}
  \label{t:main-coarse-graining-result}
Consider the setup of Section~\ref{s:micro}, including the Compression Property~\ref{ass:dilation}. 
\begin{enumerate}
\item \label{part-1:t:main-coarse-graining-result} For $z_0\in \bfZ$ and $w_0\in \bfW$ consider $(z_t,\eta_t) = \bbU_t (z_0,w_0)$ as in part~\ref{i:t:ex-un-full-ZH} of Theorem~\ref{t:ex-un-full}.
  Then the pair $(z_t,w_t:= \bbP \eta_t)$ satisfies the ODE in $\bfZ\ti\bfW$
  \begin{subequations}
    \label{eq:evol-zw-ODE}
  \begin{align}
    \dot z_t &= \Jmac(z_t) (\rmD \Hmac (z_t) + \sfC^* w_t)  \label{eq:evol-zw-z-det}\\
    \dot w_t &= -\bbD(w_t +\sfC z_t).  \label{eq:evol-zw-w-det}
  \end{align}
\end{subequations}
\item \label{part-2:t:main-coarse-graining-result} For $z_0\in \bfZ$, $w_0\in \bfW$, and $\xi_0\sim \wh\bbQ_\#\gamma_\beta$, consider $(z_t,\eta_t) = \bbU_t (z_0,w_0+\xi_0)$ as in part~\ref{i:t:ex-un-full-ZXw} of Theorem~\ref{t:ex-un-full}. Set $w_t := \wh\bbP\eta_t$.
Then the  pair $(z_t, w_t)$ satisfies the \sde
\begin{subequations}
  \label{eq:evol-zw}
  \begin{align}
    \label{eq:evol-zw-z}
    \rmd z_t &= \Jmac(z_t) (\rmD \Hmac (z_t) + \sfC^* w_t)\rmd t\\
    \rmd w_t &= -\bbD(w_t +\sfC z_t)\dd t  +  \Sigma \,\dd B_t,
    \label{eq:evol-zw-w}
  \end{align}
\end{subequations}
where  $\Sigma$ is a square root of $(\bbD+\bbD^*)/\beta$, and $(z_t,w_t)|_{t=0} = (z_0,w_0)$.  Here $B_t$ is a standard $\bfW$-valued Brownian motion.
\item \label{part-3:t:main-coarse-graining-result}
Similarly, if $(z_0,\eta_0)\sim \mu_\beta$ and $(z_t,\eta_t) = \bbU_t (z_0,\eta_0)$, then $(z_t,w_t)$ is a stationary solution of the \sde~\eqref{eq:evol-zw}. The single-time marginal is 
\begin{equation}
  \label{eqdef:nu-beta-alt}
  \nu_\beta (\rmd z\dd w):= (\pi_{z,w})_\# \mu_\beta (\rmd z\dd w) \stackrel{(*)}= \mu_{\beta,\bfZ}(\rmd z) \calN_\bfW(-\sfC z,\beta^{-1})(\rmd w) 
  \end{equation}
  (and this measure is therefore invariant under the \sde~\eqref{eq:evol-zw}).  It  can alternatively be written as 
  \begin{equation}
  \label{eqdef:nu-beta}
  \nu_\beta (\rmd z\dd w) = \frac1{Z_\beta}
  \exp \pra[\Big]{-\beta \bra[\Big]{\Hmac(z) + \tfrac12 \|w\|_{\bfH}^2 + \ip{\sfC z}{w}_\bfH}}\dd z \dd w.
  \end{equation}
\end{enumerate}
\end{theorem}

\begin{RmRemark}
From here onward the implicit probability space will be always the same as in
Lemma~\ref{l:props-Y} and Theorem~\ref{t:ex-un-full}, namely $(\bfX,\calB(\bfX)_{\gamma_\beta},\gamma_\beta)$.
\end{RmRemark}

\begin{RmRemark}
This coarse-graining approach is similar to many approaches in the Mori-Zwanzig style, but nonetheless different; we comment on this in Section~\ref{ss:relation-with-literature-1}.
\end{RmRemark}

\begin{proof}
  We first prove part~\ref{part-2:t:main-coarse-graining-result}. Since
  $\wh \sfC^* = \sfC^*\wh\bbP$, equation~\eqref{eq:evol-zw-z} follows directly
  from~\eqref{eq:EvolFullExistTheorem-ZX-1}. To show that $w$
  satisfies~\eqref{eq:evol-zw-w}, we apply $\wh\bbP$
  to~\eqref{eq:EvolFullExistTheorem-ZX-2} and use the compression property. We
  then integrate by parts to obtain
\begin{align*}
w_t &= \wh\bbP \ee^{t\JHB}\bra*{w_0+\xi_0+ \sfC z_0} - \sfC z_t + \int_0^t \bbP \ee^{(t-s)\JHB} \sfC\dot z_s\dd s\\
&= \ee^{-t\bbD} w_0 + Y_t(\xi_0) + \ee^{-t\bbD}\sfC z_0 - \sfC z_t + \int_0^t \ee^{-(t-s)\bbD}\sfC \dot z_s\dd s\\
&= \ee^{-t\bbD} w_0 + Y_t(\xi_0)  + \bbD \int_0^t \ee^{-(t-s)\bbD}\sfC  z_s\dd s,
\end{align*}
which is a mild form of the \sde
\[
\rmd w_t = -\bbD(w_t + \sfC z_t) \rmd t + \rmd Y_t (\xi_0) + \bbD Y_t(\xi_0) \rmd t.
\]
Equation~\eqref{eq:evol-zw-w} then follows from Lemma~\ref{l:Y-is-OU}.
For part~\ref{part-1:t:main-coarse-graining-result} the calculation is identical, but with the choice $\xi_0=0$.

Finally, for part~\ref{part-3:t:main-coarse-graining-result}, Th.~\ref{th:MeasInvariant} implies  that $t\mapsto (z_t,\eta_t)$ is a stationary stochastic process with single-time marginal $\mu_\beta$, and therefore $t\mapsto (z_t,w_t) = \pi_{z,w}(z_t,\eta_t)$ also is a stationary stochastic process, with single-time marginal $(\pi_{z,w})_\#\mu_\beta$. By part~\ref{part-2:t:main-coarse-graining-result}, $(z_t,w_t)$ is a solution of the SDE~\eqref{eq:evol-zw}. 

The characterization $(*)$ in~\eqref{eqdef:nu-beta-alt} follows from part~\ref{l:measurable-extensions-operators} of Lemma~\ref{l:measurable-extensions}. Taking $\bbO := \bbP$ in that lemma, we find that if $\eta\sim \calN_\bfX(0,\beta^{-1}\bbC)$, then $w := \wh\bbP \eta$ has distribution $\calN_\bfH(0,\beta^{-1}\bbP) = \calN_\bfW(0,\beta^{-1}\bbI_\bfW)$. By translating the distributions by $\sfC z$ we find the identity~$(*)$. Finally, the version~\eqref{eqdef:nu-beta} is a rewriting of~\eqref{eqdef:nu-beta-alt}.
\end{proof}

\begin{RmRemark}[Invariance of $\nu_\beta$]
The invariance of the measure in~\eqref{eqdef:nu-beta} can also be seen by explicit computation. The generator
associated with the stochastic process~\eqref{eq:evol-zw} is
\[
 \mathcal L =\Jmac (z)\bra[\big]{\rmD\Hmac (z) +  \sfC^* w }\cdot  \nabla_z   -\bbD(w +\sfC z) \cdot  \nabla_w  +  \frac1\beta \DDsym \,{:}\nabla^2_w,
\]
with adjoint operator 
\[
  \mathcal L^\star \rho = -\div_z \pra[\big]{\rho \Jmac(z)( \rmD \Hmac (z) {+} \sfC^* w)} + \div_w \pra[\big]{\rho\bbD(w
   +\sfC z )} + \frac1\beta\div_w \pra[\big]{\div _w  (\rho\DDsym)}.
\]
To show that $\nu_\beta$ is invariant, we show that $\calL^\star\nu_\beta=0$. Identifying $\nu_\beta$ with its Lebesgue density, note that 
\begin{multline*}
\nabla_z \nu_\beta = -\beta \nu_\beta (\rmD\Hmac(z) + \sfC^*w), 
\qquad
\nabla_w \nu_\beta = -\beta \nu_\beta (w{+}\sfC z), \\
\qquad\text{and}\qquad 
\frac1\beta\div_w (\nu_\beta \DDsym) = -\DDsym \nabla_w \nu_\beta.
\end{multline*}
We then calculate
\begin{align*}
\div_z \pra[\big]{\nu_\beta \Jmac(z)( \rmD \Hmac (z) {+} \sfC^* w)}
&= -\frac1\beta \div_z\pra[\big] {\Jmac(z) \nabla_z \nu_\beta}\\
&= -\frac1\beta \underbrace{\div_z\pra[\big] {\Jmac(z)}}_{\text{$=0$ by assumption}}\!\!\cdot  \nabla_z \nu_\beta
-\frac1\beta \underbrace{\Jmac(z) {:} \nabla_z^2 \nu_\beta}_{\text{$=0$ by $\Jmac^*=-\Jmac$}} =0,
\end{align*}
and 
\begin{align*}
  \div_w \pra*{\nu_\beta \bbD(w{+}\sfC z) + \frac1\beta\div_w (\nu_\beta \DDsym) }
  &=
  \div_w \pra*{-\frac1\beta  \bbD\nabla_w \nu_\beta  + \frac1\beta  \DDsym\nabla_w \nu_\beta}\\
&= -\frac1\beta \underbrace{\div_w \pra*{\DDanti \nabla_w \nu_\beta}}_{\text{$=0$ by $\DDantistar=-\DDanti$}}
= 0.
\end{align*}
Therefore $\nu_\beta$ is invariant. 

For a harmonic potential $V$, existence of an invariant measure can be shown using controllability as
in~\cite{Jaksic2017a}.
\end{RmRemark}

\section{\Generic}  
\label{s:GENERIC}

\subsection{Deterministic \Generic}
\label{ss:deterministic-GENERRIC}

The \generic framework combines dissipative and conservative evolutions  in way
that respects thermodynamic principles. An evolution equation for an unknown
$y(t)$ in a reflexive Banach space $\bfY$ is said to be in \generic form if it
can be written in terms of two functionals $\calE$ and $\calS$ and two
operators $\bbJ$ and $\bbK$ as  
\begin{equation}
  \label{eqdef:GENERIC}
\dot y = \bbJ(y)\rmD\calE(y) + \bbK(y)\rmD\calS(y).
\end{equation}
Here the functionals $\calE$, $\calS$ and operators $\bbJ$, $\bbK$ should satisfy the following (formal) requirements:
\begin{enumerate}
\item\label{cond:GENERIC-1} The `energy' $\calE$ and `entropy' $\calS$ are sufficiently smooth functionals on $\bfY$.
\item For each $y\in \bfY$, $\bbJ(y)$ and $\bbK(y)$ are  (possibly
  unbounded) linear operators from $\bfY^*$ to $\bfY$. 
\item\label{cond:GENERIC-asym-Jacobi} For each $y\in \bfY$, the Poisson operator $\bbJ(y)$ is skew-adjoint and satisfies Jacobi's identity:
 \begin{equation}
   \label{eq:JacobiIdent}
   \forall\, y \in \bfY\ \forall\, \mu_1,\mu_2,\mu_3\in \bfY^*: \quad
   \Dual{\mu_1}{ \rmD\bbJ(y)[\bbJ(y)\mu_2] \mu_3}_\bfY + \text{cyclic
     perm.} =0 ,
 \end{equation}
 where $\dual\cdot\cdot{}_\bfY$ denotes the duality paring on $\bfY^*\ti
 \bfY$ and $\rmD\bbJ(y)[v]$ is the directional derivative in the direction $v\in
 \bfY$.

\item For each $y\in\bfY$, the Onsager operator $\bbK(y)$ is symmetric and non-negative.
\item\label{cond:GENERIC-5} The following \emph{non-interaction conditions} are satisfied:
\begin{equation}
\label{eq:NIC}
\bbJ(y) \rmD \calS(y) \stackrel{\mathrm{(a)}}= 0 \qquad\text{and}\qquad
\bbK(y) \rmD\calE(y) \stackrel{\mathrm{(b)}}= 0
\qquad\text{for all }y\in \bfY.
\end{equation}
\end{enumerate}
We call the quintuple $(\bfY,\calE,\calS, \bbJ,\bbK)$ a GENERIC system if
the conditions 1.-5.\ hold. 

A pre-form of this structure was introduced by Grmela in \cite{Grme84PBFK,
  Grme85BFDT} and independently by Morrison~\cite{Morr84BFIC, Morr86PJHD} under
the name \emph{metriplectic systems}, and further developed by \"Ottinger and
Grmela as \emph{General Equation for the Non-Equilibrium
  Reversible-Irreversible Coupling}, abbreviated as
\generic~\cite{GrmOtt97DTCF1,OttGrm97DTCF2,Otti05BET}.  In \cite{Morr86PJHD}
the structures $\bbJ$ and $\bbK$ feature together with one function
$\calF=\calE-\calS$. The importance of keeping $\calE$ and $\calS$ separate and
the role of the non-interaction conditions was emphasized for the first time in
\cite{GrmOtt97DTCF1,OttGrm97DTCF2}. Grmela and \"Ottinger's treatment includes
extensions to stochastic evolutions, which are central to this paper, and we
therefore follow the \generic terminology.

\begin{RmRemark}[Hamiltonian and Poisson systems]
The property~\ref{cond:GENERIC-asym-Jacobi} above implies that $\bbJ$ generates
a Poisson manifold
% with the Poisson bracket given by~\eqref{eqdef:Poisson-bracket} 
(see e.g.~\cite[Ch.\,10]{MarRat99IMS} for the theory of Poisson manifolds).  An
evolution equation of the form $\dot y = \bbJ(y)\rmD\calE(y)$ is therefore
formally a Poisson system. This also implies that the \Generic structure
contains the class of Hamiltonian systems as a special case, corresponding to
$\bbK\rmD\calS\equiv0$.

The non-interaction condition $\bbJ\rmD\calS=0$ implies that the corresponding
Poisson structure is degenerate: the flow of any evolution equation
$\dot y = \bbJ(y)\rmD\calF(y)$ preserves the value of $\calS$; the full space
is therefore foliated into sets of the form $\{\calS = \text{constant}\}$, and
the flow preserves each of these leaves~\cite[Th.~10.4.4]{MarRat99IMS}. We will
see examples of this in Section~\ref{ss:explanation-S-part2}.
\end{RmRemark}

\begin{RmRemark}
  \label{rem:E-and-S-preserved-under-GENERIC}
  The non-interaction conditions~\eqref{eq:NIC} imply that the `energy' $\calE$
  and the `entropy'~$\calS$ have the properties that one expects from
  thermodynamics: if $y = y(t)$ satisfies~\eqref{eqdef:GENERIC} then $\calE$ is
  preserved and $\calS$ is non-decreasing:
  \begin{alignat*}3
    \frac{\rmd}{\rmd t} \calE(y(t)) &= \Dual{\rmD\calE}{\dot y} &&= \underbrace{\Dual {\rmD\calE}{\bbJ\rmD\calE}}_{=\,0
    \text{ by skew-symmetry of $\bbJ$}}
    + \quad \underbrace{\Dual{\rmD\calE}{\bbK\rmD\calS}}_{=\,0
    \text{ by \eqref{eq:NIC}(b)}}  &&=0\\
    \frac{\rmd}{\rmd t} \calS(y(t)) &= \Dual{\rmD\calS}{\dot y} &&= \quad \underbrace{\Dual {\rmD\calS}{\bbJ\rmD\calE}}_{=\,0
    \text{ by \eqref{eq:NIC}(a)}}
    + \quad \underbrace{\Dual{\rmD\calS}{\bbK\rmD\calS}}_{\geq\,0
    \text{ by non-negativity of $\bbK$}}&&\geq 0.
    \qedhere
  \end{alignat*}
\end{RmRemark}

\begin{RmRemark}[Special cases: Hamiltonian systems and gradient flows]
  We already mentioned that any Hamiltonian system is also a \Generic system, corresponding to $\bbK\rmD\calS\equiv0$. Similarly, all gradient flows are also \Generic systems, corresponding to $\bbJ\rmD\calE\equiv0$. Note that with this choice of sign the driving functional $\calS$ increases along solutions, in contrast to the more common decreasing case.
\end{RmRemark}

\begin{RmRemark}[Coordinate invariance]
  The \Generic equation~\eqref{eqdef:GENERIC} is coordinate-invariant: under linear or nonlinear changes of variables the equations retain the same form (see e.g.~\cite[Sec.~1.2.4]{Otti05BET} or~\cite[Sec.~2]{Mielke2011b}). 
\end{RmRemark}

\subsection{\Generic SDEs}
\label{ss:GENERIC-SDE}
In this paper we encounter versions of \Generic that include noise. We define a
`\generic \sde' to be a stochastic differential equation for a process $Y_t$ in
a finite-dimensional Hilbert space $\bfY$ of the form
\begin{equation}
\label{eq:GSDE-general}
\rmd Y_t = \bra[\Big]{\bbJ(Y_t)\rmD\calE(Y_t) + \bbK(Y_t)\rmD\calS(Y_t)
 + \div{\bbK}(Y_t)}\dd t 
 + \Sigma(Y_t) \dd B_t,
\end{equation}
where $B_t$ is a standard finite-dimensional Brownian motion. Note that we
restrict ourselves to finite-dimensional spaces $\bfY$ for simplicity.  We
require that the noise intensity $\Sigma$ and the dissipation operator $\bbK$
are related by a fluctuation-dissipation relation:
\begin{equation}
  \label{eq:FDR-general}
 \Sigma\Sigma^*(y)   = 2 \bbK(y) \qquad \text{for all }y.
\end{equation}
%  Here $\beta>0$ has the interpretation of inverse temperature.
Fluctuation-dissipation relations of the form~\eqref{eq:FDR-general} have a
long history; see e.g.\ the discussion
in~\cite{OttingerPeletierMontefusco21,MontefuscoPeletierOttinger21}. The
\generic \sde~\eqref{eq:GSDE-general} and the fluctuation-dissipation
relation~\eqref{eq:FDR-general} seem to appear for the first time
in~\cite{GrmOtt97DTCF1}; \"Ottinger refers to equations of the
form~\eqref{eq:GSDE-general} as `\Generic with
fluctuations'~\cite[Sec.~1.2.5]{Otti05BET}.\medskip
 
The \generic \sde inherits energy conservation from the deterministic
structure, and the fluctuation-dissipation relation~\eqref{eq:FDR-general} and
conservation condition~\eqref{eq:cond:Lebesgue-is-preserved-by-J} imply that
the measure $\ee^{\calS(y)}\dd y$ is stationary:

\begin{lemma}[Properties of a \generic \sde]
  \label{l:props-GSDE}
  Assume that $\calE$, $\calS$, $\bbJ$, and $\bbK$ satisfy the \generic
  conditions~\ref{cond:GENERIC-1}--\ref{cond:GENERIC-5} above and the
  additional condition
  \begin{equation}
      \label{eq:cond:Lebesgue-is-preserved-by-J}
      \div \bbJ = 0 \qquad \text{on }\bfY.
      \end{equation}
  In addition, assume that the equation~\eqref{eq:GSDE-general} has unique weak solutions for any given starting point $Y_0$. Then
\begin{enumerate}
  \item The evolution deterministically preserves $\calE$, i.e. for any $Y_0\in \bfY$, 
  \[
    \calE(Y_t) = \calE(Y_0) \qquad\text{almost surely, for all } t\geq0.
  \]
\item The measure 
\[
\ee^{\calS(y)} \dd y
\] 
is preserved by the flow of~\eqref{eq:GSDE-general}.
\end{enumerate}
\end{lemma}

\begin{proof}
To show the conservation of $\calE$, we calculate by It\^o's lemma that 
\begin{align*}
\rmd \calE(Y_t) &= \rmD\calE(Y_t)^*   \bra[\Big]{\bbJ(Y_t)\rmD\calE(Y_t) + \bbK(Y_t)\rmD\calS(Y_t)
+  \tfrac12 \div\pra[\big]{\Sigma\Sigma^*}(Y_t) }\dd t \\
&\qquad {} + \rmD\calE(Y_t)^*   \Sigma(Y_t) \dd B_t 
  + \frac12 \rmD^2\calE(Y_t)  \,{:}\, {\Sigma\Sigma^*(Y_t)}  \dd t \\
&\leftstackrel{(*)}= \frac12 \bra[\Big]{\rmD\calE(Y_t)^* \div\pra[\big]{\Sigma\Sigma^*}(Y_t) + \rmD^2\calE(Y_t)  \,{:}\, {\Sigma\Sigma^*(Y_t)}  }\dd t \\
&\qquad {} +  \rmD\calE(Y_t)^*    \Sigma(Y_t) \dd B_t \\
&= \frac12 \div \pra*{\rmD\calE(Y_t)^* \Sigma\Sigma^*(Y_t) }\dd t 
+ \rmD\calE(Y_t)^*   \Sigma(Y_t) \dd B_t \\
&= 0.
\end{align*}
The identity $(*)$ follows from the non-interaction conditions~\eqref{eq:NIC}, and the final two terms vanish because the range of $\Sigma$ is included in the range of $\bbK$ by~\eqref{eq:FDR-general}, which  is orthogonal to $\rmD\calE$ for the same reason.

To show that $\ee^{ \calS(y)}\rmd y$ is stationary, note that the dual generator of the process $Y_t$ in~\eqref{eq:GSDE-general} is given by 
\[
\mathcal L^*\rho  = \div\pra[\Big]{\rho\bbJ\rmD\calE + \bbK(\rho\rmD\calS- \rmD \rho)}.
\]
One checks that $\rho(y) = \ee^{ \calS(y)}$ satisfies $\mathcal L^*\rho=0$:
\begin{align}
&\div\pra*{\ee^{ \calS} {\bbJ\rmD\calE}}
  = \underbrace{ \ee^{\calS} \rmD \calS \cdot \bbJ\rmD\calE}_{=0 \text{ by }\eqref{eq:NIC}}
    + \underbrace{\ee^{\calS} \rmD \calE \cdot \div \bbJ}_{=0 \text{ by }\eqref{eq:cond:Lebesgue-is-preserved-by-J}}
    + \underbrace{\ee^{\calS} \bbJ \,{:}\, \rmD^2 \calE}_{=0 \text{ by } \bbJ^* = -\bbJ} = 0,
    \label{eq:l:mu-stat-GSDE-I}\\
&\div\pra[\big]{\bbK( \ee^{ \calS} {\rmD\calS} - \rmD\ee^{\calS})}
= 0.
\notag
\end{align}
This implies that $\ee^{\calS(y)}\rmd y$ is an invariant measure of the process~\eqref{eq:GSDE-general}.
\end{proof}

\begin{RmRemark}[Additional invariant measures] \label{rem:other-invariant-measures}
In fact, the conservation of $\calE$ by the \sde~\eqref{eq:GSDE-general} implies that there is a very large class of invariant measures. For instance, any measure of the form $f(\calE(y))\ee^{\calE(y)}\dd y$ also is conserved, for any measurable function $f\colon\R\to\R$. As another example, in Section~\ref{ss:explanation-S-part1} we construct invariant measures $\nu_\beta^{\calE_0}$ that can be interpreted as ``$\ee^\calS(y)\dd y$ restricted to a level set $\calE = \calE_0$'', which can be seen as an extreme version of $f(\calE(y))\ee^{\calE(y)}\dd y$.
\end{RmRemark}

\begin{RmRemark}[Interpretation of $\div\bbJ=0$]
\label{rem:role-of-stationary-measure-in-HamSys}
The proof of Lemma~\ref{l:props-GSDE} illustrates the role of the condition
$\div \bbJ = 0$ in \eqref{eq:cond:Lebesgue-is-preserved-by-J}.
In~\eqref{eq:l:mu-stat-GSDE-I} this condition is combined with the skew-symmetry
and non-interaction conditions to show that the measure $\ee^{\calS(y)}\rmd y$
is stationary for the Hamiltonian flow $\bbJ\rmD \calE$.

However, the condition $\div \bbJ=0$ is not coordinate-invariant; the
corresponding coordinate-invariant setup is as follows. Let $m = m(y)$ be a
smooth and positive density on $\bfY$, and assume that
\begin{equation}
  \label{eq:cond:mu-is-preserved-by-J}
  \div \bra*{m\bbJ} = 0.
\end{equation}
(The case~\eqref{eq:cond:Lebesgue-is-preserved-by-J} corresponds to
$m\equiv 1$.)  Weinstein~\cite{Weinstein97} discusses equations of this type in
the context of Poisson geometry. If such a density $m$ exists, then
$\mu(\rmd y) = m(y)\rmd y$ is invariant for each Hamiltonian flow with respect
to the Poisson structure $\bbJ$; such a Poisson manifold is called
\emph{uni-modular}.

The generalization of the \generic \sde~\eqref{eq:GSDE-general} to the case
of~\eqref{eq:cond:mu-is-preserved-by-J} is
\begin{equation}
  \label{eq:GSDE-coordinate-invariant}
  \rmd Y_t = \bra*{\bbJ(Y_t)\rmD\calE(Y_t) + \bbK(Y_t)\rmD\calS(Y_t)
 +  \frac1{ m(Y_t)}  \div\bra*{m\bbK}(Y_t)}\dd t 
 + \Sigma(Y_t) \dd B_t.
\end{equation}
This \sde conserves the measure $\ee^{ \calS(y)} m(y)\dd y$, as can be verified
in the same way as in the proof of Lemma~\ref{l:props-GSDE}.

This observation suggests the following interpretation of the interplay between
$m$, $\calS$, and the Lebesgue measure. The Poisson structure $\bbJ$ admits an
invariant measure $\mu = m \dd y$; often~$\mu$ coincides with the Lebesgue
measure. The entropy $\calS$ has the interpretation of changing the weight of
different $y$-states, thus modifying $\mu$ into $\ee^{ \calS} \mu$ (see also
Section~\ref{ss:explanation-S-part1}). The gradient term $\bbK \rmD \calS$ and
the It\^o correction term $m^{-1} \div\bra*{m\bbK}$ combine to make the measure
$\ee^{ \calS} \mu$ invariant.

The reason that the conserved measure $\mu$ often coincides with the Lebesgue
measure is related to the use of canonical coordinates in Hamiltonian
systems. In canonical coordinates in $\R^{2n}$ the space $\bfY$ is the
cotangent bundle $T^*\calM$ of the manifold $\calM = \R^n$, in which case the
Poisson structure has the canonical form ``$\rmd p \wedge \rmd q$'' and the
Lebesgue measure is invariant.

See also~\cite[(6.76) and (6.163)]{Otti05BET} for a different approach to
this question.
\end{RmRemark}

\begin{RmRemark}[Updated coordinate invariance]
  \label{rem:coord-transformation-GENERIC}
Concretely, the coordinate invariance of the modified SDE~\eqref{eq:GSDE-coordinate-invariant}  means the following.

If $Y_t$ is a solution of~\eqref{eq:GSDE-coordinate-invariant} in $\bfY = \R^n$,  and $\phi\colon\R^n\to\R^n$ is a smooth bijection, then $Z_t :=\phi(Y_t)$ is a solution of  
\begin{equation}
  \label{eq:GSDE-coordinate-invariant-hat}
  \rmd Z_t = \bra*{\wh\bbJ(Z_t)\rmD_z\wh\calE(Z_t) + \wh\bbK(Z_t)\rmD_z\wh\calS(Z_t)
  +  \frac1{ \wh m(Z_t)}  \div_z\bra*{\wh m\wh\bbK}(Z_t)}\dd t 
  + \wh\Sigma(Z_t) \dd B_t.  
\end{equation}
Here the components are defined as 
\begin{alignat}{2}
  \wh\calE(\phi(y)) &:=  \calE(y), &\qquad \wh\calS(\phi(y)) &:= \calS(y),\notag\\
  \wh \bbJ(\phi(y)) &:= \rmD\phi(y) \bbJ(y) \rmD \phi(y)^*, & \wh \bbK(\phi(y)) &:= \rmD\phi(y) \bbK(y) \rmD \phi(y)^*,  \label{eqdef:Jhat-Khat}\\
  \wh m(\phi(y)) &:= \frac 1{\det\rmD\phi(y)} m(y), &\quad\text{and}\quad \hat \Sigma(\phi(y)) &:= \rmD\phi(y) \Sigma(y).\notag
\end{alignat}
Expressions such as in~\eqref{eqdef:Jhat-Khat} should be read as 
\[
\wh \bbJ_{ij}(\phi(y)) := \sum_{k,\ell} \partial_{y_k}\phi_i(y) \bbJ_{k\ell}(y) \partial_{y_\ell} \phi_j(y).
\]
One can check that 
\begin{enumerate}
  \item $\wh m$ is again stationary in the sense of~\eqref{eq:cond:mu-is-preserved-by-J}, i.e.\ $\div_z (\wh m \wh \bbJ) = 0$;
  \item $\wh m(z)\dd z$ is the push-forward of the measure $m(y) \dd y$ under $\phi$;
  \item The measure $\wh\mu(\rmd z) := \wh m(z) \ee^{\wh\calS(z)}\dd z$ is the push-forward of $\mu(\rmd y) := m(y) \ee^{ \calS(y)} \dd y$ under~$\phi$, and is stationary for~\eqref{eq:GSDE-coordinate-invariant-hat}.
\end{enumerate}
\end{RmRemark}

\begin{RmRemark}[Reversibility and irreversibility]
  \label{re:reversibility-and-irreversibility}
  The terms `reversible' and `irreversible' sometimes give rise to
  confusion. For instance, in the acronym \generic (General Equation for the
  Non-Equilibrium Reversible-Irreversible Coupling) the word `irreversible'
  refers to the gradient-flow term $\bbK\rmD \calS$; this term causes $\calS$
  to increase along the evolution, and hence generates a form of
  irreversibility. At the same time, this gradient-flow term is associated with
  \emph{reversible} stochastic processes (see e.g.~\cite{Mielke2014a} and
  below). In this remark we attempt to resolve this confusion.

A first concept of `reversibility' comes from the theory of ordinary
differential equations. A differential equation $\dot y = f(y)$ in $\R^d$ is
called \emph{reversible} or $R$-\emph{reversible} (see~\cite{Devaney76a}
or~\cite[Ch.\,V]{Hairer2006a}) if there exists a linear involution
$R:\R^d\to\R^d$ such that $Rf(Ry) =-f(y)$.  With this property, solving the
equation forward from $y_0$ is equivalent to solving it backward from
$Ry_0$. This implies for instance that given the value of $y(t)$ one can
retrace the solution back to time $t=0$ without loss of information; in this
sense the forward solution map can be `reversed'. An important example is that
of Hamiltonian systems with Hamiltonians of the form
$H(q,p) = \frac12|p|^2 + V(q)$ with involution $R: (q,p) \mapsto (q,-p)$. The
assumptions of this paper also make Example~\ref{ex:RunningExa1} reversible,
with involution $R: (q,p,\eta) \mapsto (q,-p,\eta(-\,\cdot))$.

We now generalize this concept of reversibility by lifting it to measures on
path space. A measure $P$ on the path space $\rmC((-\infty,\infty);\R^d)$ is
defined to be $R$-\emph{reversible} if it is invariant under the combined
action of $R$ and time inversion:
\[
r _\# P = P, \qquad \text{where}\qquad r(y)(t) := Ry(-t) 
   \qquad\text{for }y\in \rmC((-\infty,\infty);\R^d)\text{ and }t\in\R.
\]
Reversible differential equations, in the first sense above, define reversible
path-space measures, provided they are in stationarity, i.e.\ provided that the
single-time marginal distributions $P_t := \mathrm{law}\, y(t)$ are independent
of $t$. In Example~\ref{ex:RunningExa1}, this is the case if we start the
evolution at any time $t_0$ with $(q,p,\eta)(t_0)$ is distributed according to
the measure~$\mu_\beta$ (see Theorem~\ref{th:MeasInvariant}).

A second important category of reversible path-space measures is generated by
stochastic processes. 
Let $\mu$ be an invariant measure of the \sde
\begin{equation}
  \label{eq:general-GENERIC-no-Ham-part}
  \rmd Y_t = \bbK(Y_t) \rmD \calS (Y_t)\dd t + \div\bbK (Y_t) \dd t + \Sigma(Y_t) \dd W_t
  \qquad \text{(i.e.~\eqref{eq:GSDE-general} without $\bbJ\rmD \calE$)} .
\end{equation}
If $\bbK \rmD(\calS-\log \rmd\mu/\rmd\Lebesgue)=0$ and $2\bbK = \Sigma\Sigma^*$ (for instance, $\mu = \ee^{\calS-\beta \calE}$, see Remark~\ref{rem:other-invariant-measures}), then this \sde 
generates a measure $P$ on path space $\rmC((-\infty,\infty);\R^d)$, that  is invariant under time inversion with
$R = \mathrm{id}$.  (As a general rule, a stationary Markov process with
single-time marginal $\mu$ generates an $R$-reversible path-space measure if
and only if
(a) $R_\#\mu = \mu$ and (b) $\calL^*f = (\calL(f{\circ} R)){\circ} R$ for any
$f$, where $\calL^*$ is the Hilbert adjoint of the generator $\calL$ in
$L^2(\mu)$.)

The {reversibility} of the stochastic process~\eqref{eq:general-GENERIC-no-Ham-part} is generated by the balance of two opposing effects. The two separate equations
\[
  \rmd Y_t = \bbK \rmD \calS (Y_t)\dd t 
  \qquad\text{and}\qquad 
  \rmd Y_t =  \Sigma \dd W_t
\]
drive $Y_t$ in different directions:
the dissipative, `concentrating' effect of $\bbK \rmD \calS $ and the `spreading' effect of $\Sigma \dd W_t$ together result in stationarity, and because of the fluctuation-dissipation relation~\eqref{eq:FDR-general} the resulting path-space measure also is reversible. 
In other words, the term $\bbK\rmD \calS$ is at the same time a drift with irreversible effects and one half of a fluctuation-dissipation pair that generates a reversible measure.

In particular, 
for the multivariate Ornstein-Uhlenbeck process given by 
\[
\rmd Y_t = - \bbD Y_t \dd t + \Sigma \dd W \quad \text{ with } \Sigma = \Big(
\frac1\beta(\bbD{+}\bbD^*)\Big)^{1/2}
\]
the equilibrium measure $P$ is Gaussian and characterized by the covariance
$\bbE^P(Y_t\ti Y_s)= \frac1\beta \ee^{-(t-s)\bbD}$ for $t\geq s$, we have
reversibility if and only if $\bbD= \bbD^*$, i.e., $\DDsym=\bbD$ or
$\DDanti = 0$. Thus, from the stochastic point of view we may consider that
$\DDsym$ generates the reversible part of the process and $\DDanti$ the
irreversible part.\medskip

Turning back to the thermodynamical terms `reversible' and
`irreversible' in the acronym \generic, the `reversibility' of the Hamiltonian
term $\bbJ\rmD \calE$ is reversibility in the first sense above, under some
involution $R$. (Note that not all Hamiltonian systems are reversible, and
therefore this usage is inspired by those important examples that do have this
property). The `irreversibility' of the term $\bbK\rmD \calS$ is the thermodynamic
irreversibility of the equation $\rmd Y_t = \bbK \rmD \calS (Y_t)\dd t $, which
causes the entropy $\calS$ to increase.

We recall that the stochastically reversible part $\DDsym$ of $\bbD$
contributes to the thermodynamically irreversible part by entering in
$\bbK_\GEN$. Vice versa, the stochastically irreversible part $\DDanti$
contributes to the thermodynamically reversible part by entering in
$\bbJ_\GEN$.
\end{RmRemark}

\subsection{Reformulation in \generic form}
\label{ss:reformulation-as-GENERIC-SDE}
We now re-interpret the ODE~\eqref{eq:evol-zw-ODE} and the \sde~\eqref{eq:evol-zw} as \generic equations as introduced in the previous sections; for this we need to enlarge the state space.

We saw above that in \generic ODEs and SDEs the energy $\calE$ is (almost surely) constant. The total energy $\calH_{\mathrm{total}}$ in~\eqref{eqdef:Ham-full-1} is a natural candidate for such an energy, but it has two drawbacks: first, it depends on $\eta$, to which we do not have access after coarse-graining, and secondly, in the case of the \sde the energy $\calH_{\mathrm{total}}$ is almost surely infinite (see Remark~\ref{rem:pos-temp-means-inf-energy}). 

As described in the introduction, we use the orthogonal split of the heat-bath variable $\eta = \wh\bbP \eta + \wh \bbQ\eta = w + \wh \bbQ \eta$ to split the heat-bath energy $\tfrac12\|\eta\|^2_\bfH$  into two parts  $\frac12 \|w\|_\bfH^2 + \frac12 \|\wh\bbQ \eta\|_\bfH^2$. We can then formally write the total energy as
\begin{equation}
\label{char:Htot-Hzw-Q}  \calH_{\mathrm{total}}(z,\eta) =  \underbrace{\Hzw(z,w)}_{\text{visible part}}
  \ \ + \!\!\!\!\underbrace{\tfrac12 \|\wh\bbQ\eta\|_\bfH^2\ ,}_{\text{invisible, heat bath part}}
\end{equation}
where we set
\[
\Hzw(z,w) := \Hmac(z)  + \ip {\sfC z}{w}_\bfH + \tfrac12\|w\|_\bfH^2.
\]
This split suggests introducing a new scalar variable $e = e(t)\in \R$ that keeps track of the energy that is exchanged between the macroscopic part and the invisible part of the heat bath:
\begin{equation}
\label{eqdef:et}
\rmd e_t = 
-\rmd \Hzw(z_t,w_t), 
\qquad \text{and}\qquad e(t=0)=e_0.
\end{equation}
This setup provides a \generic energy $\calE_{\GEN} := \calH_{zw} + e$ that is constant in time, and also solves the problem of the `infinitely-large energy' in the heat bath: the new variable $e$ tracks the  \emph{changes} in the `infinitely-large'  term $\tfrac12 \|\wh\bbQ\eta\|_\bfH^2$, while being finite itself. (We work out this point of view in more detail in Section~\ref{ss:explanation-S-part2} below.)

Using It\^o's lemma to express the right-hand side in~\eqref{eqdef:et} leads to the following system of equations in the three unknowns $(z,w,e)$. This is one of the main results of this paper.
\begin{theorem}
\label{th:recognize-as-GENERIC}
Consider the function $t\mapsto (z_t,w_t)$ obtained from any one of the three parts of Theorem~\ref{t:main-coarse-graining-result}, extended by~\eqref{eqdef:et} to $(z_t,w_t,e_t)\in \bfZ\ti\bfW\ti\R$. 
\begin{enumerate}
\item In the case of part~\ref{part-1:t:main-coarse-graining-result} of Theorem~\ref{t:main-coarse-graining-result}, $(z_t,w_t,e_t)$ satisfies the ODE
\begin{subequations}
  \label{eq:evol-zwe-ODE}
  \begin{align}
  \dot z_t &= \Jmac (\rmD \Hmac (z_t) + \sfC^* w_t),
  \label{eq:evol-zwe1-ODE}
  \\
  \dot w_t &= -\bbD(w_t {+}\sfC z_t)  ,
  \\
  \dot e_t &= \ip{\bbD(w_t{+}\sfC z_t)}{w_t{+}\sfC z_t}_\bfH ,
  \end{align}
\end{subequations}
with initial datum $(z,w,e)(t=0) = (z_0,w_0,e_0)$. These equations can be written as the \Generic ODE
\[
  \frac{\rmd }{\rmd t} \begin{pmatrix} z\\w\\e\end{pmatrix}
  = \bbJ_\GEN(z) \rmD \calE_\GEN(z,w,e) 
  + \bbK_\GEN(z,w)\rmD \calS_\GEN(z,w,e)   .
\]
\item  
In the case of part~\ref{part-2:t:main-coarse-graining-result} or~\ref{part-3:t:main-coarse-graining-result} of Theorem~\ref{t:main-coarse-graining-result}, $(z_t,w_t,e_t)$ satisfies the SDE
\begin{subequations}
  \label{eq:evol-zwe}
  \begin{align}
  \rmd z_t &= \Jmac (\rmD \Hmac (z_t) + \sfC^* w_t)\rmd t ,\label{eq:evol-zwe1}
  \\
  \rmd w_t &= -\bbD(w_t {+}\sfC z_t)\rmd t  + \Sigma\,\rmd B_t, \label{eq:evol-zwe2}
  \\
  \rmd e_t &= \Big[\ip{\bbD(w_t{+}\sfC z_t)}{w_t{+}\sfC z_t}_\bfH 
      - \frac1\beta \trace \bbD \Big] \rmd t
     - {\ip {w_t{+}\sfC z_t}{\Sigma \dd B_t}_\bfH},
  \end{align}
\end{subequations}
with deterministic initial datum $(z,w,e)(t=0) = (z_0,w_0,e_0)$.
Here $B$ is a standard $\bfW$-valued Brownian motion and $\Sigma\Sigma^*  = \beta^{-1}(\bbD + \bbD^*)$. 

These equations can be written as the \generic SDE
\begin{multline}
\rmd \begin{pmatrix} z\\w\\e\end{pmatrix}
= \pra[\Big]{\bbJ_\GEN(z) \rmD \calE_\GEN(z,w,e) 
+ \bbK_\GEN(z,w)\rmD \calS_\GEN(z,w,e)  
  + \div \bbK_\GEN(z,w)}\dd t \\
+ \bbSigma_\GEN(z,w) \dd B_t.
\label{eq:GENERIC-SDE1}
\end{multline}
\end{enumerate}

In all cases the various \generic components are defined as 
\begin{subequations}
	\label{eqdef:GEN1}
\begin{align}
\calE_\GEN(z,w,e)&= \Hmac(z) + \frac12\|w\|_\bfH^2 + \ip {\sfC z}{w}_\bfH + e\\
&= \Hzw(z,w) + e,\\
\calS_\GEN(z,w,e) &= \beta  e , \label{eqdef:GEN1:S}\\
\bbJ_\GEN  (z)&= \bma{ccc} \Jmac (z)&0&0\\0 &-\DDanti&0 \\0 &0&0\ema, \label{eqdef:GEN1:J}\\
\bbK_\GEN(z,w) 
&= \frac1\beta \bma{ccc} 
0&0&0\\ 
0&\DDsym & -\DDsym(w{+}\sfC z) \\
0&-{\ip{\DDsym(w{+}\sfC z)}{ \Box}_\bfH} &\ip{\DDsym(w{+}\sfC z)}{w{+} \sfC z}_\bfH 
\ema,
\end{align} 
\end{subequations}
where the operator $\bbD$ is decomposed into symmetric and skew-symmetric parts:
\[
    \bbD= \DDsym+\DDanti \quad \text{with } \DDsym=\frac12(\bbD{+}\bbD^*) \ \text{
      and } \ \DDanti:=\frac12(\bbD{-}\bbD^*).
\]
The noise intensity in~\eqref{eq:GENERIC-SDE1} is given as 
\begin{equation*}
\bbSigma_\GEN(z,w) := 
\begin{pmatrix}
	0 \\ \Sigma \\ -{\ip{w{+}\sfC z}{ \Sigma\,\square}_\bfH}
\end{pmatrix},
\end{equation*}
which satisfies the fluctuation-dissipation  relation
\begin{equation}
\label{eq:FDR-GEN}
\bbSigma_\GEN\bbSigma_\GEN^* = 2 \bbK_\GEN.
\end{equation}
\end{theorem}

\begin{RmRemark}
It is remarkable that  the splitting $\bbD=\DDsym + \DDanti$
is compatible with the non-interaction conditions of \generic. Moreover, taking
into account the special structure of the 
differentials $\rmD \calH_\GEN=(\rmD\Hmac(z)+ \sfC^*z,w{+}\sfC z, 1)^\top$ and
$\rmD\calS_\GEN=(0,0,\beta)^\top$ and the structure of $\bbK_\GEN$ we see that in
the differential equation~\eqref{eq:evol-zwe2} only the sum
$(\DDsym{+}\DDanti)(w{+}\sfC z)$ appears. However, in the stochastic forcing we see only the symmetric part $\DDsym$ through
\[
\Sigma\Sigma^* = \frac1{\beta} \big( \bbD{+}\bbD^*\big)= \frac2\beta \,\DDsym.
\]
\end{RmRemark}

\begin{RmExample}[Running example, part 4]
\label{ex:RunningExa4}
For the model in Examples \eqref{ex:RunningExa1} and \eqref{ex:RunningExa2},
where we identify $\bfW$ with $\R^3$ we obtain an even more specific form of 
the \generic Poisson and Onsager structure in $\R^{(2n+4)\ti(2n+4)}$: 
\[
\bbJ_\GEN=\bma{cccc} 0&\bbI&0&0 \\ \!\!{-}\bbI&0&0&0\\ 0&0&\!\!-\DDanti&0 \\
0&0&0&0 \ema \  \text{ and} \quad 
\bbK_\GEN= \frac1\beta \bma{cccc}  0&0&0&0\\ 0&0&0&0 \\
0&0&\DDsym &\!\!\!-\DDsym(w{+}\sfK q)\!\!\\ 0&0&\!\!{-}\big(\DDsym(w{+}\sfK
q)\big)^{\!\top}\!\!\!& | w{+}\sfK q|^2_{\DDsym}\ema 
\]
where $\sfK \in \R^{3\ti n}$ and $\DDsym,\,\DDanti \in \R^{3\ti 3}$ have the
explicit form 
\[
\DDanti =  \bma{ccc} 0&0&0 \\ 0&0&\!{-}\varsigma \!\! \\ 0&\varsigma&
0\ema , \quad 
\DDsym = \bma{ccc}\vartheta_1&0&0 \\ 0&\vartheta_2&0 \\
0&0&\vartheta_2 \ema, \quad 
\text{and } |\wh w|^2_{\DDsym} = \wh w {\cdot} \DDsym \wh w. 
\]

Thus, the macroscopic ODE takes the explicit form 
\[
\frac\rmd{\rmd t} \bma{c} q\\ p \\ w\\ e\ema 
 = \bma{c} p \\ \!\! - \rmD V(q) - \sfK^* w \! \\ 
        -\DDanti(w{+}\sfK q)\\ 0\ema
   + \bma{c} 0\\ 0 \\ \!\!{-}\DDsym(w{+}\sfK q)\! \\ 
            | w{+}\sfK q|^2_{\DDsym} \ema
=  \bma{c} p \\ - \rmD V(q) -\sfK^* w\\ 
      - \bbD(w{+}\sfK q)\\ | w{+}\sfK q|^2_{\DDsym} \ema . 
\]
\end{RmExample}

\subsection{Interpretation of $\calE_\GEN$, $\bbJ_\GEN$, and $\bbK_\GEN$}
\label{ss:GENERIC-interpretation-EJK}

\paragraph{Interpretation of the energy $\calE_\GEN$.}
The discussion in Section~\ref{ss:reformulation-as-GENERIC-SDE} gives a clear interpretation of the energy
\[
  \calE_\GEN(z,w,e) = \Hmac(z)  + \ip{\sfC z}{w}_\bfH + \frac12 \|w\|_\bfH^2+ e 
\] 
as a renormalized version of $\calH_{\mathrm{total}}$: $\frac12 \|w\|_{\bfH}^2$
represents the `visible' part of the heat-bath energy, $e$ represents changes
in the `invisible' part, and an additional `infinite but constant' component of
the heat-bath energy has been renormalized out. This interpretation also
explains why $\calE_\GEN$ should be preserved by the \generic
\sde~\eqref{eq:GENERIC-SDE1}.

\paragraph{Interpretation of the Poisson operator $\bbJ_\GEN$.} The Poisson
operator $\bbJ_\GEN$ is a combination of the symplectic operator $\bbJ_{\calA}$
with the skew-symmetric part $\DDanti$ of the compression operator~$\bbD$. This
shows how a skew-symmetric component of the compression dynamics should be
considered as a conservative (Hamiltonian) component of the evolution.
Clearly $\bbJ_\GEN(y)$ is skew-symmetric, and it can be easily shown that the
Jacobi identity~\eqref{eq:JacobiIdent} holds, because $\bbJ_\calA$ depends
on $z$ but not on $(w,e)$.

  It would be interesting to understand in what sense or in which cases the formula~\eqref{eqdef:GEN1:J} for
  $\bbJ_\GEN$ coincides with \"Ottinger's 
  coarse-graining formula in~\cite[Eq.\,(6.67+68)]{Otti05BET}:
  \begin{equation} %\tag{6.68}
  \label{eq:Ottinger-def-CG-Poisson-operator}
  \bbJ(z,w,e) = \Big\langle \rmD \Pi(z,\eta) \bbJ_\mafo{micro}(z,\eta) \rmD
  \Pi(z,\eta)^* \Big\rangle_{(z,w,e)} ,
  \end{equation} 
  where $\langle \cdot \rangle_{(z,w,e)}$ means averaging with respect to the conditional
  equilibrium measure given the macro-state $(z,w,e)$.

  In our case, omitting the energy variable $e$, the projection mapping reads
  \[
  \begin{pmatrix} z\\ w\end{pmatrix} = \wt\Pi(z,\eta) = \begin{pmatrix} z\\ \bbP \eta\end{pmatrix}.  
  \]
  Hence, \"Ottinger's formula would take the form 
  \[
  \bbJ_\text{macro}(z,w)= \skp*{\rmD \wt\Pi(z,\eta) \begin{pmatrix} \bbJ_\calA(z)&0\\
  0&\JHB\end{pmatrix} \rmD \wt\Pi(z,\eta) }_{(z,w)}
  =  \begin{pmatrix} \bbJ_\calA(z)&0\\ 0&\text{``}\bbP\,\JHB\bbP\text{''}\end{pmatrix}. 
  \]
  Here, averaging is not needed, because the integrand is constant.  
  However, as discussed in Remark~\ref{re:Dom.JB.bfW}, the term $\bbP\,\JHB\bbP$ is not well defined, because
  $\bfW \cap \mathrm{dom}(\JHB) =\{0\}$ by construction.  
  
  Thus this formula should be considered in a renormalized sense, i.e.\ by
  approximating~$\JHB$ suitably. Because of the generator property, we have
  the approximation property by a symmetric difference of going forward and
  backward in time (which means extracting irreversible information):
  \[
  \forall \, \eta \in \mafo{dom}(\JHB): \quad \frac1{2h}\big(
  \ee^{h\JHB} - \ee^{-h\JHB}\big) \eta 
   = \frac1{2h}\big( \ee^{h\JHB}{-}\bbI\big)\eta
     - \frac1{2h}\big(\ee^{-h\JHB}{-}\bbI\big) \eta
   \;\stackrel{h\downarrow 0}\longrightarrow \;\JHB \eta. 
  \]
  At the same time, by the Compression Property~\eqref{eqass:dilation} we obtain a
  different behavior when applying $\bbP$ from left and right (for $h\searrow 0$):
  \[
  \forall \, v\in \bfW:\quad \bbP \frac1{2h}\big(
  \ee^{h\JHB} - \ee^{-h\JHB} \big) \bbP\, v
  \overset{\text{comp}}=  
  \frac1{2h}\big(\ee^{-h \bbD} {-} \bbI\big) v 
  -\frac1{2h}\big(\ee^{-h\bbD^*} {-} \bbI\big) v \longrightarrow - \bbD_\mafo{skw} v.
  \]
  Thus, although this approximation is rather special (because of the symmetry
  with $\pm h$), we see that it is reasonable to replace the ill-defined
  expression $\bbP\!\;\JHB \bbP$ by $-\bbD_\mathrm{skw}$. From this point of view it may be reasonable to interpret the definition~\eqref{eqdef:GEN1:J} of $\bbJ_\GEN$ as a version of~\eqref{eq:Ottinger-def-CG-Poisson-operator}. 

\paragraph{Interpretation of the Onsager operator $\bbK_\GEN$.}
The structure of the Onsager operator $\bbK_\GEN$ can be understood in the following way. To start with, note that in the absence of interaction with $z$ or $e$, the component $w = \wh\bbP\eta$ decays  as $\ee^{-t\bbD} w$, which implies that the norm $\|w\|^2_\bfH$ then decays with rate  $2\DDsym$:
\[
\frac{\rmd}{\rmd t}  \|w_t\|_\bfH^2 
= -\ip{w_t }{\bbD w_t}_\bfH  - \ip{\bbD w_t }{w_t}_\bfH = -2\ip{w_t}{\DDsym w_t}_\bfH.
\]
This explains the operator $\DDsym$ in the $(w,w)$-block of $\bbK_\GEN$. The $(w,e)$-component of $\bbK_\GEN$ then necessarily has the form $-\DDsym(w{+}\sfC z)$ because of the non-interaction condition~\eqref{eq:NIC}, and the remaining components of $\bbK_\GEN$ follow by symmetry. 

In fact, in equation~\eqref{eq:evol-zwe2} the drift term  is generated by the $(w,e)$-component of $\bbK_\GEN$, not the $(w,w)$-component, since $\bbK_\GEN$ multiplies $\rmD\calS_\GEN = (0,0,\beta )^\top$. This does not change the interpretation above, since the non-interaction condition implies that the same drift term also would be generated by an alternative entropy, $(z,w,e) \mapsto -\calH_{z,w}(z,w)$. We discuss the role of different entropies in the next section.

\subsection{Interpretation of $\calS_\GEN$}
\label{ss:explanation-S-part1}

For the interpretation of $\calS_\GEN$ we start by commenting on the historical development of `entropy'.
Entropy is commonly associated with a description of a system at two levels of aggregation, a `microscopic' level and a `macroscopic' level. Different authors, however, have emphasized different ways of placing entropy in this context. Boltzmann's famous postulate 
\begin{equation}
  \label{eq:Boltzmann-grave}
  \calS(x) = k \ln W(x)
\end{equation}
already illustrates this: the symbol $W(x)$ could either represent probability \emph{(Wahrscheinlichkeit)}~\cite[p.~44]{Brenig75}, or alternatively `the number of microstates associated with the macro\-state~$x$'. Boltzmann presumably based his discussion on equal probability of microstates, in which case these two interpretations reduce to the same thing. 

Einstein~\cite{Einstein10} postulated that by rewriting~\eqref{eq:Boltzmann-grave} as 
\begin{equation}
  \label{eq:Einstein-prob-S}
\Prob(x) = \ee^{ \frac1k \calS(x)},
\end{equation}
this expression could be abstracted away from the specific microscopic
situation, and could therefore hold also in situations without equal
probability of microstates; in fact, the probability distribution over
microstates only needs to be stationary under the microscopic dynamics. This
approach focuses on an interpretation of entropy as `logarithm of
probabilities' rather than as `number of microstates'. We take the same view
here.

\bigskip
In continuous state spaces an expression such as~\eqref{eq:Einstein-prob-S} should be interpreted as a density with respect to some fixed measure. We  illustrate this on a simple example:  in the case of SDEs of the form
\[
\rmd X_t = -K(X_t)\nabla V(X_t)  \dd t + \Sigma(X_t) \dd B_t,
\qquad
\Sigma\Sigma^* = 2 K,
\]
one can consider $-V$ to be  an entropy in the sense above, because the flow preserves the measure $\ee^{- V(x)}\rmd x$. The appearance of the Lebesgue measure `${\rmd x}$' is explained by observing that this Lebesgue measure is invariant under the Brownian motion $B_t$, and the function $-V$ converts this invariant measure into $\ee^{-V(x)}\rmd x$ through the pair $(K,\Sigma)$ as described by the Girsanov theorem. In this sense the entropy can be understood as describing an exponential \emph{modification} of an underlying measure, and that is the sense in which~\eqref{eq:Einstein-prob-S} should be interpreted.

This point of view provides the  SDE~\eqref{eq:evol-zw} for $(z,w)$ (i.e.\ before extending with $e$) with a natural entropy
\begin{equation}
\label{eqdef:S-for-zw}
  \calS(z,w):=  \log \frac{\rmd \nu_\beta}{\rmd\Lebesgue_{\bfZ\ti \bfW}}(z,w)
  \stackrel{\eqref{eqdef:nu-beta}} = -\beta \Hzw(z,w) + \text{constant},
\end{equation}
where $\nu_\beta$ is given in~\eqref{eqdef:nu-beta} and $\Lebesgue_{\bfZ\ti\bfW}$ is the Lebesgue measure on $\bfZ\ti\bfW$.

Below we follow this entropy $\calS$  as it propagates into the extended system with variables $(z,w,e)$, but that propagation is akin to bookkeeping. The actual appearance of entropy has happened here, at the coarse-graining step from variables $(z,\eta)$ to variables $(z,w)$.

\bigskip

While~\eqref{eqdef:S-for-zw} gives a clear suggestion for an entropy for the variables $(z,w)$, an indeterminacy arises when extending  to variables $(z,w,e)$. Indeed, after extension there is a multitude of invariant measures: for instance, if we fix an energy level $\calE_0\in \R$, then under the extension map
\[
T^{\calE_0}: (z,w) \mapsto \bra[\big]{z,w,e := \calE_0 - \Hzw(z,w)}
\]
the measure $\nu_\beta$ is mapped to the measure
\[
\nu_\beta^{\calE_0} := T^{\calE_0}_\#\nu_\beta, 
\]
which is concentrated on the `$\calE_0$-energy-leaf' $\{(z,w,e)\in \bfZ\ti\bfH\ti\R: \calE_\GEN(z,w,e) = \calE_0\}$. The measure $\nu_\beta^{\calE_0}$ is invariant under the flow of the extended \sde~\eqref{eq:evol-zwe}, for each choice of $\calE_0$. In fact, for any measure $\zeta\in\ProbMeas(\R)$ on the set of energy values $\calE_0$, the compound measure 
\[
\nu_\beta^\zeta(A) := \int_\R \nu_\beta^{\calE_0}(A)\, \zeta(\rmd \calE_0)
\qquad \text{for }A\subset \bfZ\ti\bfW\ti\R,
\]
has energy values $\calE_\GEN$ distributed according to $\zeta$ and is  invariant under the \sde~\eqref{eq:evol-zwe}.

An equivalent way of observing the same indeterminacy is the following. If $ \nu_\beta^\zeta$ is one of the invariant measures that we just constructed, then there is a family of other invariant measures of the form
\begin{equation}
\label{eq:family-nu-beta-f}
f(\calE_\GEN(z,w,e)) \,  \nu_\beta ^\zeta (\rmd z \dd w \dd e), \qquad \text{for arbitrary }f\colon\R\to\R.
\end{equation}
Indeed, $f{\circ} \calE_\GEN\,\nu_\beta ^\zeta = \nu_\beta ^{\zeta_f}$ with $\zeta_f(\rmd \tilde e) = f(\tilde e)\zeta(\rmd \tilde e)$.

\bigskip

This non-uniqueness of the invariant measure generates a corresponding non-uniqueness in the entropy. The family~\eqref{eq:family-nu-beta-f}  generates the corresponding family of `entropies'
\begin{equation}
\label{eq:family-S-g}
\calS_g (z,w,e) :=   
-\beta \Hzw(z,w) + 
g(\calE_\GEN(z,w,e)),
\quad \text{for arbitrary }g\colon\R\to\R.
\end{equation}
The choice $\calS_\GEN(z,w,e) := \beta e$ of~\eqref{eqdef:GEN1:S} is one member of this family, corresponding to $g(\tilde e) = \beta \tilde e$. 

Although~\eqref{eq:family-S-g} shows that many functions could be called `entropy', we claim that $\calS_\GEN(z,w,e) := e$ is the one that should appear in a \generic structure, and we give two reasons for this. The first reason is that the \generic structure requires $\bbJ_\GEN \rmD \calS_\GEN = 0$ (see~\eqref{eq:NIC}); within the family~\eqref{eq:family-S-g}, only ``$\beta e + \text{constant}$'' satisfies this. The second reason takes more time to explain, and we do this in the next section.

\subsection{Characterization of entropy through finite-dimensional approximations}
\label{ss:explanation-S-part2}

In this section we give an independent explanation why the expression $\calS(z,w,e) = \beta e$ should characterize the entropy in the \generic \sde~\eqref{eq:GENERIC-SDE1}. The origin of this expression lies in the scaling of `surface area', and we now describe this in detail. 

This interpretation of entropy as `surface area' can be found in some form or other in many treatments of the connection between statistical mechanics and thermodynamics (see e.g.~\cite[Sec.\,3.2]{Chandler87}, \cite[p.\,69]{Chorin94}, \cite[Sec.\,1.5]{Berdichevsky1997a}, or~\cite[(2.1.6)]{FrenkelSmit01}). While the particular dependence of the surface area on $n$, $\beta$, and $e$ given by Lemma~\ref{l:Sn} below must be generally known, we could not find a reference, and therefore we provide the details here. The same holds for the full-space equivalence-of-ensembles result of Lemma~\ref{l:conv-mu-beta-n}.

% NB Berdichevsky (end of Sec 1.5) refers to Gibbs and Paul Hertz for the relation between phase space volume and entropy. Hertz seems to be Hertz, P. (1910). Über die mechanischen Grundlagen der Thermodynamik. Annalen Der Physik, 338(12), 225–274. doi:10.1002/andp.19103381202; at the same time, while the word Entropy does appear in that paper, I can't find any definition of entropy.

\medskip

We choose $(\Xi_n)_{n\geq1}$ to be a particular increasing sequence of  $n$-dimensional subspaces of the orthogonal complement $\bfW^\perp \subset \bfH$, such that 
\[
\bfW^\perp = \overline{\Span \bigcup_{n\geq 1}\Xi_n}^{\|\cdot\|_\bfH}.  
\]
The spaces $\Xi_n$ are in the `invisible' part of the heat bath, since $\bbP\Xi_n = 0$ and $\bbQ\Xi_n = \Xi_n$. 
We describe this choice in more detail in Appendix~\ref{app:microcanonical-measure}.

Fix $\beta>0$ and $e\in \R$, and for each $n$ let $\zeta_{\beta,n,e} $ be the (non-normalized)  measure on $\Xi_n$ defined by
\begin{equation}
\label{eqdef:zeta-n-beta-e}
\zeta_{\beta,n,e} (\rmd \xi) := \pdelta\pra*{\frac12\|\xi\|_\bfH^2 - \frac n{2\beta} - e}(\dd \xi) ,
\end{equation}
where $\pdelta$ is the measure that is heuristically defined by (see Appendix~\ref{app:microcanonical-measure})
\[
\int_{\Xi_n} \varphi(\xi) \pdelta\pra*{f(\xi) - e}(\rmd \xi) := 
  \lim_{h\downarrow 0} \frac1h \int_{\Xi_n} \varphi(\xi) \One\big\{e \leq f(\xi) < e+h\big\}\dd \xi
  \qquad \text{for all }\varphi\in C_c(\Xi_n).
\]
The total mass of the measure $\pdelta\pra*{f(\xi)-e} = \pdelta\pra*{f(\xi)-e}(\rmd \xi)$ can be interpreted as the area of the surface $f^{-1}(e)\subset \Xi_n$ (weighted by $|\nabla f|$; see Appendix~\ref{app:microcanonical-measure}).

In the thermodynamical or statistical-mechanical literature the  measure $\zeta_{\beta,n,e}$ would be called `microcanonical'. If the expression $\frac12 \|\xi\|_\bfH^2$ is interpreted as an energy, as in the case of this paper, then this measure is concentrated on an `energy shell' at energy level $n/2\beta + e$. This particular choice of energy level is geared towards the limit $n\to\infty$:  this choice causes each of the $n$  degrees of freedom in  $\Xi_n$ to have approximately energy level $1/2\beta$. 

The following two properties of this sequence of microcanonical measures probably are known (see e.g.~\cite[p.~69]{Chorin94} for a calculation similar to the first one), but for neither we could find a statement that suits our purposes;  for completeness we give a proof in Appendix~\ref{app:microcanonical-measure}. 
% Recall that $\bbQ$ is the orthogonal projector in $\bfH$ onto $\bfW^\perp$, and $\wh \bbQ$ is its measurable linear extension. 
\begin{lemma}
  \label{l:Sn}
  Fix $\beta>0$ and $e\in\R$.
  \begin{enumerate}
  \item\label{l:Sn:part1} There exists $C(\beta,n)$, converging to $\infty$ as $n\to\infty$, such that 
  \begin{equation*}
    % \label{def:partition-function-zeta}
   \sfZ_{\beta,n,e} :=  \int_{\Xi_n} \zeta_{\beta,n,e}(\rmd x)   
  \end{equation*}
  has the asymptotic behavior
  \begin{equation}
    \label{conv:partition-function-zeta}
    \log \sfZ_{\beta,n,e}- C(\beta,n)\; \xrightarrow{n\to\infty }\; \beta e.
  \end{equation}  
\item\label{l:Sn:part2} We have the narrow convergence of measures on $\bfX$ 
\begin{equation}
\label{conv:zeta_beta_n_e-to-Gaussian}
\frac1{\sfZ_{\beta,n,e}}\zeta_{\beta,n,e} \longrightharpoonup \wh \bbQ_\# \gamma_\beta \qquad \text{as } n\to\infty.
\end{equation}
\end{enumerate}
\end{lemma}
\noindent
The first statement above shows that for fixed $\beta$ and $e$, 
\[
\text{surface area scales as}\quad \ee^{C(\beta,n) + \beta e}  \qquad \text{as $n\to\infty$.}
\]
Note how the largest contribution $C(\beta,n)$ diverges as $n\to\infty$, and the dependence on $e$ remains as an $O(1)$ contribution inside the exponential. 

The second statement above is a version of the `equivalence of ensembles': it can be interpreted as stating that in the large-$n$ limit the microcanonical `ensemble' $\zeta_{\beta,n,e}$ is similar to the `canonical ensemble' described by the Gaussian measure $\wh \bbQ_\# \gamma_\beta$ (see Lemma~\ref{l:equivalence-of-ensembles}). It differs from the usual treatment of equivalence-of-ensemble results in that it describes the behavior of the whole measure, not of a finite-marginal reduction. Note that under the limit Gaussian measure $ \wh \bbQ_\# \gamma_\beta$ each scalar degree of freedom has variance $1/\beta$, or equivalently energy $1/2\beta$, thus matching the scaling choice described above. 

For the discussion here it is relevant that the dependence on $e$ in $\zeta_{\beta,n,e}$ is lost in the convergence~\eqref{conv:zeta_beta_n_e-to-Gaussian} of the measures themselves: the limit is independent of $e$. At the same time, the value of $e$ \emph{can} be detected by following the $e$-dependence of the partition function $\sfZ_{\beta,n,e}$, as described by~\eqref{conv:partition-function-zeta}. 

\bigskip

We now apply this observation to the microscopic system of Section~\ref{s:micro}. 
We set $\bfH_n := \bfW+\Xi_n$; note that $\bbP \bfH_n = \bfW$ and $\bbQ \bfH_n = \Xi_n$.

Similarly to the construction above we fix an `energy level' $\calE_0$ and  define the probability measure $\mu_{\beta,n}$ on the subspace $\bfZ\ti\bfH_n$ by 
\[
\mu_{\beta,n} (\rmd z \dd \eta) := \frac1{\sfZ_{\beta,n,\calE_0}} 
 \pdelta\pra*{\calH_{\mathrm{total}}(z,\eta) - \frac{n}{2\beta} - \calE_0}(\rmd z \dd \eta).
\]
By extension we can also consider $\mu_{\beta,n}$ as a probability measure on $\bfZ\ti\bfX$. In contrast to $\zeta_{\beta,n,e}$ above we choose to normalize this measure. 

The choice of energy level $n/2\beta + \calE_0$ is very similar to the one above: in the limit $n\to\infty$ the bulk of the energy will be located in the $n$ degrees of freedom in the invisible part $\Xi_n$ of the heat bath, with approximately energy $1/2\beta$ per degree of freedom. Indeed, we have an equivalence-of-ensembles statement for the full system that mirrors that of the heat bath in Lemma~\ref{l:Sn}:
\begin{lemma}
  \label{l:conv-mu-beta-n}
As $n\to \infty$ the measure $\mu_{\beta,n}$ converges weakly on $\bfZ\ti\bfX$ to $\mu_\beta$. 
\end{lemma}
\noindent
Again,  only the $O(n)$ scaling of the energy level $\calE_0 + n/2\beta$ is visible in the limit $\mu_\beta$; the $O(1)$ part given by the choice of $\calE_0$ can not be recovered from the limit measure $\mu_\beta$. The proof is given Appendix~\ref{app:delayed-proofs}. 

\bigskip

We next map this microcanonical measure $\mu_{\beta,n}$ to a new state space in which we can track the behavior of $z$, $w$, and $e$ independently from the rest of the heat bath. Defining the closure in $\bfX$ of the $\Xi_n$ as  
\[
\Xi^\bfX := \overline{\bigcup_{n\geq 1}\Xi_n }^{\|\cdot\|_\bfX},
\]
we define the `decomposition' map
\begin{align*}
  T_n &\colon\bfZ\ti\bfH_n \longrightarrow \bfZ\ti\bfW\ti\Xi^\bfX\ti\R,\\
  &(z,\eta)\longmapsto \bra*{z,\, w:= \bbP \eta, \,\xi = \bbQ \eta,\, e := \frac12 \|\xi\|^2_\bfH - \frac{n}{2\beta}},
\end{align*}
and  the resulting measure on $\bfZ\ti\bfW\ti\Xi^\bfX\ti\R$,
\begin{equation}
\label{eqdef:mu-beta-n-e}
\mu_{\beta,n,e} := (T_n)_\# \mu_{\beta,n}.  
\end{equation}
The map $T_n$ makes explicit the orthogonal split $\eta = w + \xi = \bbP \eta + \bbQ\eta$, and in addition exposes the renormalized energy variable $\|\xi\|^2/2 - n/2\beta$. 

Since for $\eta\in \bfH_n$ we have 
\[
\calH_{\mathrm{total}}(z,\eta)   = \Hzw(z,\bbP\eta) + \frac12 \|\bbQ\eta\|_{\bfH}^2,
\]
the properties of the `delta function' $\pdelta$ allow us to rewrite $\mu_{\beta,n,e}$ as (see Lemma~\ref{l:delta-times-delta})
\begin{equation}
\label{char:mu-beta-n-e-in-two-deltas}
\mu_{\beta,n,e}(\rmd z \dd w \dd \xi \dd e)
= \frac1{\sfZ_{\beta,n,\calE_0}} 
\pdelta\pra[\Big]{\Hzw(z,w) - \calE_0 + e} (\rmd z\dd w)\,
\pdelta\pra*{\frac12 \|\xi\|_\bfH^2 - \frac{n}{2\beta} - e} (\rmd \xi) \,\dd e.
\end{equation}
By Lemma~\ref{l:Sn} the final $\pdelta$-measure has the asymptotic behavior 
\[
  \ee^{-C(\beta,n)}\pdelta\pra*{\frac12 \|\xi\|_\bfH^2 - \frac{n}{2\beta} - e} (\rmd \xi)
  % = \exp\bra[\big]{\beta C(\beta,n) + \beta e + o(1) } \,
  \ \longrightharpoonup\  \ee^{\beta e}
  \,(\wh \bbQ_\# \gamma_\beta)(\rmd \xi)
  \qquad \text{as }n\to\infty.
\]
Combining the factor $\ee^{-C(\beta,n)}$ with $\sfZ_{\beta,n,\calE_0}$  we thus have proved
% 
% we find the following limit behaviour (see Appendix~\ref{app:delayed-proofs} for the full proof):
\begin{lemma}
  \label{l:mu-beta-n-e-limit}
As $n\to \infty$ we have the narrow convergence on $\bfZ\ti\bfW\ti\Xi^\bfX\ti\R$,
\begin{subequations}
\label{eq:mu-beta-n-e-limit}
\begin{align}
\label{eq:mu-beta-n-e-limit-convergence}
&\mu_{\beta,n,e} \longrightharpoonup \mu_{\beta,e},\\
&\mu_{\beta,e}(\rmd z \dd w \dd \xi \dd e ):= \frac1{\sfZ_{\beta,\calE_0}}  \,
\ee^{\beta e} \,\pdelta\pra[\big]{\Hzw(z,w) - \calE_0 + e} (\rmd z \dd w) \dd e \;(\wh\bbQ_\#\gamma_\beta)(\rmd \xi).
\label{eq:mu-beta-n-e-limit-definition}
\end{align}
\end{subequations}
\end{lemma}
\noindent

\bigskip

We now use this convergence to understand the interpretation of $\calS_\GEN(z,w,e) = \beta e + \mathrm{constant}$. The microcanonical measure $\mu_{\beta,n,e}$, written in the form~\eqref{char:mu-beta-n-e-in-two-deltas}, can be interpreted as the product of two surface measures, the first for System~A  and the second one for the heat bath. Each surface measure is calculated at some energy level parametrized by $e\in \R$, but neither are normalized; the normalization happens at the level of the total product, not  the individual factors. This gives rise to a trade-off: as a function of~$e$, the surface area of the heat bath at energy level  `$c_1+e$' is multiplied by the surface area of System~A  at energy level `$c_2-e$'. 

From the discussion above we know that for large $n$ the surface area of the heat bath at energy level $c_1+e$ has the exponential dependence  $\ee^{\beta e}$ on $e$.  Therefore higher heat-bath energy levels, which correspond to larger $e$ in~\eqref{char:mu-beta-n-e-in-two-deltas},  have a larger contribution in~\eqref{char:mu-beta-n-e-in-two-deltas}, and this dependence gives a downward pressure on the energy of System~A. This downward pressure can be recognized by integrating $\mu_{\beta,e}$ over both $\xi$ and $e$;  we then obtain
\begin{multline*}
  \frac1{\sfZ_{\beta,\calE_0}} \int_\R  
  \ee^{\beta e} \,\pdelta\pra[\big]{\Hzw(z,w) - \calE_0 + e}(B)  \dd e \\
  = \frac1{\sfZ_{\beta,\calE_0}}  \int_B \ee^{-\beta \Hzw(z,w) + \beta 
  \calE_0}\dd z\dd w,
  \qquad \text{for all }B\subset \bfZ\ti \bfW,
\end{multline*}
or more concisely, as measures on $\bfZ\ti\bfW$,
\[
  \frac1{\sfZ_{\beta,\calE_0}} \int_\R  
  \ee^{\beta e} \,\pdelta\pra[\big]{\Hzw(z,w) - \calE_0 + e}  \dd e
  = \frac1{\sfZ_{\beta,\calE_0}}  \ee^{-\beta \Hzw(z,w) + \beta 
  \calE_0} .
\]
(See Definition~\ref{def:pdelta} for this type of manipulation.)  The
expression on the right-hand side coincides with~\eqref{eqdef:nu-beta}, and
shows how the increase in heat-bath surface area with heat-bath energy
penalizes energy increases in System~A.

\bigskip

Summarizing, the formula $\calS_\GEN(z,w,e) = \beta e + \mathrm{constant}$ can
be interpreted as a characterization of the dependence of microcanonical
heat-bath surface area on heat-bath energy, and the discussion above shows how
this dependence has a depressing effect on the energy in System~A.

\AAA

\subsection{Reduction to a dissipative system A via strong coupling and fast relaxation}
\label{su:DissA.QuasiStat}

In applications it is often interesting to study a damped version of the
Hamiltonian system A coupled directly to the heat bath, such that no additional
internal variable $w\in \bfW$ appears. We discuss here on the formal level how
such models can occur as a limit of systems including $w$. The idea is to make
the coupling between $ w$ and $z$ very strong by setting
$\sfC = \lambda \wt\sfC$ and to make the damping matrix even stronger by setting
$\bbD=\lambda^2 \wt\bbD$. Here $\lambda\gg 1$ is a large real parameter such
that $\lambda^2$ measures the quotient of the slow time scales in System A and
the fast relaxation time for $w$. We will be interested in the limit $\lambda
\to \infty$ while keeping the (inverse) temperature $\beta>0$ fixed. 

We start from the SDE \eqref{eq:evol-zwe-ODE} for the variables
$(z_t,w_t,e_t)$ and assume that $\Hmac$ is given in the form
$\Hmac(z)= \wt\calH(z)+ \frac12\lambda^2\|\wt\sfC z\|^2_\bfH$ (cf.\
\eqref{eqdef:Ham-full-2}). For simplicity, we omit the equation for the energy
$e_t$, and obtain
\begin{equation}
  \label{eq:ScaledSDE}
  \binom{\rmd z_t}{\rmd w_t} = \binom{\Jmac\big( \rmD\wt\calH(z_t) + \lambda
  \wt\sfC^*(w_t{+}\lambda \wt\sfC z_t)\big)} {-\lambda^2 \wt\bbD (w_t{+}\lambda
  \wt\sfC z_t)} \dd t + \binom{0}{\lambda \wt\Sigma\rmd B_t}  
\text{ with } \wt\Sigma \wt\Sigma^*=\frac1\beta(\wt\bbD{+}\wt\bbD^*).
\end{equation}
The equation for $w$ is linear and it follows that for $\lambda \gg1$ one needs
$w_t \approx -\lambda \sfC z_t$ to cancel the highest order terms in
$\lambda$. To see the influence of $w_t$ on the macroscopic variable $z_t$, it
is advantageous to introduce the new variable
$u^\lambda =\lambda(w{+}\lambda \sfC z)$, which satisfies the SDE
\begin{align*}
  \rmd z_t &=\Jmac(z_z) \rmD\wt\calH(z_t) + \Jmac(z_z) \wt\sfC^* \dd u^\lambda_t, 
\\
\rmd u^\lambda_t& = 
- \lambda^2 \big(\, \wt\bbD - \wt\sfC \Jmac(z_t)
  \wt\sfC^*\big) u^\lambda_t \dd t + \lambda^2 \wt\sfC \Jmac(z_t) 
  \rmD \wt\calH(z_t) + \lambda^2 \wt\Sigma \dd B_t.
\end{align*}

For simplicity, we now further assume that $\Jmac$ is independent of $z$, so that
$\wt\bbG:= \wt\bbD - \wt\sfC\Jmac\wt\sfC^*$ satisfies $\wt\bbG{+} \wt\bbG^*=
\wt\bbD{+}\wt\bbD^*>0$ and hence generates a contraction semigroup
$\big(\ee^{-t\wt\bbG}\big)_{t\geq 0}$. With this, $u^\lambda_t$ can be expressed by
Duhamel's formula as
\[
u^\lambda_t= \ee^{-t\lambda^2\wt\bbG}u_0 
   + \int_0^t \lambda^2\, \ee^{-(t-s)\lambda^2\wt\bbG}\: 
              \wt\sfC \Jmac \rmD\wt\calH(z_s) \dd s 
   + \int_0^t \lambda^2 \,\ee^{-(t-s)\lambda^2\wt\bbG} \:\wt\Sigma \dd B_s.
\]
From this, we find the deterministic limit $
U_t=\lim_{\lambda\to \infty} u^\lambda_t = \wt\bbG^{-1} \wt\sfC \Jmac
  \rmD\wt\calH(z_t)   + \wt\bbG^{-1} \wt\Sigma \,B_t $, and 
inserting this relation into the equation for $z_t$ yields an SDE for
$z_t$ alone, namely 
\begin{align}
 \label{eq:SDE.zt.alone}
  &\dd z_t = \Jmac \rmD\wt\calH(z_t)\dd t + \Jmac\wt\sfC^* \dd U_t  =
\big( \,\wh\sfJ - \wh\sfD\big) \rmD \wt\calH(z_t)\dd t  +
\wh\bfSigma  \dd B_t \\[0.3em]
\nonumber
& \text{with } \wh\sfJ= \Jmac-\tfrac12\Jmac
\wt\sfC^*(\wt\bbG^{-1}{-}\wt\bbG^{-*}) \wt\sfC\,\Jmac^*, \\
\nonumber  
&\phantom{\text{with }}\wh\sfD= \tfrac12\Jmac
\wt\sfC^*(\wt\bbG^{-1}{+}\wt\bbG^{-*}) \wt\sfC\,\Jmac^*, \text{ and
}  \wh\bfSigma= \Jmac \wt\sfC^*\wt\bbG^{-1} \wt\Sigma.
\end{align}
This SDE can be seen as the limit for $\lambda \to \infty$ of the scaled SDE
\eqref{eq:ScaledSDE}. Clearly, $\wh\sfJ=-\wh\sfJ^*$ and $\wh\sfD=\wh\sfD^*\geq
0$ (since $\wt\bbG$ has positive symmetric part), hence \eqref{eq:SDE.zt.alone} 
is a damped Hamiltonian system with stochastic forcing. 

It remains unclear whether in the general case the SDE \eqref{eq:SDE.zt.alone}
satisfies the fluctuation-dissipation relation and has a simple invariant
measure.  However, for the special case $\wt\sfC \Jmac \wt\sfC^*=0$ we have
$\wt\bbG=\wt\bbD$, and from
$\wt\bbD^{-1}\wt\Sigma \wt\Sigma^*\wt\bbD^{-*}=
\frac1\beta(\wt\bbD^{-1}{+}\wt\bbD^{-*})$ the fluctuation-dissipation relation
\[
\wh\bfSigma \wh\bfSigma^*=\frac1\beta\,\Jmac \wt\sfC^*
\big(\wt\bbD^{-1}{+} \wt\bbD^{-*}\big) \wt\sfC \Jmac^* = \frac2\beta \wh\sfD
\]
follows. With this we see that \eqref{eq:SDE.zt.alone} has 
the invariant measure $\wh \nu_\beta (\rmd z) = Z_\beta^{-1} \ee^{-\beta \wt H(z)}
\dd z$, cf.\ \eqref{eqdef:nu-beta}. 

\begin{example}[Molecular dynamics]\label{ex:MolecDynam}
  The above approach is often applied in studies of molecular dynamics at finite
  temperature (see e.g.\
  \cite{LelievreRoussetStoltz10, DuongPeletierZimmer13, Leimkuhler2015b}). For
  this, we consider a classical mechanical system for
  $z=(q,p) \in \R^m\ti \R^m=\bfZ$ with
  $\wt\calH(q,p)=\Phi(q) + \frac12 p^*M^{-1}p$ and
  $\bbJ=\binom{\ 0\ \ I}{-I\ 0}$.  For the coupling we choose $w\in \R^m=\bfW$,
  $\wt\sfC(q,p)= \sfK q$, and $\wt\bbD=\wt\bbD^*>0$, and simple calculations give
\[
 \wt\sfC \bbJ\wt\sfC^*=0, \quad \wt\bbG= \wt\bbD, \quad
 \wh\sfD= \sfK^* \wt\bbD^{-1}\sfK, \quad \wh\bfSigma =
  \binom{0}{-\sqrt{2/\beta} \: \sfK^* \wt\bbD^{-1/2}}.
\]
Hence, the fluctuation-dissipation relation holds, and the classical 
invariant measure reads
\[
\wh\nu_\beta(\rmd z) = \frac1{Z_\beta}\: \ee^{-\beta \wt\calH(z)} 
\dd z = \frac1{Z_\beta}\, \exp\Big(\! -\beta \Phi(z) - \frac\beta2 
p^* M^{-1}p \Big)  ,
\] 
i.e.\ the momentum has a Gaussian (or Maxwellian) distribution.
\end{example}

\EEE

\section{Conclusion and discussion}
\label{s:conclusion}

While the coarse-graining of Hamiltonian systems is a classical topic, the
introduction of the compression property from dilation theory gives a new type
of insight into this question. Specifically, we have shown how the compression
property gives an abstract characterization of heat-bath behaviour that
automatically generates two important effects.

First, the existence of a compression means in particular that the Hamiltonian
evolution of the coupled system can be projected to the finite-dimensional
subspace $\bfZ \ti \bfW$.  
Secondly, and just as importantly, the compression property implies that the
combination of random heat-bath initial data, Hamiltonian propagation, and
projection generates a noise that is memoryless, leading to the Brownian noise
terms in the equations~\eqref{eq:evol-zw} and~\eqref{eq:evol-zwe}. This is best
expressed by Lemma~\ref{l:Y-is-OU}, and the proof of that lemma clearly shows
how the projection and Hamiltonian propagation combine to give the covariance
function of an Ornstein-Uhlenbeck process.\medskip

In addition, the same compression property has a clear connection to the
\generic structure. The operator $\bbD$ that appears in the compression
property generates terms in both the Poisson operator $\bbJ_\GEN$ and the
Onsager operator $\bbK_\GEN$: the skew-symmetic part of $\bbD$ contributes to
$\bbJ_\GEN$, and the symmetric part to $\bbK_\GEN$. In this way the compression
property generates new structure in the (formerly) Hamiltonian system, which
turns out to be \generic structure.

Finally, the fluctuations generated by the noise are linked to the Onsager
operator $\bbK_\GEN$ by the fluctuation-dissipation
relation~\eqref{eq:FDR-GEN}. This puts the \generic framework on a firmer
footing, and gives modelling insight into the various components of \generic.

\subsection{Relation with models and results in the literature}
\label{ss:relation-with-literature-1}

Coupled Hamiltonian systems consisting of a fixed system and a `heat bath' have
been studied by many authors. The earliest example that we could find is by
Zwanzig~\cite{Zwanzig1973a}, and similar systems return
in~\cite{Zwanzig80,Zwanzig01,KSTT02LTBL}. The system of this paper is a generalization of 
that of~\cite{Zwanzig1973a} to an infinite-dimensional heat bath. Infinite-dimensional heat baths were also used in the
related work of Jaksic and Pillet~\cite{JaksicPillet98} and
Rey-Bellet~\cite{Reyb06OCS}.

The fact that a particular tuning of the properties of the heat bath and its
projector can lead to Markovianity of the projected subsystem was observed by
Ford, Kac, and Mazur~\cite{Ford1965a}, and later used by many
others~\cite{Zwanzig1973a,Zwanzig80}.  Stuart and
co-workers~\cite{Stuart1999a,KSTT02LTBL} appear to be the first to give a
rigorous convergence statement in the limit of infinite dimensions of the heat
bath (and the Compression Property~\ref{ass:dilation} is effectively
equivalent to the choice of the frequencies of the $n$ oscillators in these two
papers). 

Very similar ideas were also used to explain dissipation in quantum systems
(so-called open quantum systems), in particular in the famous Caldeira-Leggett model
\cite{CalLeg83QTDS,CalLeg83PIAQ}, see also \cite{Tokieda2020a} where the heat
bath is interpreted in terms of a phonon-number representation. The
straightforward usage of
the classical dilation theory leads to so-called quantum-state diffusion, see
e.g.\ \cite{GisPer92QSDM}, which does not fully account for the true nature of
quantum probability. A genuine dilation theory for open quantum system needs
a quantum probability theory, see e.g.\ \cite{AcFrGo84QPAQ,Meye95QPP}.

\begin{RmRemark}
  The setup of Jaksic and Pillet~\cite{JaksicPillet98} is similar to ours, but
  their condition (H3) effectively excludes the case of this paper. This is due
  to the {already-remarked} fact that $\bfW\cap \mafo{dom}(\JHB) = \{0\}$,
  which is proved in Remark~\ref{re:Dom.JB.bfW}. Jaksic and Pillet find from
  their (H3) that their equivalent of the process $Y_t$ has $C^1$ dependence on
  the initial datum $(z_0,\eta_0)$ (their Theorem 2.1), while because of the
  lack of continuity we only obtain measurability.
\end{RmRemark}

The coarse-graining of this paper is closely related to the `Mori-Zwanzig'
reduction method \cite{Mori65, Zwan61MEIT, Zwanzig01, Chorin2000a,
  Schilling22}. The first step in this method is the expression of the `hidden'
variables in terms of the `visible' variables. In this paper
equation~\eqref{eq:z-Y} plays the same role. The integral term characterizes
the changes in $z$ that are generated by $z$ itself but through the heat bath:
\[
  \int_0^t \underbrace{\sfC^*  \ee^{(t-s)\JHB}\sfC}_{=:\,K(t-s)} \dot z_s\dd s.
\]
The Compression Property~\eqref{eqass:dilation} amounts to assuming that there
exists a finite-dimen\-sional subspace $\bfW$ with corresponding orthogonal
projection $\bbP\colon\bfH\to \bfW$ and a dissipative operator $\bbD$ on $\bfW$
such that
\begin{equation}
\label{eq:ass-dilation-again}
\text{(a)}\qquad  \sfC^* \bbP = \sfC^* \qquad \text{and\qquad (b)} \qquad K(t) = \sfC^* \ee^{-t\bbD}   \sfC \qquad \text{for } t\geq0.
\end{equation}
Condition (a) means that $\bfW$ contains all information  necessary to characterize the effect of the heat bath on $z$.
Condition (b), on the other hand, amounts to an assumption that $K$ can be written as a finite sum of exponentials (one for each of the coordinates); this is sometimes used as an approximation in numerical algorithms~\cite{Schilling22} (see also~\cite{SteSte19MERM}).

The compression property bypasses the usual solving of the `hidden-variables' equation, which in the case of this paper would involve the orthogonal component $\xi_t = \wh\bbQ \eta_t$. The formal generator of this `orthogonal dynamics' should be either $\bbQ\JHB$ or $\bbQ\JHB\bbQ$, which both are not well defined and tend to be  difficult to study, even formally. In the context of the compression property this dynamics is bypassed, and the effect of the hidden variables on the visible ones is instead fully characterized by~\eqref{eq:ass-dilation-again}.

Mori-Zwanzig reduction schemes come in two forms, depending on whether the split is performed in the state space itself (as in this paper) or in the space of probability measures or `ensembles' on the state space. The latter approach has the advantage of lifting a nonlinear evolution to a linear one, thereby allowing the application of linear theory such as the Mori-Zwanzig reformulation. In the case of this paper, we can apply the simpler, first type, since the nonlinearity is confined to System~A.

\subsection{Relation between positive and zero temperature}

In the zero-temperature limit, $\beta\to\infty$, the positive-temperature equations~\eqref{eq:evol-zw} and~\eqref{eq:evol-zwe} converge to the corresponding deterministic ODEs~\eqref{eq:evol-zw-ODE} and~\eqref{eq:evol-zwe-ODE}. This can be recognized in the characterization $\Sigma\Sigma^*  =\beta^{-1}(\bbD+\bbD^*)$ of the noise intensity, which implies that $\Sigma\to0$ as $\beta\to\infty$.

It is a natural question whether the \generic structure itself converges. The answer is that it does not, at least not in this form. The entropy $\calS_\GEN = \beta e + \text{constant}$ becomes singular as $\beta\to\infty$, and the Onsager operator $\bbK_\GEN$ vanishes. The product $\bbK_\GEN \rmD \calS_\GEN$, however, is independent of $\beta$, and this suggests that in order to study this limit it would make sense to rescale by
\[
\calS_{\GEN}^{\text{rescaled}} := \frac1\beta \calS_\GEN  ,\qquad\text{and}\qquad
\bbK_\GEN^{\text{rescaled}} := \beta \bbK_\GEN.
\]
The rescaled objects $\calS_{\GEN}^{\text{rescaled}}$ and $\bbK_\GEN^{\text{rescaled}}$ are independent of $\beta$, and therefore trivially converge.

As observed in Section~\ref{s:CG-deterministic}, the deterministic \generic evolution, however, is independent of $\beta$.

\subsection{Well-preparedness of the heat bath at time zero}

The central construction of the microscopic evolution in this paper is done in part~\ref{i:t:ex-un-full-ZXw} of Theorem~\ref{t:ex-un-full}. There the initial state of the `visible' part of the heat bath $w_0 = \wh\bbP\eta_0$ is assumed to be given, while the `invisible' part $\xi_0 = \wh\bbQ \eta_0$ is drawn randomly from the corresponding conditioned equilibrium measure $\wh \bbQ_\# \gamma_\beta$ of the heat bath. This randomness then generates the randomness in the \sde~\eqref{eq:evol-zw} of the projected system $(z,w)$. 

In the philosophical debate about irreversibility and coarse-graining (see
e.g.~\cite{Robertson20}) such a setup has been criticized: it amounts to an
assumption that the microscopic state at time $t=0$ is rather special. This can
be recognized in the fact that while the initial hidden variables
$\xi = \wh\bbQ \eta$ have distribution $\wh \bbQ_\# \gamma_\beta$ at time 0,
they \emph{will not} have that same distribution at later or earlier
times. (The distribution $\gamma_\beta$ is invariant under the heat-bath
evolution, but $\wh \bbQ_\# \gamma_\beta$ is not.)

In fact, it turns out that the altered distribution of $\xi$ at times different
from zero \emph{does not} change the effect of the coarse-graining. To
illustrate this, consider starting the heat bath in
part~\ref{i:t:ex-un-full-ZXw} of Theorem~\ref{t:ex-un-full} at some time
$t_0\in \R$ with distribution $\xi_{t_0}\sim \wh \bbQ_\# \gamma_\beta$. At time
$t=0$ the corresponding state equals $\extG_{-t_0}\xi_{t_0}$. Note that this
propagated state is distributed according to
$(\extG_{-t_0} \wh\bbQ)_\# \gamma_\beta$, and is not completely hidden
($\wh\bbP \extG_{-t_0} \xi_{t_0}\not=0$). In equation~\eqref{eq:z-Y}
describing the evolution of~$z$, only the first term changes:
\begin{equation}
  \label{eq:z-Y-modified}	
  \dot z_t = \Jmac \sfC^*Y_{t-t_0}(\xi_{t_0})
  + \Jmac \bra*{\rmD\Hmac(z) + \sfC^*(\ee^{t\JHB}\sfC z_0- \sfC z_t)} 
   + \Jmac \int_0^t \sfC^*  \ee^{(t-s)\JHB}\sfC \dot z_s\dd s.
\end{equation}
However, under the compression property by Lemma~\ref{l:props-Y} the process
$t\mapsto Y_{t-t_0}(\xi_{t_0})$ is again a non-stationary Ornstein-Uhlenbeck
process; in other words the \emph{noise} in this process is stationary, even
though the state at time $t=0$ is not in equilibrium. It follows that the noise
generated by the heat bath has the same distribution, also if one starts with a
hidden variable of the form $\extG_{-t_0}\xi_{t_0}$, where $\xi_{t_0}\sim \wh
\bbQ_\# \gamma_\beta$.

\subsection{Other derivations of \generic by coarse-graining}
\label{ss:other-derivs-of-GENERIC-by-CG}

\paragraph{\"Ottinger's coarse-graining of Hamiltonian systems.}
The closest in spirit to this paper is the coarse-graining approach by
\"Ottinger~\cite[Ch.~6]{Otti05BET}. The starting point is again a
Hamiltonian system, and a central role is played by a `coarse-graining map'
similar to $\pi_{z,w,e}$ in~\eqref{eqdef:projection-zwe}. The two approaches
differ in the definition of the coarse-grained state: while in this paper we
consider $\pi_{z,w,e}(z,\eta)$ itself to be the coarse-grained state (up to
renormalization of the infinite energy), in~\cite[Ch.~6]{Otti05BET} the
coarse-grained state is in effect a \emph{distribution} $\rho_{z,w,e}$ on
$\bfZ\ti\bfX$, {parametrized} by $(z,w,e)$.

It is tempting to consider \"Ottinger's distribution $\rho_{z,w,e}$ as an
approximation of the conditional distribution of $(z,\eta)$ given the value of
$\pi_{z,w,e}(z,\eta)$. Unfortunately this approach seems to fail: in simple
examples one can verify that the \emph{actual} conditional distribution of
$(z,\eta)$ is different from the $\rho_{z,w,e}$ postulated in \"Ottinger's
approach, sufficiently different to make it unclear why $\rho_{z,w,e}$ should
be a good approximation for the purpose of coarse-graining.  While there are
clear parallels between the two approaches, at this stage it is not clear to us
whether they should be considered effectively the same or not.

\paragraph{Obtaining \generic from large deviations.}
\Generic has also been derived from large-deviation principles of stochastic
processes. In~\cite{DuongPeletierZimmer13}, Duong and two of us showed how the
\generic structure of the kinetic Fokker-Planck equation can be recognized in
the large-deviations rate functional of a sequence of `more microscopic'
stochastic interacting particle systems. Since a large-deviations principle can
also be interpreted as a form of coarse-graining, this gives an independent
derivation of the \generic structure, but from a process with stochastic
forcing and deterministic initial data, the opposite of the setting of this
paper. In~\cite{KraaijLazarescuMaesPeletier18,KraaijLazarescuMaesPeletier20}
the arguments of~\cite{DuongPeletierZimmer13} were generalized to a larger
class of systems.

\appendix

\section{Properties of Gaussian measures}
\label{app:GaussianMeasures}

In this appendix we recall a number of properties of Gaussian measures that we use in this paper. We will mostly use the notation and definitions of the book on Gaussian Measures by Bogachev~\cite{Boga98GM}.  

\subsection{Basic setup}

In $\R$, a \emph{centered Gaussian random variable} $\eta$ with variance $\sigma^2$, written as $\eta\sim \calN(0,\sigma^2)$, has the probability density
\begin{equation*}
  \mu(\rmd \eta) = \frac 1 {\sqrt{2 \pi} \sigma} \exp\left(- \frac{\eta^2}{2\sigma^2}\right) \dd\eta.
\end{equation*}
Note that such a scalar random variable has the  equivalent formulation
\[
\text{for all $t\in \R$, }\quad 
\Expectation \pra*{\ee^{i t\eta}} = \ee^{-\frac{\sigma^2t^2}{2}}.  
\]

Centered Gaussian random variables on $\R^n$ have the probability density
\begin{equation}
  \label{eq:A1:Rn-Gaussian-general}
   \mu(\rmd \eta)= \frac 1 {\sqrt{(2 \pi)^n} \det(\Sigma)}
  \exp\left(- \frac 1 2 \eta^T \Sigma^{-2}\eta\right)\dd \eta,
\end{equation}
where $\Sigma$ is a symmetric positive definite $n\times n$ matrix. We now  define the $n$-dimensional Hilbert space $\bfH$ to be $\R^n$ with the norm $\|\eta\|_\bfH^2 := \eta^T\Sigma^{-2}\eta$, which implies that the dual norm is $\|f\|_{\bfH^*} = f^T \Sigma^2 f$. Then the expression above takes the form
\[
 \mu(\rmd\eta) =  \frac1{\sfZ} \exp\bra*{-\frac12\|\eta\|_\bfH^2}\dd \eta
\]

This leads to the following equivalent characterizations of centered Gaussian measures on any finite-dimensional Hilbert space~$\bfH$.
\begin{lemma}
  \label{l:finite-dim-GM}
Let $\bfH$ be a finite-dimensional real Hilbert space, let $\mu\in \ProbMeas(\bfH)$, and fix $\beta>0$. The following are equivalent:
\begin{align}
&\mu(\rmd\eta) = \frac1{\sfZ} \exp\bra*{-\frac\beta2\|\eta\|_\bfH^2}\dd \eta;  \label{l:finite-dim-GM:1d-dens}\\
 \label{l:finite-dim-GM:1d-distr}
 &\text{For all $f\in \bfH^*$, $f(\eta) \sim \calN_\R\bra[\big]{0,\beta^{-1}\|f\|_{\bfH^*}^2}$;}\\[2\jot]
\label{l:finite-dim-GM:1d-distr-v2}
&\text{For all $h\in \bfH$, $\ip{h}{\eta}_\bfH \sim \calN_\R\bra[\big]{0,\beta^{-1}\|h\|_\bfH^2}$.}
\end{align}
\noindent
In this case, if $\eta\sim \mu$, then for any bounded linear operator $\bbO:\bfH\to\bfH$, we have 
\begin{equation}
\label{eq:finite-dimensional-trace}
\Expectation \ip{\bbO\eta}{\eta}_{\bfH} = \trace_\bfH \bbO.
\end{equation}
\end{lemma}
\noindent
Formulations~\eqref{l:finite-dim-GM:1d-distr} and~\eqref{l:finite-dim-GM:1d-distr-v2} above can be interpreted as stating that the covariance structure of the measure $\mu$ is given by the norms of $\bfH$ and $\bfH^*$, as in~\eqref{eq:A1:Rn-Gaussian-general}. 
In this paper we need~$\eta$ to be infinite-dimensional (see for instance Remark~\ref{re:Dom.JB.bfW}), while still having  covariance structure determined by $\bfH$ and $\bfH^*$. The following lemma allows us to construct such Gaussian measures. 

The critical step is that we need to consider larger spaces $\bfX$, such that $\bfH$ is Hilbert-Schmidt-embedded into $\bfX$; the measure is then defined on $\bfX$, and the random variable~$\eta$ takes values in $\bfX$.

\begin{lemma}
  \label{l:inf-dim-GM}
Let $\bfH$ and $\bfX$ be finite- or infinite-dimensional real Hilbert spaces, and let $\mu\in \ProbMeas(\bfX)$. If the embedding $\bfH\hookrightarrow \bfX$ is Hilbert-Schmidt, then there exists a measure $\mu\in \ProbMeas(\bfX)$ with the properties
\begin{align}
&\text{For all $f\in \bfX^*$, $f(\eta) \sim \calN_\R\bra[\big]{0,\beta^{-1}\|f\|_{\bfH^*}^2}$.}
\label{l:inf-dim-GM:char-X*}\\[2\jot]
&\text{For all $x\in \bfX$, $\ip{x}{\eta}_\bfX \sim \calN_\R\bra[\big]{0,\ip{\cov_\mu x}{x}_\bfX}$.}
\label{l:inf-dim-GM:char-X}
\end{align}
Here $\cov_\mu$ is a self-adjoint non-negative trace-class operator on $\bfX$. If $f(\eta) = \ip{x}{\eta}_\bfX$ for some $x\in \bfX$, then $\ip{\cov_\mu x}{x}_\bfX = \beta^{-1}\|f\|_{\bfH^*}^2$. 
\end{lemma}
\noindent
This generalizes the finite-dimensional case, where one can  take $\bfH=\bfX$ and $\cov_\mu = \beta^{-1}\mathrm{id}$, and the properties above reduce to those of Lemma~\ref{l:finite-dim-GM}.

\begin{RmRemark}[Definition of infinite-dimensional Gaussian measures]
The usual definition of a Gaussian measure on an infinite-dimensional space $\bfX$ is that for each bounded linear form $f$ on $\bfX$ the projected measure $f_\#\mu$ is a one-dimensional Gaussian measure (e.g.~\cite[Def.~2.2.1]{Boga98GM}); since $\bfX$ is Hilbert, each such $f$ is an inner product with an $x\in \bfX$, and this coincides with~\eqref{l:inf-dim-GM:char-X}.
\end{RmRemark}

\begin{RmRemark}[Necessity of Hilbert-Schmidt embedding]
The formulation of Lemma~\ref{l:inf-dim-GM} suggests that the embedding $\bfH\hookrightarrow \bfX$  needs to be Hilbert-Schmidt for the measure $\mu$ to exist, and indeed this is necessary. In particular there is no Gaussian measure on an infinite-dimensional space with the identity operator as covariance. 
\end{RmRemark}

\bigskip
The construction above concerns centered Gaussian measures. Non-centered measures are constructed by translating by $a\in \bfX$, resulting in a random variable with expectation~$a$.  We write the corresponding measure as $\calN_\bfX(a,\mathrm{cov})$. Formally, such a measure $\gamma$ has the form
\begin{equation}
  \label{eq-app:formal-Gaussian-measure-in-terms-of-cov}
\frac1{\sfZ}\ee^{-\frac12 \ip{\mathrm{cov}^{-1} (\eta-a)}{(\eta-a)}_{\bfX}} \dd \eta.
\end{equation}
This is the sense in which we construct the `canonical' Gaussian measure with formal structure~\eqref{eq:gauss-formal}.

\begin{proof}[Proof of Lemma~\ref{l:inf-dim-GM}]
Writing $\mathrm {id}$ for the embedding $\bfH\hookrightarrow \bfX$, the operator $\bbC := \mathrm{id}\,\mathrm{id}^*$ is a self-adjoint non-negative trace-class operator; by~\cite[Th.~2.3.1]{Boga98GM} there exists a Gaussian measure $\mu$ with mean zero and covariance operator $\cov_\mu = \beta^{-1}\bbC$, which therefore satisfies~\eqref{l:inf-dim-GM:char-X} by the definition of a Gaussian measure.  

From the definition of $\mathrm{id}:\bfH\hookrightarrow \bfX$ we have $\ip{x}{\eta}_\bfX = \ip{\mathrm{id}^* x}{\eta}_\bfH$. For $f\in \bfX^*$ of the form $f(\eta) = \ip{x}\eta_\bfX = \ip{\mathrm{id}^* x}{\eta}_\bfH$ for some $x\in \bfX$, we then calculate
\begin{align*}
\|f\|_{\bfH^*} &= \sup_{\eta\in \bfH }\frac {f(\eta)}{\|\eta\|_\bfH} 
= \sup_{\eta\in \bfH }\frac {\ip{\mathrm{id}^* x}{\eta}_\bfH}{\|\eta\|_\bfH} 
= \ip{\mathrm{id}^* x}{\mathrm{id}^* x}_\bfH = \ip{\bbC x}{x}_\bfX,
\end{align*}
by which also~\eqref{l:inf-dim-GM:char-X*} follows from~\eqref{l:inf-dim-GM:char-X}.
\end{proof}

\subsection{Construction of $\bfX$ by Gelfand triple}
\label{app:ss:Gelfand-construction}

In Section~\ref{ss:Microscopic-positive-temperature} we aim to construct a Gaussian measure of the formal structure
\begin{equation}
\label{eq-app:formal-version-of-gamma-beta}
\ee^{-\frac\beta2\|\eta\|^2_\bfH }\dd \eta. 
\end{equation}
This expression can be made rigorous as a Gaussian measure on a larger Hilbert space $\bfX$ with the property that $\bfH$ admits a Hilbert-Schmidt embedding into $\bfX$. While essentially `any' such embedding is possible (see~\cite[Lemma~3.2.2]{Boga98GM}), we choose a particular one based on the well-known Gelfand triple. We now describe this setting, which is visualized Fig.~\ref{fig:gelfand}.

\begin{figure}[hbt]
  \tikzcdset{every label/.append style = {font = \small}}
\[\begin{tikzcd}
	\bfV && {\bfH \simeq \bfH^*} && {\bfV^*} \\
	{\bfX^*} &&&& \bfX
	\arrow[hook, from=1-1, to=1-3, "\bbA"]
	\arrow[hook, from=1-3, to=1-5, "\bbA"]
	\arrow[equal, from=1-5, to=2-5, "\textrm{\ definition}"]
	\arrow[dotted, no head, from=1-1, to=2-1]
  \arrow[dotted, no head, from=2-1, to=1-5, shorten >= 10pt]
\end{tikzcd}
\]
\caption{{The Gelfand-triple setup. The top row is the classical Gelfand triple $\bfV \hookrightarrow \bfH \simeq \bfH^* \hookrightarrow \bfV^*$, and the bottom row connects this triple to the notation of this paper. The space $\bfX$ is defined as $\bfV^*$, and $\bfX^*$ is defined as the space of bounded linear maps on~$\bfX$. The two dotted lines indicate the two Riesz identifications that would be possible in this diagram; because of the risk of confusion we make these identifications explicitly, not implicitly (see~\eqref{eq:identification-X*} below).}}
\label{fig:gelfand}
  \end{figure}
The space $\bfX$ is implicitly characterized by the choice of  an unbounded, densely defined, self-adjoint linear operator $\bbA\colon \Dom(\bbA)\subset \bfH \to\bfH$ with the property that $\bbA^{-1}\colon\bfH \to \Dom(\bbA)$ is well-defined and Hilbert-Schmidt. Given such an operator we construct $\bfX$ in the following steps.
\begin{enumerate}
\item We define $\bfV := \Dom(\bbA)$, which is a Hilbert space with norm $\|v\|_\bbV := \|\bbA v\|_\bfH$. This implies that $\bfV = \bbA^{-1}\bfH$.
\item We identify $\bfH^*$ with $\bfH$ and define the norm $\|\cdot\|_{\bfV^*}$ on $\bfH$ by
\[
\frac12 \|h\|_{\bfV^*}^2  := 
\sup_{v\in \bfV } \ip{h}{v}_\bfH - \frac12 \|v\|_\bfV^2
\qquad\text{for }h\in \bfH.
\]
We then define $\bfV^*$ as the closure of $\bfH$ in the norm $\|\cdot\|_{\bfV^*}$. By construction the elements of $\bfV^*$ are in duality with $\bfV$, and we write ${}_{\bfV^*}\!\dual{\xi}{v}_\bfV$ for the dual product.  We then also have  $\bfH \hookrightarrow \bfV^*$, in the sense that  any $h\in \bfH$ is a bounded linear form on $\bfV$ by 
\[
{}_{\bfV^*}\!\dual{h}{v}_\bfV := \ip{h}{v}_\bfH
\qquad\text{for all }v \in \bfV.
\]
\item We extend $\bbA^{-1}$ as an operator $\bfV^*\to \bfH$ by setting for any $\xi\in \bfV^*$
\[
\ip{\bbA^{-1}\xi}{h}_\bfH := {}_{\bfV^*}\!\dual{\xi}{\bbA^{-1}h}_\bfV \qquad \text{for all }h\in \bfH.
\]
We can then write 
\begin{equation}
  \label{eq:norms-V*-H}
\|\xi\|_{\bfV^*} = \|\bbA^{-1}\xi\|_\bfH \qquad\text{for all }   \xi\in \bfV^*.
\end{equation}
The extension is also a symmetric operator in $\bfV^*$, since for $h_1,h_2\in \bfH$ we have 
\[
\ip{\bbA^{-1}h_1}{h_2}_{\bfV^*} = \ip{h_1}{\bbA^{-1}h_2}_\bfH =   \ip{h_1}{\bbA^{-1}h_2}_{\bfV^*}.
\]
We similarly extend $\bbA$ as an operator $\bfH\to\bfV^*$.

\item We define $\bfX := \bfV^*$. The space $\bfX$ is a Hilbert space with norm $\|\xi\|_\bfX = \|\bbA^{-1}\xi\|_\bfH$ for all $\xi\in \bfX$.
% \item 
We have $\bfH\hookrightarrow \bfX$ as mentioned above, with $\|h\|_{\bfX} = \|\bbA^{-1}h\|_\bfH$; this shows that 
the embedding $\bfH\hookrightarrow \bfX$ is Hilbert-Schmidt.

\item We characterize elements of the dual space
\[
\bfX^* := \left\{ f\colon\bfX \to \R \text{ linear and bounded }\right\}
\]
in two different ways, by applying Riesz representation on either $\bfX$ or $\bfH$, leading to an element $x_f\in \bfX$ or equivalently an element $v_f\in \bfV$:
\begin{subequations}
  \label{eq:identification-X*}
\begin{align}
f(\xi) &= \ip{\xi}{x_f}_\bfX && \qquad \text{for all }\xi 
\in \bfX,\\
&= {}_{\bfV^*}\!\dual{\xi}{v_f}_\bfV &&\qquad \text{for all }\xi 
\in \bfX,
\label{eq:identification-X*-xf}\\
&= \ip{\xi}{v_f}_\bfH &&\qquad \text{whenever }\xi 
\in \bfH.
\end{align}
\end{subequations}
We have $x_f = \bbA^2  v_f$ and $\|f\|_{\bfX^*} = \|x_f \|_\bfX = \|v_f\|_\bfV$.
\end{enumerate}

\bigskip
The setup above defines the space $\bfX$. In order to define the measure~\eqref{eq-app:formal-version-of-gamma-beta} we define the trace-class operator
\[
\bbC := \bbA^{-2}\colon \bfX\to \bfX 
\qquad\text{or equivalently}\qquad
\ip{\bbC \xi}\xi_\bfX := \|\bbA^{-1}\xi\|_\bfX^2 \text{ for }\xi\in\bfX.
\]
We then set, {for $\beta>0$,} $\gamma_\beta$ to be the Gaussian measure on $\bfX$ 
\[
\gamma_\beta = \calN_{\bfX}(0,\beta^{-1}\bbC) .
\]
This can be considered a rigorous version of~\eqref{eq-app:formal-version-of-gamma-beta}, because from $\|\eta\|_\bfX = \|\bbA^{-1} \eta\|_\bfH$ we deduce that
\[
\ee^{-\frac\beta 2 \|\eta\|_\bfH^2 }
= \ee^{-\frac\beta 2 \|\bbA \eta\|_\bfX^2 }
= \ee^{-\frac\beta 2 \ip{\bbA^2 \eta}{\eta}_\bfX}
= \ee^{-\frac\beta 2 \ip{\bbC^{-1} \eta}{\eta}_\bfX},
\]
which matches with~\eqref{eq-app:formal-Gaussian-measure-in-terms-of-cov}.

Note that $\bbC$ also is trace-class as a map $\bfH\to\bfH$, since in terms of an orthonormal basis $e_k$ of $\bfH$ we can write
\[
\trace_\bfH \bbC = \sum_{k=1}  ^\infty \ip{\bbC e_k}{e_k}_\bfH
= \sum_{k=1}  ^\infty \|\bbA^{-1}e_k\|_\bfH ^2 ,
\]
and $\bbA^{-1}$ is Hilbert-Schmidt on $\bfH$.

\begin{RmRemark}[Positive temperature implies infinite energy]
  \label{rem:positive-temperature-means-infinite-energy}
A well-known property of infinite-dimensional Gaussian measures such as $\gamma_\beta$ is that for any $\beta>0$, $\gamma_\beta(\bfH) = 0$. This shows that if $\eta$ is distributed according to $\gamma_\beta$, then the natural energy $\frac12 \|\eta\|_\bfH^2$ is almost surely infinite. 
See e.g.~\cite[Sec.\,2.1]{Reyb06OCS} or~\cite[Rem.\,2.2.3]{Boga98GM} for a discussion of this property. 
\end{RmRemark}

To facilitate applying the theorems of~\cite{Boga98GM} we make explicit the relationship between the objects of this paper and the notation of~\cite{Boga98GM}. 
The covariance bilinear form $R_{\gamma_\beta}$ is given by 
\begin{equation}
\label{eq:Rgammabeta-C}
R_{\gamma_\beta}(f)(f) = \frac1 \beta \ip{\bbC x_f}{x_f}_\bfX 
\qquad\text{for $f\in\bfX^*$ and for $x_f$ given by~\eqref{eq:identification-X*-xf}},
\end{equation}
and the Cameron-Martin Hilbert space $H(\gamma_\beta)$ coincides with $\bfH$, up to a scaling of the norm:
\begin{equation}
\label{eq:Hgammabeta}
\|h\|_{H(\gamma_\beta)}^2 = \beta \|h\|_\bfH^2 
\qquad \text{for all }h\in \bfH.
\end{equation}

\subsection{Measurable extensions of linear maps}

In Section~\ref{s:micro} we constructed several extensions of linear operators on $\bfH$ to $\bfX$. In this section we give the details of these extensions. 

\begin{definition}[{Measurable linear extensions 
  of  \cite[Def.\,2.10.1+3.7.1]{Boga98GM}}]   
\label{def:measurable-linear-extensions}
Let $\bfY$ be a normed linear space, and let $f\colon\bfH\to\bfY$ be given. If there
exists a subspace  $\bfL$ with  
$\bfH\subset \bfL \subset \bfX$ and  $\gamma_\beta(\bfL)=1$ and a
linear map $\wh f\colon \bfL\to \bfY$, measurable with respect to
$(\calB(\bfX)_{\gamma_\beta}, \calB(\bfY))$, such that $\wh f$ coincides with
$f$ on $\bfL$, then we call $\wh f$ a \emph{measurable linear extension} of $f$.
\end{definition}

Such an extension is characterized ${\gamma_\beta}$-a.e.\ by its values on
$\bfH$~\cite[Cor.\,2.10.8]{Boga98GM}, and therefore can be considered
$\gamma_\beta$-unique.

\medskip Note that if $\eta\sim \gamma_\beta$ and $g$ is a measurable map from
$\bfX$ to $\bfY$, then $g(\eta)$ has distribution $g_\# \gamma_\beta$, the
push-forward measure of $\gamma_\beta$ under $g$, defined by
\[
\int_\bfY \varphi(\xi) (g_\#\gamma_\beta)(\rmd \xi)  
= 
\int_\bfX \varphi(g(\eta)) \gamma_\beta(\rmd \eta)
\qquad\text{for all }\varphi\in C_b(\bfY).
\]
  
\begin{lemma}[Extensions of bounded linear operators on $\bfH$]
\label{l:measurable-extensions}
  Let $\eta$ have distribution $\gamma_\beta$.
\begin{enumerate}
\item \label{l:measurable-extensions-real-valued}
If $h\in \bfH$, then the bounded linear map $f_h\colon \bfH \ni \eta \mapsto \ip{h}{\eta}_\bfH\in \R$
can be extended to a measurable linear map $\wh f_{h}$, and we have for all $h_1,h_2\in \bfH$
\begin{equation}
  \label{eq:formula-for-covariance}
  \Expectation \pra*{\wh f_{h_1}(\eta) \wh f_{h_2}(\eta)}
  = \raisearound{``}{''}{\; \Expectation \pra[\big]{\ip{h_1}{\eta}_\bfH \ip{h_2}{\eta}_\bfH}} = \frac1\beta \ip{h_1}{h_2}_\bfH.
\qquad 
\end{equation}

\item \label{l:measurable-extensions-operators}
If $\bbO\colon\bfH\to\bfH$ is a bounded linear operator on $\bfH$, then there exists a measurable extension~$\wh \bbO$. Then $\wh\bbO\eta$  is an $\bfX$-valued Gaussian random variable with covariance $\beta^{-1}\bbC_\bbO$ defined by 
\begin{equation}
\label{char:covariance-O}
\frac1\beta \ip{\bbC_\bbO x}{x}_\bfX := 
\frac1\beta \|\bbO^*\bbC x\|_\bfH^2, \qquad x\in \bfX,
\end{equation}
where $\bbO^*$ is the adjoint of $\bbO$ in $\bfH$.

If in addition $\bbO$ has finite rank, then  $\wh\bbO\eta$  is an $\bfH$-valued Gaussian random variable with covariance $\beta^{-1}\bbO\bbO^*$ on $\bfH$. In particular, $\wh\bbO\eta$  
has finite $\bfH$-variance given by
\begin{equation}
  \label{eq:variance-of-finite-rank-operator}
\Expectation \norm[\big]{\wh \bbO \eta}_\bfH^2 = \frac1\beta \trace_\bfH (\bbO\bbO^*).
\end{equation}

\item \label{l:measurable-extensions-adding-independent-variables}
Let $\bbP$ be an orthogonal projection on $\bfH$ and $\bbQ:= I-\bbP$ the complementary projection. Let $X\sim \wh\bbP_\#\gamma_\beta$ and $Y\sim \wh \bbQ_\#\gamma_\beta$ be independent random variables. Then $X+Y$ is distributed according to $\gamma_\beta$.

\medskip

In addition, in order to recognize a probability measure $\zeta\in \ProbMeas(\bfX)$ as $\wh \bbQ_\#\gamma_\beta$ it is sufficient that there exists a subspace $\wt \bfV\subset \bfV$, dense in $\bfH$ in the $\bfH$-norm, such that for all $h\in \wt\bfV$
\begin{equation}
\label{eq:char-FT-dense-subspace}
  \int_{\bfX} \ee^{i\ip{h}{\xi}_\bfH} \zeta(\rmd \xi)
= 
 \exp\bra*{-\frac1{2\beta} \|\bbQ h\|_\bfH^2 }.
\end{equation}
\end{enumerate}
\end{lemma}

\begin{RmRemark}
\label{re:l:measurable-extensions-operators-unitary}
Part~\ref{l:measurable-extensions-operators} above implies in particular that
if $\bbO$ is a unitary operator on $\bfH$, then $\wh \bbO\eta$ has the same
covariance as $\eta$.
\end{RmRemark}

\begin{proof}
  Part~\ref{l:measurable-extensions-real-valued}
  is~\cite[Th.~2.10.11]{Boga98GM} in combination with~\eqref{eq:Hgammabeta}.
  In part~\ref{l:measurable-extensions-operators}, the general case
  is~\cite[Th.~3.7.6]{Boga98GM}. To calculate the covariance, first note that
  since measurable linear extensions are characterized by their values on
  $\bfH$, the two linear extensions $\wh f_1$ and $\wh f_2$ of
\[
f_1(\xi) := \ip h {\wh \bbO \xi}_\bfH
\qquad\text{and}\qquad
f_2(\xi) := \ip {\bbO^* h}\xi_\bfH  
\]
are $\gamma_\beta$-a.e.\ equal. It follows that 
\[
  \Expectation |f_1(\eta)|^2 = \Expectation |f_2(\eta)|^2
  \stackrel{\eqref{eq:formula-for-covariance}}= \frac1\beta \|\bbO^*
  h\|_\bfH^2.
\]
The characterization~\eqref{char:covariance-O} follows by remarking that for
$x\in \bfX$, $\ip x\xi_\bfX = \ip {\bbC x}\xi_\bfH$, so that $f_1$ equals
$\ip x\xi_\bfX $ whenever $h=\bbC x$.

The case of finite rank follows from
part~\ref{l:measurable-extensions-real-valued}, by choosing orthonormal bases
$\{e_k\}_{k=1}^\infty$ and $\{f_k\}_{k=1}^\infty$ of $\bfH$ such that that
$\Range(\bbO)\subset \Span\bra*{\{e_k\}_{k=1}^d}$ and
$\Nullspace(\bbO)^\perp \subset \Span\bra*{\{f_k\}_{k=1}^d}$, and expanding the
squared norm (see also~\cite[Ex.\,3.7.14]{Boga98GM}).

In part~\ref{l:measurable-extensions-adding-independent-variables} the distribution of $X+Y$ is given by~\cite[Lemma~3.11.25]{Boga98GM}. The characterization $\zeta=\wh\bbQ_\#\gamma_\beta$ can be obtained as follows. By~\cite[Cor.~A.3.13]{Boga98GM}, the identity~\eqref{eq:char-FT-dense-subspace} for all elements of $\wt\bfV$ uniquely identifies $\zeta$ as a centered Gaussian measure on $\bfX$. Writing for $h = \bbC x$ the exponent as 
\[
\ip h\xi_\bfH  = \ip{\bbC x}\xi_\bfH =   \ip{x}\xi _\bfX,
\]
we observe that $\zeta$ has covariance operator 
\[
\frac1\beta \ip{\bbC_\zeta x}x_\bfX := \frac1\beta 
\|\bbQ h\|_\bfH^2 = \frac1\beta \|\bbQ \bbC x\|_\bfH^2.
\]
Using part~\ref{l:measurable-extensions-operators} we recognize this Gaussian measure as the law of $\wh\bbQ \eta$ where $\eta\sim \gamma_\beta$. 
\end{proof}

An important consequence of the above extension results is that the
extended group $(\extG_t)_{t\in \R}$ on $\bfX$ preserves the measure
$\gamma_\beta$. 

\begin{lemma}[Invariance of the Gaussian measure] 
\label{le:GauMea.inv.extG} 
The extension $\extG_t:\bfX\to \bfX$ of $\ee^{t\JHB}:\bfH\to \bfH$ (see
Assumption \ref{ass:HB-evol-op-ct-on-X}) leaves the Gaussian measure
$\gamma_\beta$ invariant. 
\end{lemma}
\begin{proof}
We apply Part~\ref{l:measurable-extensions-operators} of
Lemma~\ref{l:measurable-extensions} with the unitary operator
$\bbO=\ee^{t\JHB}$ and its extension $\wh\bbO=\extG_t$. Hence, starting from
$\eta \sim \gamma_\beta$ we see that the random variable $\wh\bbO\eta= \extG_t
\eta$ is a centered Gaussian in $\bfX$ with covariance $\beta^{-1}\bbC_t$
characterized by  
\[
  \frac1\beta \ip{\bbC_t x}{x}_\bfX  = \frac1\beta \|\ee^{-t\JHB} \bbC
  x\|_\bfH^2 =  \frac1\beta \| \bbC x\|_\bfH^2.
\]
In particular, this means $\beta^{-1}\bbC_t = \beta^{-1} \bbC$; hence
$\wh S_t \eta$ has the same distribution~$\gamma_\beta$.
\end{proof}

\begin{RmRemark}[Invariance of $\gamma_\beta$ via inner products]
\label{rem:InvarViaInnerProd}
The invariance of $\gamma_\beta$ can also be shown by working in $\bfX$ rather
than in $\bfX^*$ as in the above proof. We do not have the commutation property
$\bbC \extG_t= \extG_t \bbC$ but we have the invariance property 
\begin{equation}
\label{eq:S*t.C.St=C}
\extG_t^* \bbC \extG_t= \bbC, 
\end{equation}
which implies $\ip{\bbC \extG_t\eta}{\extG_t\eta} = \ip{\extG_t^*\bbC
  \extG_t\eta}{\eta}= \ip{\bbC \eta}\eta$.   

To see the validity of \eqref{eq:S*t.C.St=C}, we need to calculate the
Hilbert-space adjoint $\extG_t^* $ of $\extG_t$. For an operator $A\colon\bfU\to
\bfV$ denote the Banach-space adjoint by $A^\top\colon\bfV^*\to \bfU^*$, which is
defined via $\langle A^\top\nu,u\rangle_\bfU=\langle \nu, Au\rangle_\bfV$. 
Using the dual space $\bfX^*$ satisfying
$\bfX^*\subset \bfH \subset \bfX$ and the duality isomorphism $\bbC\colon\bfX \to
\bfX^*$ we have the relation $\extG_t^*= \bbC \,\extG_t^\top \,\bbC^{-1}$,
where $\extG_t^\top\colon \bfX^*\to \bfX^*$.
As $\extG_t$ is the extension of $\ee^{t\JHB}$ from $\bfH$ to $\bfX$, the
Banach-space adjoint is simply the restriction
$\ee^{t\JHB}|_{\bfX^*}$. Together with $(\ee^{t\JHB})^\top= \ee^{-t\JHB}$ in
$\bfH$, we found the representation
$\extG_t^*= \bbC \,\ee^{-t\JHB}|_{\bfX^*}\,\bbC^{-1} $, and the desired
identity \eqref{eq:S*t.C.St=C} follows.
\end{RmRemark}

\section{Delayed proofs}
\label{app:delayed-proofs}

\subsection{Section~\ref{s:micro}: Microscopic setup}

\begin{proof}[Proof of Lemma~\ref{l:props-Y}]
Part~\ref{l:props-Y:ct-in-H} follows from the assumption that $\ee^{t\JHB}$ is
a strongly continuous group on $\bfH$ and from the fact that on $\bfH$ the operator
$\wh\bbP$ reduces to the bounded linear operator $\bbP$.  
  
For part~\ref{l:props-Y:ct-in-X}, first note that $Y(\eta_0)$ is stationary
since $\gamma_\beta$ is invariant under $\extG_t$. We next show that both
processes are measurable with respect to the product $\sigma$-algebra
$\calB(\R)\ti\calB(\bfX)_{\gamma_\beta}$. The following decomposition shows
that $Y_t(\eta_0)$ is a measurable function of $(t,\eta_0)$, with each arrow a
measurable map between the two $\sigma$-algebras:
\begin{alignat*}{4}
  (t,\eta_0)\qquad &\quad\mapsto\quad &\extG_t \eta_0\quad &\quad\mapsto \quad&
  \wh\bbP \extG_t\eta_0
  \\
  \calB(\R)\ti \calB(\bfX)_{\gamma_\beta} & &\calB(\bfX)_{\gamma_\beta} & &\
  \calB(\bfW)
\end{alignat*}
The first arrow is measurable by the continuity given by
Assumption~\ref{ass:HB-evol-op-ct-on-X}, and the second by the measurability of
$\wh\bbP$ as a map from $(\bfX,\calB(\bfX)_{\gamma_\beta})$ to
$(\bfW, \calB(\bfW))$ (Lemma~\ref{l:measurable-extensions}).
  
To show that $Y=Y(\wh\bbQ\eta_0)$ similarly is measurable we decompose it as
\[
  Y_t(\wh\bbQ\eta_0) = \wh\bbP \extG_t \wh\bbQ\eta_0 = \wh \bbP
  \extG_t\eta_0 - \wh\bbP \extG_t\wh\bbP \eta_0.
\]
We just showed that the first term is measurable from
$\calB(\R)\ti \calB(\bfX)_{\gamma_\beta}$ to $\calB(\bfW)$. The
measurability of the second follows from remarking that $
\wh\bbP \extG_t \wh\bbP \eta_0 = \bbP\ee^{t\JHB}\wh \bbP \eta_0$,
since the range of $\wh\bbP$ is included in $\bfW\subset\bfH$. With this we then have
\begin{alignat*}{4}
  (t,\eta_0)\qquad &\quad\mapsto\quad &(t,\wh\bbP \eta_0)\qquad
  &\quad\mapsto\quad &\ee^{t\JHB} \wh\bbP\eta_0 &\quad\mapsto \quad& \bbP
  \ee^{t\JHB} \wh\bbP\eta_0.
  \\
  \calB(\R)\ti \calB(\bfX)_{\gamma_\beta} && \calB(\R)\ti \calB(\bfW) &
  &\calB(\bfH)\ & &\ \calB(\bfW)\
\end{alignat*}
Here the first arrow is measurable again by the construction of $\wh\bbP$
(Lemma~\ref{l:measurable-extensions}), the second by the strong continuity of
$\ee^{t\JHB}$ on $\bfH$, and the third by the continuity of $\bbP$ on $\bfH$.
  
Finally, to show that $Y(\eta_0)$ and $Y(\wh\bbQ\eta_0)$ have sample paths in
$L^2_{\mathrm{loc}}$, we calculate for $t\in \R$ that
\[
  \Expectation _{\eta_0\sim\gamma_\beta} \|Y_t(\eta_0)\|_\bfH^2
  = \Expectation_{\eta_0} \norm[\big]{\wh\bbP\ee^{t\JHB}\eta_0}_\bfH^2
  \stackrel{\eqref{eq:variance-of-finite-rank-operator}}= \frac1\beta
  \trace_\bfH \bra[\big]{\bbP\ee^{t\JHB}\bra*{\bbP\ee^{t\JHB}}^*} = \frac1\beta
  \trace_\bfH\bbP = \frac{\dim \bfW}\beta. 
\]
The measurability allows us to apply Fubini and calculate for
$-\infty<a<b<\infty$
\[
  \Expectation_{\eta_0\sim\gamma_\beta} \int_a^b \|Y_t(\eta_0)\|_\bfH^2 \dd t
  = \int_a^b \Expectation _{\eta_0} \|Y_t(\eta_0)\|_\bfH^2
  = (b{-}a) \!\; \frac{\dim \bfW}\beta .
\]
This shows that almost surely $Y(\eta_0)$ is an element of
$L^2_{\mathrm{loc}}(\R;\bfW)$.
  
To show the analogous result for $Y(\wh\bbQ \eta_0)$, we similarly obtain
\[
  \Expectation_{\eta_0\sim\gamma_\beta} \int_a^b
  \norm[\big]{Y_t(\wh\bbQ\eta_0)}_\bfH^2 \dd t = \frac1\beta (b{-}a)\trace_\bfH
  \bra{\bbP\ee^{t\JHB}\bbQ^2\ee^{-t\JHB}\bbP}.
\]
To estimate the trace we choose an orthonormal basis $(e_k)_{k=1}^d$ of $\bfW$
of size $d=\dim\bfW$, and calculate
\begin{align*}
  \trace_\bfH \bra{\bbP\ee^{t\JHB}\bbQ^2\ee^{-t\JHB}\bbP}
  &= \sum_{k=1}^d \abs*{\ip{\bbP\ee^{t\JHB}\bbQ^2\ee^{-t\JHB}\bbP e_k}{ e_k}_\bfH}
  \leq \sum_{k=1}^d \norm*{\bbP\ee^{t\JHB}\bbQ^2\ee^{-t\JHB}\bbP e_k}_\bfH\norm*{e_k}_\bfH\\
  &\leq \sum_{k=1}^d \norm*{e_k}_\bfH^2 = d.
\end{align*}
Again it follows that $Y(\wh\bbQ\eta_0)$ is almost surely in
$L^2_{\mathrm{loc}}(\R;\bfW)$.
  \end{proof}

\subsection{Microcanonical measures and Section~\ref{s:macro}}
\label{app:microcanonical-measure}

In the physics literature it is common to define `microcanonical ensembles', which effectively are measures on the set of possible states. A useful tool for this is the following delta-function-like object.

Let $f\colon\R^n\to \R$ be continuous.
For $a\in \R$ the non-negative measure $\mu\in \calM(\R^n)$
\[
\mu (\rmd x) := \pdelta\pra*{f(x)-a }(\rmd x)
\]
is defined by the property 
\[
\mu(A)
= \lim_{h\downarrow 0} \frac1h \int_{A}
  \One\big\{ a \leq f(x) < a+h\big\}\dd x \in [0,\infty]
  \qquad\text{for all }A\subset \R^n,
\]
whenever this definition makes sense. 

One way of making this definition rigorous is by integration against test functions:
\begin{definition}[{\cite[(1.25)]{LelievreRoussetStoltz10}}]
  \label{def:pdelta}
The measure $ \pdelta\pra*{f(x)-a}$ on $\R^n$ is defined by the property that for all $F\in C_c(\R^n)$ and $G\in C(\R)$, 
\begin{equation}
\label{eqdef:pdelta}
\int_\R G(a) \int_{\R^n} F(x) \pdelta\pra*{f(x)-a}(\rmd x) \dd a
= \int_{\R^n} G(f(x)) F(x) \dd x.
\end{equation}
\end{definition}
\noindent
The right-hand side above is finite for all $F\in C_c(\R^n)$ and $G\in C(\R)$,
since $\supp F$ is compact and $G\circ f$ therefore is bounded on this
support. This implies that $\pdelta\pra*{f(x)-a}(\rmd x) \dd a$ can be
considered a locally finite non-negative measure on $(x,a)\in \R^n\ti \R$. For
almost all $a\in \R$ the disintegration $\pdelta\pra*{f(x)-a}(\rmd x)$ is
then a locally finite non-negative measure on $\R^n$.

By the co-area formula (see, e.g.,~\cite[Sec.~3.2.1]{LelievreRoussetStoltz10}) this measure can also be expressed with help of the $(n{-}1)$-dimensional Hausdorff measure $\calH^{n-1}$ on the set $S_a := \{x\in \R^n: f(x) = a\}$ as
\begin{equation}
\label{eq:pdelta-co-area}
  \pdelta\pra*{f(x)-a }(\rmd x) 
  = \frac1{|\nabla f(x)|}  \calH^{n-1}\big|_{S_a}(\rmd x).
\end{equation}

\medskip

In Section~\ref{ss:explanation-S-part2} we made use of the following property,
which allows us to write a microcanonical measure on a product space as the
product of two such measures on the individual spaces. The crucial element is
the fact that the `energy' function has the additive structure
$(x,y)\mapsto f(x) + g(y)$.
\begin{lemma}
  \label{l:delta-times-delta}
Let $X$ and $Y$ be finite-dimensional spaces, and let $f\colon X\to\R$ and $g\colon Y\to \R$ be continuous. 
Set 
\[
T\colon X\ti Y \to X\ti Y \ti \R,
\qquad
(x,y) \mapsto (x,y,e := g(y)).
\]
For Lebesgue almost all $a\in \R$ we then have
\[
\bra[\Big]{T_\# \pdelta\pra*{f(x)+g(y)- a}} (\rmd x\dd y \dd e)
= \pdelta\pra*{f(x)-a+e}(\rmd x) \pdelta\pra*{g(y)-e}(\rmd y ) \dd e.
\]
% and for all measurable $A\subset X$, $B\subset Y$, we have
% \begin{multline}
% \int_\R \pdelta\pra*{f(x)-e-E}(A) \pdelta\pra*{g(y)+e}(B) \dd e
% = \pdelta\pra*{f(x)+g(y)-E}(A\ti B) \\
% \qquad \text{for almost every }E\in \R.
% \end{multline}
\end{lemma}

\begin{proof}
By the definition of the push-forward $T_\#\pdelta$ we need to show that for all $\varphi\in C_c(X\ti Y \ti \R)$ we have for almost all $a\in \R$ 
\begin{multline*}
\int_{X\ti Y}   \varphi(x,y,g(y)) \pdelta\pra*{f(x)+g(y)- a}(\rmd x\dd y) \\
= \int_{X\ti Y\ti \R}   \varphi(x,y,e) \pdelta\pra*{f(x)-a+e}(\rmd x) \pdelta\pra*{g(y)-e}(\rmd y ) \dd e.
\end{multline*}
Since in the right-hand integral we can replace $\varphi(x,y,e)$ by $\varphi(x,y,g(y))$, it is sufficient to show that for almost all $a\in \R$ and for all measurable $A\subset X$, $B\subset Y$
\begin{equation}
  \label{eq:delta-times-delta-to-prove}
  \pdelta\pra*{f(x)+g(y)-a}(A\ti B) = 
  \int_\R \pdelta\pra*{f(x)+e-a}(A) \pdelta\pra*{g(y)-e}(B) \dd e
  .
\end{equation}

To prove this identity we take $G\in C(\R)$ and $F\in C_c(X\ti Y)$ and rewrite the left-hand side of~\eqref{eq:delta-times-delta-to-prove} using~\eqref{eqdef:pdelta},
\begin{align*}
\MoveEqLeft \int_\R G(a) \int_{X\ti Y} F(x,y) \pdelta\pra*{f(x)+g(y)-a}(\rmd x \dd y ) \dd a \\
&\leftstackrel{\eqref{eqdef:pdelta}}= \int_{X \ti Y} G(\underbrace{f(x)+g(y)}_{=: g_x(y)}) F(x,y) \dd y\dd x\\
&\leftstackrel{\eqref{eqdef:pdelta}}= 
\int_X \int_\R G(\tilde e) \int_Y F(x,y) \pdelta\pra*{g_x(y)-\tilde e}(\rmd y) \dd \tilde e \dd x\\
&\leftstackrel{\tilde e  = e+ f(x)} = \int_X \int_\R G( e+f(x)) \underbrace{\int_Y F(x,y) \pdelta\pra*{g(y)- e}(\rmd y) }_{=: \tilde F_e(x)}\; \dd  e\dd x\\
&= \int_\R \int_X G(e+f(x)) \tilde F_e(x) \dd x \dd e\\
&\leftstackrel{\eqref{eqdef:pdelta}}= 
\int_\R \int_\R G(a) \int_X \tilde F_e(x) \pdelta\pra*{e+f(x) - a}(\rmd x) \dd a \dd e\\
&= \int_\R G(a) \int_{X\ti Y\ti \R} F(x,y) \pdelta\pra*{f(x)+ e-a}(\rmd x) \pdelta\pra*{g(y)-e}(\rmd y) \dd e\; \dd a.
\end{align*}
It follows that for almost all $a\in \R$ the identity~\eqref{eq:delta-times-delta-to-prove} holds. 
\end{proof}

For the next two lemmas we adopt the setting of Section~\ref{ss:explanation-S-part2}. Since the aim is to illustrate the interpretation of the formula $\calS = \beta e$, we simplify by assuming that $\bfW\subset \bfV$. We then can choose an orthonormal basis $\{e_k\}_{k\geq 1}$ of both $\bfW$ and $\bfH$ consisting of elements of $\bfV$ (see~\cite[Cor.~2.10.10]{Boga98GM}). The benefit of choosing the basis in $\bfV$ is that the basis vectors can be interpreted as elements of $\bfX^*$, allowing us to apply existing density results.

Letting $\bfW$ be $\Span \{e_k\}_{k=1}^d$, we then  construct an increasing sequence of $n$-dimensional subspaces $\Xi_n$ of the complement $\bfW^\perp$,
\[
\Xi_n := \Span \{e_k\}_{k= d+1}^{d+n}.
\]
We then define the microcanonical probability measure of radius $\sqrt R$:
\begin{equation}
  \label{eqdef:micro-can-measure-nu-nR}
  \nu_{n,R}(\rmd \xi ) := \frac1{Z_{n,R}} \pdelta\pra*{\frac12\|\xi\|^2_\bfH - \frac{nR}2 }\bigg|_{\Xi_n}(\rmd \xi).
\end{equation}
This measure is a normalized version of the measure $\zeta_{n,\beta,e}$ in~\eqref{eqdef:zeta-n-beta-e} for the choice $R = \beta^{-1} + 2e/n$.

The first result is the `equivalence of ensembles', also known as `Boltzmann's equivalence principle' or `Poincar\'e's Lemma' (see \cite[Sec.~6]{DiaconisFreedman87} for a discussion of the history). This result states that if $X_n$ is distributed according to $\nu_{n,R}$, then the first $k$ coordinates resemble a Gaussian random variable in the limit $n\to\infty$:

\begin{lemma}[Equivalence of ensembles; e.g.\ \cite{DiaconisFreedman87}]
\label{l:equivalence-of-ensembles}
Let $\{e_k\}_{k=1}^\infty $ be an orthonormal basis of $\bfH$ such that $\{e_k\}_{k=1}^n$ is a basis of $\Xi_n$. 
Let $R_n\to R>0$  and let $X_n$ have law $\nu_{n,R_n}$. Fix $k\geq 1$. As $n\to\infty$, the first $k$ coordinates 
\[
\bra[\big]{\ip{X_n}{e_1}_\bfH, \dots, \ip{X_n}{e_k}_\bfH}
\]
converge in distribution to a centered Gaussian random variable in $\R^k$ with covariance~$R\,I_k$. 
\end{lemma}

The second result is a variance bound for $\nu_{n,R}$, which probably is standard, but for which we could not find a convenient reference. Note that it does not follow from  equivalence-of-ensemble results such as Lemma~\ref{l:equivalence-of-ensembles} since it concerns the variance of the full random variable, not a finite-dimensional marginal.

\begin{lemma}[Variance bound for microcanonical measures]
\label{l:variance-bound-MCmeasure}
Let $Y_n$ have distribution $\nu_{n,R}$, and let $\bbT$ be a symmetric non-negative trace-class operator on $\bfH$. 
Then there exists $C>0$, independent of $R$ and $\bbT$, such that 
\[
\Expectation \|\bbT^{1/2}Y_n\|_\bfH^2\leq C R\trace_\bfH \bbT \qquad \text{for all }n\geq 1.
\]
In particular,
\[
  \Expectation \|Y_n\|_\bfX^2 = 
\Expectation \|\bbC^{1/2}Y_n\|_\bfH^2\leq C R\trace_\bfH \bbC \qquad \text{for all }n\geq 1.
\]
\end{lemma}

\begin{proof}
Let $X_n$ be a standard normal random variable in $\Xi_n$; then $Y_n \stackrel{\mathrm{dist}}= \sqrt R\, X_n/\|X_n\|_\bfH$, and therefore
\[
\Expectation \|\bbT^{1/2}Y_n\|_\bfH^2 
=  R\;\Expectation \frac{\|\bbT^{1/2}X_n\|_\bfH^2}{\|X_n\|_\bfH^2}.
\]
The random variable $\|X_n\|^2_\bfH$ is $\chi^2(n)$-distributed,  and therefore  $\|X_n\|^2_\bfH/n$ concentrates onto $1$; it follows that there exists $r_n\to0$ such that 
\[
  \Expectation \|\bbT^{1/2}Y_n\|_\bfH^2 
  \leq  R(1+r_n) \Expectation\|\bbT^{1/2}X_n\|_\bfH^2
  \stackrel{\eqref{eq:finite-dimensional-trace}}= R(1+r_n) \trace_{\Xi_n}\bbT,
\]
where $\trace_{\Xi_n} \bbT = \trace_\bfH \calP_n \bbT \calP_n$ in terms of the orthogonal projection $\calP_n$ in $\bfH$ onto $\Xi_n$. Since $ \trace_{\Xi_n}\bbT \leq \trace_\bfH\bbT$, the result follows. 
\end{proof}

\bigskip
We conclude by giving the proofs of Lemmas~\ref{l:Sn} and~\ref{l:conv-mu-beta-n}.

\begin{proof}[Proof of Lemma~\ref{l:Sn}]
To prove part~\ref{l:Sn:part1}, remark that $\supp \zeta_{n,\beta,e}$ is an $n$-dimensional sphere of radius $r=\sqrt{2e + n/\beta}$; using the co-area formula~\eqref{eq:pdelta-co-area}  we find
\begin{align*}
\sfZ_{\beta,n,e}&:= \int_{\Xi_n} \pdelta\pra*{\frac12\|\xi\|_\bfH^2 - \frac n{2\beta} - e}\dd \xi 
\ =\ \  r^{-1}  \!\!\!\!\!\underbrace{n\omega_n r^{n-1}}_{\text{area of sphere of radius $r$}},
\end{align*}
so that 
\begin{align*}
\log \sfZ_{\beta,n,e}&=
\log n\omega_n + \frac{n-2}2 \log \bra*{\frac n\beta + 2e }\\
&= \underbrace{\log n\omega_n + \frac{n-2}2 \log \frac n\beta}_{=:\, C(\beta,n)} \ +\  \frac{n-2}2 \log \bra*{1 + \frac {2\beta e}n },
\end{align*}
and the convergence~\eqref{conv:partition-function-zeta} follows from a Taylor development of the final term.

\medskip
To prove part~\ref{l:Sn:part2}, first note that the sequence $\sfZ_{\beta,n,e}^{-1}\zeta_{\beta,n,e}$ of probability measures is not tight on $\bfH$ (in fact its support tends to infinity in the $\bfH$-norm). However, applying the variance bound of Lemma~\ref{l:variance-bound-MCmeasure} we observe that 
\[
\frac1 {\sfZ_{\beta,n,e}} \int _{\Xi_n} \|\xi\|^2_\bfX \zeta_{\beta,n,e}(\rmd \xi)
=
\Expectation_{\nu_{n,R_n}} \|\bbC^{1/2}\xi\|_\bfH^2 
\leq C \frac1\beta \trace_\bfH\bbC,
\]
where we set $R_n := \beta^{-1} + 2e/n\to \beta^{-1}$. 
Therefore the sequence is tight when considered as probability measures on $\bfX$, and along a subsequence the sequence converges narrowly  to a probability measure $\zeta_\beta$ on~$\bfX$ (e.g.~\cite[Th.~2.3.4]{Bogachev18}). 
% Since each $\zeta_{\beta,n,e}$ is supported on the set $\Xi_n\subset\Xi$ and $\Xi$ is closed in $\bfX$, the limit also is supported on $\Xi$.

To characterize the limit, note that for any $h\in \Xi_k$ the equivalence-of-ensemble result of Lemma~\ref{l:equivalence-of-ensembles} implies that 
\[
\int_{\Xi_n} \ee^{i\ip{h}{\xi}_\bfH} \zeta_{\beta,n,e}(\rmd \xi)
\stackrel{n\to\infty}\longrightarrow \exp\bra*{-\frac1{2\beta} \|h\|_\bfH^2 }.
\]
Taking elements $h$ of the form
\begin{equation}
\label{eqdef:spaceHn}
h = h_k +  h^\perp \in \bra*{\bigcup_{k\geq 1 }\Xi_k} \oplus (\Xi^\bfH)^\perp,
\end{equation}
we note that $h_k=\calQ h$, and since $h^\perp$ vanishes on each $\Xi_n$ we find 
\[
\int_{\Xi_n} \ee^{i\ip{h}{\xi}_\bfH} \zeta_{\beta,n,e}(\rmd \xi)
= 
\int_{\Xi_n} \ee^{i\ip{h_k}{\xi}_\bfH} \zeta_{\beta,n,e}(\rmd \xi)
\stackrel{n\to\infty}\longrightarrow \exp\bra*{-\frac1{2\beta} \|\calQ h\|_\bfH^2 }.
\]
Since the space in~\eqref{eqdef:spaceHn} is a subspace of $\bfV$ and is dense in $\bfH$, by part~\ref{l:measurable-extensions-adding-independent-variables} of Lemma~\ref{l:measurable-extensions} we recognize the limit $\zeta_\beta$ as the push-forward measure $\wh\calQ_\# \gamma_\beta$. This also implies that the whole sequence converges. 
\end{proof}

\begin{proof}[Proof of Lemma~\ref{l:conv-mu-beta-n}]
The assertion of Lemma~\ref{l:conv-mu-beta-n} follows from that of Lemma~\ref{l:mu-beta-n-e-limit} by integration: for any $\varphi\in \rmC_\rmb(\bfZ\ti\bfX)$ we have 
\begin{align*}
\MoveEqLeft\int\limits_{\bfZ\ti\bfX}\varphi(z,\eta) \, \mu_{\beta,n}(\rmd z \dd\eta)\ 
\stackrel{\eqref{eqdef:mu-beta-n-e}}=   \int\limits_{\bfZ\ti\bfW\ti\Xi^\bfX\ti\R}\varphi(z,w+\xi) \, \mu_{\beta,n,e}(\rmd z \dd w\dd\xi\dd e )\\
&\leftstackrel{n\to\infty}\longrightarrow \frac1{\sfZ_{\beta,\calE_0}}  
\int\limits_{\bfZ\ti\bfW\ti\Xi^\bfX\ti\R}\varphi(z,w+\xi) \,
\ee^{\beta e} \,\pdelta\pra[\big]{\Hzw(z,w) - \calE_0 + e} (\rmd z \dd w)  \;(\wh\bbQ_\#\gamma_\beta)(\rmd \xi)\dd e\\
&\leftstackrel{\eqref{eqdef:pdelta}}=\frac1{\sfZ_{\beta,\calE_0}}  
\int\limits_{\bfZ\ti\bfW\ti\Xi^\bfX}\varphi(z,w+\xi) \,
\ee^{-\beta \Hzw(z,w) +\beta \calE_0}   \;(\wh\bbQ_\#\gamma_\beta)(\rmd \xi)\dd z \dd w\\
&= \frac1{\wt \sfZ_{\beta}}  
\int\limits_{\bfZ\ti\bfW\ti\Xi^\bfX}\varphi(z,w+\xi) \,
\ee^{-\beta \Hzw(z,w) }  (\wh\bbQ_\#\gamma_\beta)(\rmd \xi) \dd z \dd w  \\
&\leftstackrel{\eqref{char:Htot-Hzw-Q}}=
\int_{\bfZ\ti\bfX} \varphi(z,\eta) \mu_\beta(\rmd z\dd \eta).\qedhere
\end{align*}
\end{proof}

\section{Dilations}
\label{sec:Appendix-Dilations}

The notion of dilations and compressions was introduced in Section~\ref{de:ContractDilat}. 
We now recall the existence of dilations.

\begin{theorem}[{\cite[Thm.\,8.1]{SznFoi70HAOH}}] 
  \label{th:Dilation}
  Every strongly continuous contraction semigroup $(C_t)_{t\geq 0}$ on a Hilbert space $\bfG$ admits a
  dilation $(S_t)_{t\in \R}$ on a Hilbert space $\bfH\supset \bfG$.

  Moreover, there is a minimal dilation in the sense that
  \[
    \mafo{span}\Big( \bigcup\nolimits_{t\in \R} S_t \bfG \Big)\ \text{ is dense in \ } \bfH.
  \]
  Any two minimal dilations $\big( \bfH, (S_t)_{t\in \R}\big)$ and $\big( \wt\bfH, (\wt S_t)_{t\in \R}\big)$
  are isomorphic in the sense that there exists a Hilbert-space isomorphism $\Psi\colon\bfH \to \wt\bfH$ such
  that $\wt S_t = \Psi S_t \Psi^{-1}$ for all $t \in \R$.
\end{theorem}

It turns out that even for finite-dimensional $\bfG$ the dilation space $\bfH$
will be infinite-dimensional as soon as $(C_t)_{t\geq 0}$ is not unitary, i.e.,
as soon as $C_t$ is a non-trivial contraction.  Moreover, the subspace $\bfG$
cannot be contained in the domain of the generator of the unitary group
$(S_t)_{t\in \R}$, as we now show.

\begin{RmRemark}[Domains of the generators]
  \label{re:Dom.JB.bfW}
  The unitary group has a densely-defined generator
  $\bbJ\colon\mafo{dom}(\bbJ)\subset \bfH \to \bfH$ which is skew-adjoint,
  i.e., $\bbJ^*=-\bbJ$. For $h_0\in \mafo{dom}(\bbJ)$ the function
  $h(t)=S_th_0=\ee^{t\bbJ }{h_0}$ is differentiable and satisfies
  $\dot h(t)= \bbJ h(t)$.

  The {contraction} semigroup can be written as $C_t= \ee^{-t\bbD}$ for a
  generator $\bbD$ satisfying $\bbD+\bbD^*\geq 0$. If $C_t$ is not unitary,
  then $\bbD$ is not skew-adjoint and there exists a $g\in \bfG$ such that
  $\gamma:=\ip{g}{\bbD{g}}_\bfG>0$. We set $p(t)=\|\ee^{-t\bbD}g\|^2$ and find
  $\lim_{t\to 0^+}\dot p(t)= -2\gamma$. Assuming $g\in \mafo{dom}(\bbJ)$, the
  dilation property gives
  $\dot p(0)=\frac{\rmd}{\rmd t} \|\bbP \ee^{t\bbJ}g\|^2 \big|_{t=0}= 2
  \ip{\bbP g}{\bbP \bbJ g }_\bfH= 2 \ip{ g}{ \bbJ g }_\bfH=0$.  This is a
  contradiction, and we conclude $g\not\in\mafo{dom}(\bbJ)$.

  Hence, $\bbJ$ is an unbounded operator, and $\bfH$ is
  infinite-dimensional. If $\ip{g}{\bbD g}_\bfG >0$ for every $g\not=0$, then
  we even have $\mafo{dom}(\bbJ)\cap \bfG = \{0\}$.
\end{RmRemark}

In the context of this paper, it is advantageous to generalize the notion of dilation slightly.

\begin{definition}[Generalized dilation] 
  \label{de:GenerDilat}
  A triple $\big(\bfH,\Phi,(S_t)_{t\in \R}\big)$ is called a \emph{generalized dilation} of the contraction
  semigroup $\big(\bfW,(\wh C_t)_{t\geq 0}\big)$ if $(S_t)_{t\in \R} $ is a strongly continuous unitary group on $\bfH$ 
  and the linear map $\Phi\colon\bfW\to \bfH$ satisfies
  \[
    \Phi^*\Phi= \mafo{id}_\bfW \quad \text{and} \quad 
    \Phi^*S_t \Phi = \wh C_t \in \mafo{Lin}(\bfW;\bfW) \text{ \ for all } t\geq 0. 
  \]
\end{definition}

Note that from this definition we see that $\bbP=\Phi\Phi^*\colon\bfH\to \bfH$ is an orthogonal projection
onto $\bfG= \Phi \bfW \subset \bfH$. In particular, setting $C_t= \Phi \wh C_t \Phi^*|_\bfG $ the unitary
group $\big(\bfH,(S_t)_{t\in \R}\big)$ is a dilation of $\big(\bfG, (C_t)_{t\geq 0}\big)$ in the sense of
Definition~\ref{de:ContractDilat}. However, since $\Phi\colon\bfW\to \bfG$ is an isometry with inverse
$\Phi^*|_\bfG$, the two contraction semigroups $ (C_t)_{t\geq 0}$ and $(\wh C_t)_{t\geq 0}$ are isomorphic. As
a consequence, all minimal dilations are isomorphic to any minimal generalized dilation, which is defined via
the minimality condition that $\bigcup_{t\in \R} \Phi \big(C_t\bfW\big)$ is dense in $\bfH$.

\subsection{An explicit construction of a generalized dilation}
\label{su:ExplicGenDilat}

Here we repeat the explicit construction of a generalized dilation from~\cite[Thm.\,2.1]{LaxPhi89ST}
or~\cite[Sec.\,1.2.3]{KumSch83SUD} and give the full proofs because of its importance of compressions and dilations for our analysis.

We now write the contraction semigroup $\big(\bfW,(C_t)_{t\geq 0}\big)$ in the form $C_t = \ee^{-t
  \bbD}$. Moreover, we follow our assumptions in Section~\ref{s:micro} and assume
\begin{equation}
  \label{eq:Dilat.Assum}
  \mafo{dim}(\bfW)<\infty \quad \text{and} \quad \exists\,\delta>0:\ \|
    \ee^{-t\bbD}\| \leq \ee^{-\delta t} \text{ for all } t \geq 0.  
\end{equation}
To construct a generalized dilation we  define
\[
  \ol\bfH:= \rmL^2(\R;\bfW) \text{ with }  \|f\|_{{\ol\bfH}}^2 =\int_\R\|f(y)\|_\bfW^2
  \dd y 
\quad \text{and} \quad \bfG:=\bigset{\Phi v \in \ol\bfH}{v\in \bfW},
\]
where the bounded linear operator $\Phi\colon\bfW\to \bfG \subset \ol\bfH$ is given by
\begin{equation}
  \label{eq:def.Phi}
  (\Phi w)(y) := \Sigma_\bbD\, \ee^{y \bbD}w \,\bm1_{{]{-}\infty,0]}}(y), \quad
  \text{with } \Sigma_\bbD:= (\bbD{+}\bbD^*)^{1/2},
\end{equation}
with adjoint operator $ \Phi^*\colon  \ol\bfH \to \bfW$,
$\Phi^* f = \int_{-\infty}^0 \ee^{y\bbD^*}\,\Sigma_\bbD\, f(s) \dd y$.

In contrast to~\cite{LaxPhi89ST,KumSch83SUD}, we included the factor $\Sigma_\bbD$ into the mapping $\Phi$ to obtain a
isometry in the given norm. For the convenience of the reader we provide the full proof that
$\big( \ol\bfH,\Phi,(\bbS_t)_{t\in \R}\big)$ is a generalized dilation for
$\big(\bfW,(\ee^{-t\bbD})_{t\geq 0}\big)$.

\begin{proposition}[Explicit dilation]
  \label{pr:ExplicitDilation}
  Assume that $\big(\bfW,(\ee^{-t\bbD})_{t\geq 0}\big)$ satisfies~\eqref{eq:Dilat.Assum}, define $ \ol\bfH$,
  $\bfG$, and $\Phi\colon\bfW\to \bfG\subset \bfG$ as above, and set
  $\bbP =\Phi\Phi^*\colon \ol\bfH\to \bfG\subset \ol\bfH$.
\begin{enumerate}
  \item \label{pr:ExplicitDilation:i:projection} 
  We have \ $ \Phi^*\Phi=\bbI_\bfW$, \quad $\bbP=\bbP^2 = \bbP^* $, \quad $\mafo{range}\,
  \bbP=\bfG$. \\
  In particular, $\Phi\colon\bfW\to \ol\bfH$ is an isometry and $\Phi\colon\bfW \to \bfG$ is unitary with
  inverse $\Phi^*|_\bfG$. Moreover, $\bbP$ is the orthogonal projector from $ \ol\bfH$ onto $\bfG$.

  \item\label{pr:ExplicitDilation:i:dilation}
  Defining the strongly continuous unitary shift group $(\bbS_t)_{t\in \R}$ on $ \ol\bfH=\rmL^2(\R;\bfW)$
  via $\bbS_tf= f(\cdot -t)$ the semigroup $\big(\bfW,(\ee^{-t\bbD})_{t\geq 0}\big)$ has the generalized
  dilation $\big( \ol\bfH,\Phi,(\bbS_t)_{t\in \R} \big)$. Moreover, for all $t \in \R$ and
  $g\in \bfG=\Phi\bfW$ we have the relation
  \begin{equation}
    \label{eq:Dilat.5.3}
    \bbP \, \bbS_t g = \bfR(t) g, \quad  \text{where } \bfR(t)
    = \left\{ \ba{cl}   \Phi \,\ee^{-t\bbD}\,\Phi^* & \text{for } t\geq 0,\\[-0.2em] 
      \Phi \,\ee^{\,t\,\bbD^*} \, \Phi^* & \text{for } t\leq 0. \ea\right. 
  \end{equation}
\end{enumerate}
\end{proposition}

\begin{proof} \emph{Part~\ref{pr:ExplicitDilation:i:projection}}. The property of an isometry is based on the standard relation
\begin{align}
  \label{eq:bbD.relation}
\int_{-\infty}^0 \!\! \big(\Sigma_\bbD \ee^{y\bbD} \big)^*\,
  \Sigma_\bbD\ee^{y\bbD} \dd y =
\int_{-\infty}^0 \!\!\ee^{y\bbD^*}(\bbD^* {+}\bbD)\,\ee^{y\bbD} \dd y =
\int_{-\infty} ^0 \!\frac\rmd{\rmd y} \big( \ee^{y\bbD^*}\ee^{y\bbD} \big) \dd y
= \bbI_\bfW. 
\end{align}
To establish $\Phi^*\Phi= \bbI_\bfW$ we simply calculate, for $w_1,w_2\in
\bfW$,  
\begin{align*}
( \Phi w_1| \Phi w_2)_{\bfH}&= \int_{-\infty}^0 \!\! 
\big( \Sigma_\bbD \ee^{y\bbD}w_1\big| \Sigma_\bbD \ee^{y\bbD}w_2\big)_\bfW \dd y 
= \int_{-\infty}^0 \!\!\big( \ee^{y \bbD^*}\Sigma_\bbD^2 \ee^{y \bbD} w_1
\big| w_2\big)_\bfW \dd y  \overset{\text{\eqref{eq:bbD.relation}}} 
= ( w_1 | w_2)_\bfW .
\end{align*}

With this and $\bbP=\Phi\Phi^*$ the identities $\bbP^2= \bbP$ and $\bbP= \bbP^*$ follow immediately, i.e.,
$\bbP$ is an orthogonal projector. Moreover, $\bbP f = \Phi(\Phi^*f)$ implies
$\mafo{range}\,\bbP \subset \mafo{range}\,\Phi = \bfG$, whereas $\bbP \Phi v= \Phi\Phi^*\Phi w= \Phi w$ shows
$\mafo{range}\,\bbP \supset \mafo{range}\,\Phi = \bfG$.\medskip

\emph{Part~\ref{pr:ExplicitDilation:i:dilation}}. Now we consider $t\geq 0$ and $ v\in \bfW$ and obtain
\begin{align*}
 \Phi^*\big(\bbS_t \Phi w\big)  &
=  \int_{-\infty}^0 \!\!  \ee^{s \bbD^*} \, \Sigma_\bbD  
      \big(\Phi w \big)(s{-}t) \dd s 
 =  \int_{-\infty}^0 \!\!  \ee^{s \bbD^*} \Sigma^2_\bbD\, 
\ee^{(s-t)\bbD} w \,\bm1_{]{-}\infty,0]}(s{-}t) \,\dd s
\\
& 
 =  \int_{-\infty}^0 \!\! \ee^{s \bbD^*} \Sigma_\bbD^2 \,
 \ee^{s\bbD}  \dd s \ \ee^{-t\bbD}w
 \overset{\text{\eqref{eq:bbD.relation}}} = \ee^{-t\bbD}w.
\end{align*}
Hence, $\big(\ol\bfH,\Phi,(\bbS_t)_{t\in \R} \big)$ is a generalized dilation in the sense of
Definition~\ref{de:GenerDilat}.

For $t \leq 0$ we have $\bm1_{]{-}\infty,0]}(s{-}t)=0$ for $s>t$ and obtain
\begin{align*}
\Phi^*\big(\bbS_t  \Phi w\big)&
= \int_{-\infty}^t \!\!  \ee^{s \bbD^*} \Sigma^2_\bbD 
\ee^{(s-t)\bbD} w \,\dd s  =  \int_{-\infty}^0 \!\! 
\ee^{t\bbD^*}\ee^{r\bbD^*} \Sigma_\bbD^2\, \ee^{r\bbD} w \,\dd s  
 =\ee^{t\bbD^*} w. 
\end{align*}
Applying the operator $\Phi$ to the two relations for $\Phi^*\big(\bbS_t \Phi w\big) $ (for $t\geq 0$ and $t\leq 0$, respectively)
 and replacing $\Phi w$ by $g=\Phi w$ such that $w= \Phi^*g$, we find the final
relation~\eqref{eq:Dilat.5.3}.
\end{proof}

The major advantage of this explicit generalized dilation is that it
introduces, in a very natural way, the time shift $\bbS_t$
on the Hilbert space $\ol\bfH = \rmL^2(\R;\bfW)$ which will take over the place of the linear evolution
$\ee^{t \JHB}$. The specification to the time shift is advantageous as the time shift is also the
natural variable in the random processes $(Y_t)_{t\in \R}$.

% {Alex added {\bfseries\large minimality} on 21 Feb 2023}
\begin{RmRemark}[Minimality]
  The above construction of a generalized dilation also allows us to show that the dilation is minimal. For
  this we have to show that
  \[
    \mathfrak G:= \mafo{span}\bigset{ \bbS_t g}{ t\in \R,\ g\in \bfG} 
    =  \mafo{span}\bigset{ \bbS_t (\Phi w)}{ t \in \R , \ w\in \bfW }  
  \] 
  is dense in $\ol\bfH = \rmL^2(\R;\bfW)$. For $h\in {]0,1[}$ and $w\in \bfW$, we define the functions
  \[
    \Psi_{h,w} := \bbS_h\big(\Phi w) - \Phi w: \ y \mapsto  \begin{cases} 
      0             & \text{for }y> h,\\ 
      \Sigma_\bbD\ee^{(y-h)\bbD}w& \text{for }y\in {]0,h]} , \\
      \Sigma_\bbD \ee^{y\bbD} \big(\ee^{-h\bbD}{-}\bbI\big) w& \text{for } y\leq 0.
    \end{cases}
  \]
  Clearly, $\Psi_{h,w} \in \mathfrak G$.  Using $\mafo{dim}(\bfW)<\infty $ we have
  $\|\ee^{-h\bbD}{-}\bbI\|\leq C h$ for $h\in [0,1]$ and obtain
  \[
    \big\| \Psi_{h,w} - w\bm1_{{[0,h[}} \big\|_{\rmL^2(\R;\bfW)} \leq C_* h \| w\|_\bfW,
  \]
  with a constant $C_*$ depending only on $\bbD$.

  To show that $f\in \rmL^2(\R;\bfW)$ can be approximated by linear combinations
  of $\bbS_t \Psi_{h,w}$, we fix $\eps>0$ and find a piecewise constant $\ol
  f_N\in \rmL^2(\R;\bfW)$ such that
  \[
    \big\| f - \ol f_N\big\|_{\rmL^2(\R;\bfW)} < \eps/2 \quad \text{and} \quad 
    \ol f_N= \sum_{j=-N^2}^{N^2} w^N_j \bm1_{{[j/N,(j+1)/N[}} \ \text{ with }N\in \N
    \text{ and } w^N_j\in \bfW.
  \]
  Here $N$ can be made as large as desired, by refining the partition. We now approximate $\ol f$ by
  $\wh f_N=\sum_{-N^2}^{N^2} \bbS_{j/N}\Psi_{1/N,w^N_j}$, which clearly lies in $ \mathfrak G$, and obtain
  \begin{align*}
    & \textstyle \big\| \ol f_N - \wh f_N\|_{\rmL^2(\R;\bfW)}^2 = \sum_{j=-N^2}^{N^2} \big\|
      w_j^N \bm1_{{[j/N,(j+1)/N[}} - \ol\bbS_{j/n}\Psi_{1/N,w_j^N}
      \|_{\rmL^2(\R;\bfW)}^2 \\
    & \textstyle \leq \ \sum_{j=-N^2}^{N^2} C_*^2 (1/N)^2\|w^N_j\|_\bfW^2 \  = \  \frac{C_*}N \,
      \big\| \ol f_N\big\|^2_{\rmL^2(\R;\bfW)} \ \leq \ \frac{C_*}N\,\big( \, \|
      f\|^2_{\rmL^2(\R;\bfW)} {+}\frac\eps2\, \big)^2 .   
  \end{align*}
  Choosing $N\in \N$ sufficiently big, the last term is less than $\eps^2/4$, and thus we conclude
  $\| f - \wh f_N\|_{\rmL^2(\R;\bfW)} < \eps$ as desired.  Hence, $\mathfrak G$ is dense in
  $\ol\bfH=\rmL^2(\R;\bfW)$.
\end{RmRemark}

\subsection{Continuous extension of the group $(\ee^{t\JHB})_{t\in \R}
  $ from $\bfH$ to $\bfX$}
\label{ss:extens-chosen-contin}

In Assumption \ref{ass:HB-evol-op-ct-on-X} we assumed that the unitary group
$\big(\bfH,(\ee^{t\JHB})_{t\in \R}\big)$ describing the heat bath can be
extended to a strongly continuous group on the larger space $\bfX$, on which we
are able to realize the Gaussian measure $\gamma_\beta$ for the
finite-temperature heat bath.  We will show in Proposition~\ref{pr:Ass3.2holds}
that this assumption is already a consequence of the dilation property from
Assumption~\ref{ass:dilation}.

As a first step we construct the Gaussian measure $\ol\gamma_\beta$ on an
extension $\ol\bfX$ of the special space $\ol\bfH=\rmL^2(\R;\bfW)$. This leads
to a classical white noise process, but now taking values in the
finite-dimensional Hilbert space $\bfW$. While translations $\bbS_t$ can easily be
defined for the white noise process, we have to be more careful and construct a
Hilbert space $\ol\bfX$ on which we can extend the unitary group
$\big(\ol\bfH,(\bbS_t)_{t\in \R}\big)$ to a strongly continuous group
$\big(\ol\bfX,(\bbS^\bfX_t)_{t\in \R}\big)$, which will be no longer
unitary. If $\bbN\colon\ol\bfH\to \ol\bfX$ is the embedding operator, then the
extension property means 
\begin{equation}
  \label{eq:Dil.ExtensProp}
\forall\: t\in \R\ \ \forall\: h \in \ol\bfH:\quad   \bbS^\bfX_t \bbN h  = \bbN
\bbS_t h.  
\end{equation}
As explained in Section \ref{ss:Microscopic-positive-temperature} we further
need that $\bbC=\bbN \bbN^*\colon\ol\bfX \to \ol\bfX$ is trace class. 

We give an explicit construction following the ideas in \cite{Reyb06OCS}. We
define the self-adjoint operator
$\bbA\colon \Dom(\bbA) \subset \ol\bfH \to \ol\bfH$ via
\begin{align*}
&\Dom(\bbA)= \bigset{ u \in \rmH^2(\R;\bfW) }{t \mapsto (1{+}t^2)u(t) \in
    \rmL^2(\R;\bfW) },
\\
& (\bbA u)(t) = -\frac12 u''(t) + \frac{t^2}2\, u(t).
\end{align*}
Here $\rmH^2(\R;\bfW)$ is the Sobolev space of measurable $\rmL^2$-functions
$u\colon\R\to\bfW$ with two weak derivatives in $\rmL^2$.  This is a $\bfW$-valued
version of the quantum harmonic oscillator, and hence is self-adjoint with the
discrete spectrum $\bigset{n{-}1/2 }{n\in \N} \subset {[1/2,\infty)}$, where
all eigenvalues have multiplicity $\dim(\bfW) < \infty$. Consequently $\bbA$ is
an invertible operator, and its inverse $\bbA^{-1}$ is a Hilbert-Schmidt
operator. This implies that $\bbA$ satisfies the assumptions for the
construction in Section~\ref{app:ss:Gelfand-construction}, i.e.\ we can define
\[
\ol\bfV = \Dom(\bbA)= \rmH^2(\R;\bfW)\cap \rmL^2_{(2)}(\R;\bfW)  \quad
\text{and} \quad \ol\bfX=\ol\bfV^* = \rmH^{-2}(\R;\bfW)+\rmL^2_{(-2)}(\R;\bfW),
\]
where the subscripts ``$_{(m)}$'' denote the $\rmL^2$ space with measure
$(1{+}t^2)^m\rmd t$. In $\ol\bfV$ we choose the graph norm $\|u\|_{\ol\bfV}= \|
\bbA u \|_{\ol\bfH}$.  

Concerning the shift group $(\ol\bfH,\bbS_t)_{t\in \R}$ we observe that it can be
restricted to $\ol\bfV$, since $u \in \ol\bfV$ implies $\bbS_t u \in \ol\bfV$
for all $t$. Setting $\bbS^\bfV_t = \bbS_t|_{\ol\bfV}:\ol\bfV\to \ol\bfV$  we
have 
\begin{align*}
 \| \bbS^\bfV_t u\|_{\ol\bfV} &\leq C\big( \|({\bbS}_t u)''\|_{\rmL^2} +
 \|(1{+}s^2)({\bbS}_tu)(s)\|_{\rmL^2} \big) 
\\
&\leq 2C\big( \|u''\|_{\rmL^2} +
 (1{+}t^2)  \|(1{+}r^2)u(r)\|_{\rmL^2} \big) \leq 
(1{+}t^2) C\| u\|_{\ol\bfV}. 
\end{align*}
Hence, $\big(\ol\bfV,(\bbS^\bfV_t)_{t\in \R}\big)$ forms a strongly continuous group of shift
operators. By dualization (with respect to the duality in $\ol\bfH$) we can define the extension 
\[
\bbS^\bfX_t\colon \ol\bfX \to \ol\bfX;\ \ \bbS^\bfX_t = \big( \bbS^\bfV_{-t}\big)^*
\]
which provides the strongly continuous group $\big(\ol\bfX,(\bbS^\bfX_t)_{t\in
  \R}\big)$ satisfying $\|\bbS^\bfX_t\|_{\ol\bfX}\leq C^2(1{+}t^2)$.\medskip

\begin{proposition}[Assumption \ref{ass:HB-evol-op-ct-on-X} holds] 
\label{pr:Ass3.2holds}
If the contraction semigroup $\big(\bfG,(\ee^{-t\bbD})_{t\geq 0}\big) $
satisfies \eqref{eq:Dilat.Assum} and
$\big(\bfH,(\ee^{t\JHB})_{t\in \R}\big) $ is an arbitrary minimal
dilation, then there exists a larger space $\bfX$ with embedding
$\bbN:\bfH\to \bfX$ and a strongly continuous group $(\extG_t)_{t\in \R}$
such
\begin{align*}
&\bbN\bbN^*\colon\bfX\to \bfX \quad \text{is trace class}
\\
& \forall \, t\in \R\ \forall \, h \in \bfH:\quad 
\extG_t \bbN h = \bbN \ee^{t\JHB} h.
\end{align*}
\end{proposition}
\begin{proof} We use the explicit construction of the generalized dilation
  $\big( \ol\bfH,(\bbS_t)_{t\in \R}\big)$ and of the explicit extension to
  $\big(\ol\bfX , (\bbS^\bfX_t)_{t\in \R}\big)$. 

By Theorem \ref{th:Dilation} the two dilations are unitarily equivalent by a
unitary map $U\colon\bfH \to \ol\bfH$. If $\ol\bbN$ is the embedding of $\ol\bfH$
into $\ol\bfX$, then we can define 
\[
\bfX := \ol \bfX, \quad \bbN:=\ol\bbN \,U, \quad \extG_t = \bbS^\bfX_t.
\]
Clearly, $\bbN\bbN^*= \ol\bbN U\,U^*\ol\bbN^* =\ol\bbN\,\ol\bbN^* $ is
trace class by the explicit construction. Moreover, 
\[
\extG_t\bbN h= \bbS^\bfX_t \ol\bbN Uh \overset{\text{(1)}}=  \ol\bbN\bbS_tUh 
= \ol\bbN\, U U^{-1}\bbS_tUh \overset{\text{(2)}}= \bbN\, \ee^{t\JHB} h
\]
holds for all $h\in \bfH$. In $\overset{\text{(1)}}=$ we used the explicit
extension of $\bbS^\bfX_t$ of $\bbS_t$ with embedding $\ol\bbN$, and in
$\overset{\text{(2)}}= $ we used the definition of $\bbN$ and that the two
dilations are equivalent via the unitary operator~$U$.  
\end{proof}
\medskip

The advantage of using the explicit generalized dilation involving
$\ol\bfH=\rmL^2(\R;\bfW)$ becomes even more apparent if we use the structure of
the mapping $\Phi\colon\bfW\to \bfG\subset \ol\bfH$ and
$\Phi^*\colon\ol\bfH\to \bfW$.  Combining $\Phi^*$ and the shift $\bbS_t$ we
obtain the mapping
\[
\calU: \ol\bfH \to \ol\bfH; \quad (\calU f)(t):= \Phi^* \big(\bbS_t f\big).
\]
This corresponds to the mapping $t \mapsto Y_t(\eta_0) := \wh\bbP \extG_t
\eta_0$ from Lemma \ref{l:props-Y}. We simply ``dropped'' the Hilbert-space
isomorphism $\Phi:\bfW \overset{1\text{-}1}\longleftrightarrow \bfG$. 

We now show that $\calU$ has a bounded linear extension $\wh\calU\colon\ol\bfX\to 
\ol\bfX$  that is consistent with $Y_t$ in Lemma \ref{l:props-Y}
  and is defined everywhere, i.e., we do not need the measurable extension
  $\wh\bbP$ defined $\gamma_\beta$ a.e.  
For this we observe that 
\[
\big(\calU f\big)(t)= \int_{-\infty}^0 
\Sigma_\bbD\,\ee^{s \bbD^*} f(s{-}t)\dd s. 
\]
Thus, $\calU$ is realized by convolution with the kernel $s \mapsto
\Sigma_\bbD\,\ee^{s \bbD^*} \bm1_{]-\infty,0]}$. In Fourier space the
convolution is realized by the multiplication with the Fourier multiplier 
\[
\R\ni \sigma \mapsto M(\sigma)^*:=\Sigma_\bbD \big( \bbD^*- \ii \sigma
\bbI\big)^{-1}
\]
which is smooth and bounded. As the operator $\bbA$ stays invariant (up to
changing constants) under Fourier transform $\mathfrak F$, we also see that
$\mathfrak F$ maps ${\bfV}$ into itself, and hence by duality also ${\bfX}$ into
${\bfX}$. Clearly, multiplying by $M$ defines a bounded linear operator in
${\bfV}$, because $M$ lies in $\rmC^2_\rmb \big(\R;\mafo{Lin}(\bfW)\big)$. By
duality the multiplication multiplying by $M^*$ defines a bounded linear
operator in ${\bfX}={\bfV}^*$. This shows that there is a bounded linear extension
$\wh\calU\colon{\bfX}\to {\bfX}$ of $\calU\colon{\bfH}\to {\bfH}$. 

\paragraph*{Declarations}
The authors have no relevant financial or non-financial interests to disclose.

\bibliographystyle{alpha_AMs}
\bibliography{alex_pub,bib_alex,bib-jz,bib-Mark-\jobname}
% \bibliography{alex_pub,bib_alex,bib-jz,ref}
% \bibliography{ref}

\end{document}